\documentclass[%
 reprint,
superscriptaddress,
 amsmath,amssymb,
 aps,
 prx,
]{revtex4-2}

\usepackage[utf8]{inputenc}
\usepackage[english]{babel}
\usepackage{CJKutf8}

\usepackage{graphicx}% Include figure files
\usepackage{dcolumn}% Align table columns on decimal point
\usepackage{bm}% bold math
\usepackage{physics}
\usepackage{amsthm}
\usepackage{xcolor}
\usepackage[percent]{overpic}
\usepackage[normalem]{ulem}
\usepackage{booktabs}
\usepackage{hyperref}% add hypertext capabilities
\hypersetup{colorlinks=true,linkcolor=blue,citecolor=blue,filecolor=blue,urlcolor=blue}

\newtheorem{theorem}{Theorem}

\newtheorem{proposition}{Proposition}

\newtheorem{assumption}{Assumption}
\newtheorem{corollary}{Corollary}

\newcommand{\R}{\mathbb{R}}
\newcommand{\E}{\mathbb{E}}
\newcommand{\Cov}{\operatorname{Cov}}
\newcommand{\cF}{\mathcal{F}}

\newcommand{\cG}{\mathcal{G}}

\begin{document}

% \preprint{APS/123-QED}

\title{
Repairability of Inexact Solvers in Recursive State Estimation with Machine Learning
}% Force line breaks with \\

\begin{CJK*}{UTF8}{gbsn}
\author{Yanjun Ji (冀彦君)}%
\email{y.ji@fz-juelich.de}
\affiliation{%
Jülich Supercomputing Centre, Forschungszentrum Jülich, 52425 Jülich, Germany
}%

\author{Dennis Willsch}
\email{d.willsch@fz-juelich.de}
\affiliation{%
Jülich Supercomputing Centre, Forschungszentrum Jülich, 52425 Jülich, Germany
}%
\affiliation{
FH Aachen University of Applied Sciences, 52066 Aachen, Germany
}

\author{Orkun \c{S}ensebat}
\affiliation{%
Jülich Supercomputing Centre, Forschungszentrum Jülich, 52425 Jülich, Germany
}%

\author{Priyanka Arkalgud Ganeshamurthy}
\affiliation{%
Institute for Automation of Complex Power Systems, RWTH Aachen University, Aachen, Germany
}%

\author{Zhi Pei}
\affiliation{Key Laboratory of Specialty Fiber Optics and Optical Access Networks, School of Communication and Information Science, Shanghai University, Shanghai 200444, China}

\author{M. Sahnawaz Alam}
\affiliation{%
Institute for Automation of Complex Power Systems, RWTH Aachen University, Aachen, Germany
}%

\author{Ivelina Stoyanova}%
\affiliation{%
Institute for Automation of Complex Power Systems, RWTH Aachen University, Aachen, Germany
}%

\author{Frank K. Wilhelm}
\affiliation{%
Institute for Quantum Computing Analytics (PGI-12), Forschungszentrum Jülich, 52425 Jülich, Germany
}%
\affiliation{Theoretical Physics, Saarland University, 66123 Saarbrücken, Germany}

\author{Bo Zhao}
\affiliation{
Department of Computer Science, Aalto University, 02150 Espoo, Finland
}

\author{Chao Wang (王潮)}
\affiliation{Key Laboratory of Specialty Fiber Optics and Optical Access Networks, School of Communication and Information Science, Shanghai University, Shanghai 200444, China}

\author{Kristel Michielsen}
\affiliation{%
Jülich Supercomputing Centre, Forschungszentrum Jülich, 52425 Jülich, Germany
}%
\affiliation{Faculty of Mathematics and Natural Sciences, University of Cologne, 50923 K\"oln, Germany}

\date{\today}% It is always \today, today,

\begin{abstract}

Recursive state estimation often executes approximate numerical solutions inside a feedback loop, where highly accurate local steps do not guarantee better overall results. For a fixed linear Kalman model, we characterize when a correction within a prescribed subspace and norm budget can meet a local admissibility tolerance, and how the defects actually executed affect the finite-horizon covariance response. Centering each defect on the exact gain for the implemented covariance separates current solve error from inherited gain drift. Expanding the exact residual–drift identity reveals opposing quartic contributions beyond the quadratic response: innovation-covariance inflation enters positively, while local-gain reoptimization enters subtractively. Under matched initialization, an absolute sixth-order remainder bound, uniform over bounded defect sequences at fixed horizon, gives sufficient conditions for quadratic under- or overprediction. Machine learning proposes bounded corrections, while a learner-independent residual certificate and verified fallback govern execution of classical and quantum candidates without changing the reference estimator. In a power-grid tolerance study, learned correction lowers the minimum conjugate-gradient iteration count for deployment without fallback relative to uncorrected solves under the same residual certificate. Gains reconstructed from a variational quantum linear solver and from an annealing-based binary encoding, with small-scale terminal measurements on superconducting hardware and sampling on a quantum annealer, are executed through the same interface. By linking local repairability to nonlinear error propagation, the framework evaluates approximate solvers and learned corrections through independent certification and finite-horizon response, providing a practical basis for studying hybrid quantum--classical computation under shared estimator-level criteria.

\end{abstract}

\maketitle
\end{CJK*}

% \tableofcontents

\section{Introduction}
\label{sec:intro}

Recursive state estimation tracks a dynamical system by alternating model-based prediction with measurement updates. It underlies power-grid dynamic state estimation \cite{zhao2019dse}, large-scale data assimilation \cite{BardsleyEtAl2013,Freitag2020}, adaptive-optics feedback \cite{poyneer2023lqg}, and feedback control of levitated mechanical systems \cite{setter2018realtime,magrini2021realtime}. In Kalman-type estimators, each update applies a gain obtained by solving a linear system involving the innovation covariance. When this gain is recomputed during operation, finite computational resources enter the inference loop directly. If the gain system is solved approximately, the resulting error does not remain local: the executed gain changes the posterior covariance, which in turn affects subsequent gain systems. The central question is therefore not only how accurately each gain is computed, but which computational errors can be repaired before execution and how the defect actually executed propagates through the subsequent recursion.

\begin{figure*}[t]
  \centering
  \includegraphics[width=\textwidth]{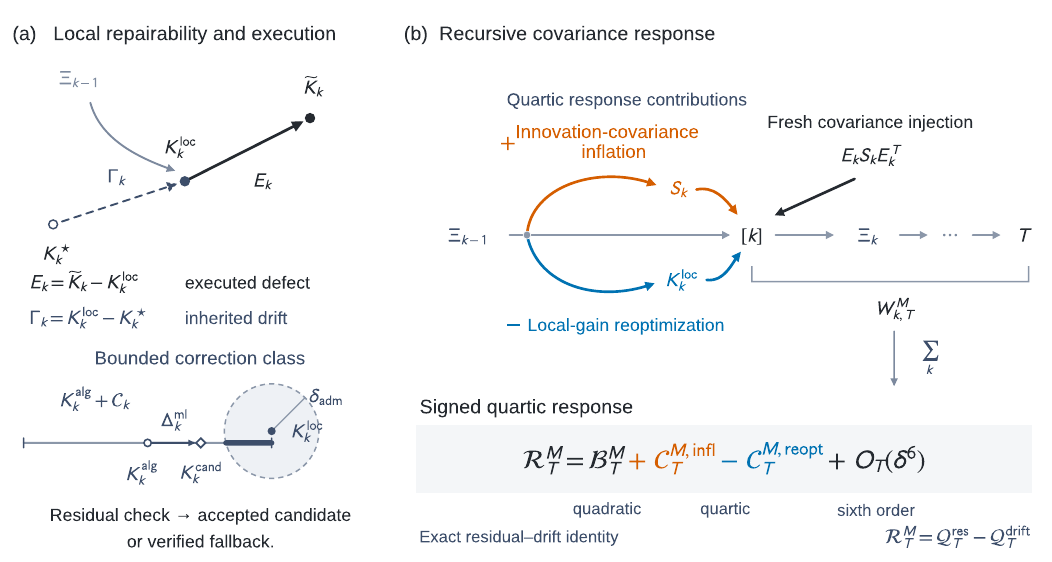}
  \caption{Schematic of local repairability, certified execution, and recursive response. (a) Gain coordinates separate the executed defect \(E_k=\widetilde K_k-K_k^{\rm loc}\) from inherited drift \(\Gamma_k=K_k^{\rm loc}-K_k^\star\). The inset shows a one-dimensional bounded correction class \(\mathcal C_k=\{\Delta\in\mathcal V_{\rm corr}:\|\Delta\|_F\le\delta_c\}\). The reachable segment \(K_k^{\rm alg}+\mathcal C_k\) intersects the admissible ball of radius \(\delta_{\rm adm}\) centered on \(K_k^{\rm loc}\). The intersection contains admissible gains, while the learned candidate \(K_k^{\rm cand}=K_k^{\rm alg}+\Delta_k^{\rm ml}\) (diamond) lies outside: correction existence does not guarantee learned attainment. Independent residual checking determines whether the candidate or a verified fallback is executed. (b) Executed defects inject \(E_kS_kE_k^{\mathsf T}\) into the covariance recursion; \(\Xi_k:=\widehat P_k-\widehat P_k^\star\) denotes the posterior covariance mismatch. The inherited mismatch \(\Xi_{k-1}\) generates opposing quartic contributions through innovation-covariance inflation (\(+\), orange) and local-gain reoptimization (\(-\), blue). The operators \(W_{k,T}^M\) combine fixed-reference propagation and task weighting over the remaining horizon. Summing the weighted residual-minus-drift contributions gives the exact finite-horizon response. The small-defect expansion has an absolute \(O_T(\delta^6)\) remainder, uniform over the prescribed bounded defect class at fixed horizon under matched initialization, with \(\delta=\max_{1\le j\le T}\|E_j\|_F\).}
  \label{fig:conceptual_response}
\end{figure*}

Computation inside recursive estimation has been studied from several perspectives. Stable Kalman implementations and filters designed under implementation uncertainty treat finite-precision and gain-perturbation effects \cite{verhaegen1986numerical,mahmoud2004resilient,deoliveira2005implementation}. Krylov methods have been applied directly to approximate Kalman computation \cite{BardsleyEtAl2013}, while large-scale and low-rank methods address related computational bottlenecks \cite{Freitag2020,LeProvost2022LowRankEnKF}; inexact Krylov theory analyzes approximate numerical operations within iterative solvers \cite{simoncini2003inexact}. Probabilistic numerical and computation-aware formulations instead incorporate incomplete computation into the inference uncertainty \cite{cockayne2019bayescg,pfortner2025compaware}.
Analytically, finite-horizon Kalman policy optimization gives exact performance-difference relations between gain sequences \cite{li2025policy}, Riccati perturbation theory characterizes covariance sensitivity \cite{sun1998sensitivity,sun1998perturbation,aalto2018spatial}, and goal-oriented error analysis relates numerical residuals to quantities of interest \cite{meidner2009goal,endtmayer2020twoside}. We build on these ingredients to separate the error of the current gain computation from gain drift inherited through the recursive covariance, to determine when the resulting locally recentered defect is repairable within a constrained correction class, and to characterize how covariance feedback modifies its finite-horizon response beyond leading quadratic order.

We work with a fixed linear model, matched initialization, and a deterministic, predeclared covariance schedule. The exact reference gain sequence is therefore measurement independent and can be computed offline. This controlled construction isolates computational approximation from model or covariance adaptation; it does not establish an online computational requirement for the reference estimator. State-dependent linearization or parameter updates \cite{fan2013ekfdse}, topology changes \cite{wang2018pmujacobian}, and data-dependent covariance adaptation \cite{akhlaghi2017adaptive} are the settings in which gains must be updated during operation. Extending the analysis to them requires relaxing the present assumptions.

We first derive an exact residual--drift representation of the finite-horizon task-weighted covariance response for a fixed model and reference estimator. Each computational defect is measured relative to the exact gain associated with the current implemented covariance, separating current solve error from gain drift induced by preceding covariance displacement. The resulting identity expresses the response through the residual of each implemented gain system and a nonnegative subtractive gain-drift contribution, which is the exact slack in the corresponding residual-based upper bound. On the model and objective class shared with Ref.~\cite{li2025policy}, this identity follows by algebraic recentering of the established finite-horizon performance-difference relation. At fixed current covariance, these coordinates isolate the solve error that additional iterations or a gain correction can modify from the gain drift inherited from earlier steps, and provide the starting point for the higher-order response analysis.

The distinction already matters in the simplest case. Consider, for example, a scalar random walk with unit observation coefficient and noisy measurements, with the implemented and reference filters starting from the same covariance. An inexact gain at the first update produces a larger posterior error variance than the exact reference gain would. The displaced covariance changes the exact gain associated with the implemented filter at the next update, so a gain difference can persist even when the current gain equation is solved exactly. This difference is inherited gain drift rather than unresolved current computation. At a fixed current covariance, by contrast, the gain-equation residual measures the local computational defect that additional solution or correction can reduce.

Beyond this exact identity, we derive the leading correction to the fixed-reference quadratic response arising from defect-induced covariance feedback. Under matched initialization, the covariance displacement is second order in the executed defects, so its feedback first contributes at quartic order. Two explicit nonnegative quartic functionals enter with opposite signs: innovation-covariance inflation contributes positively, whereas local-gain reoptimization contributes subtractively. Reoptimization here denotes the shift of the exact local gain as the covariance changes. The difference between these functionals gives the signed quartic correction, with an absolute sixth-order remainder bound uniform over the prescribed bounded defect class at fixed horizon. When the magnitude of the net quartic correction exceeds this bound, its sign determines whether the quadratic prediction underestimates or overestimates the exact response.

The computational defects that drive the finite-horizon response are also the quantities on which a local gain correction acts. At fixed current covariance, we establish an exact one-step repairability boundary for corrections restricted to a prescribed subspace and Frobenius-norm budget. Its necessary and sufficient condition determines whether the current computational defect can be brought within the prescribed local admissibility tolerance and distinguishes obstruction by unavailable correction directions from obstruction by insufficient correction budget. This boundary specifies what is achievable within the correction class, independently of how a candidate correction is constructed.

Machine learning has been used more broadly to learn or augment recursive estimators under model and statistical uncertainty \cite{revach2022kalmannet,mortada2025recursivekalmannet,Shlezinger2025AIAidedKF}. Learned numerical solvers also modify iterative updates while preserving convergence guarantees~\cite{hsieh2019learning} or accelerate linear solves through learned conjugate directions~\cite{kaneda2023deep}. Here the learned model has a narrower role: it uses numerical information from the approximate gain computation to propose a bounded correction within a prescribed correction subspace and norm budget. The target estimator and residual certificate are defined independently of the learner, and the corrected gain is evaluated by the same certificate as any other represented candidate. Geometric repairability, learned attainment, and residual acceptance are therefore distinct questions. Only gains that satisfy the execution criterion enter the recursion. The response analysis then quantifies the finite-horizon consequences of the resulting executed defects under the stated model assumptions.

Figure~\ref{fig:conceptual_response} summarizes this complete chain from local computation to recursive consequence. Figure~\ref{fig:conceptual_response}(a) separates repairability, learned correction, residual-based execution, and inherited gain drift, while Fig.~\ref{fig:conceptual_response}(b) shows how the resulting executed defects and covariance mismatch generate the exact finite-horizon response and its opposing quartic feedback mechanisms.

Quantum linear-system methods provide alternative mechanisms for generating approximate gain candidates \cite{Stoyanova2025Quantum,ganeshamurthy2024bridging,golestan2023quantum,feng2024noisy,shi2024quantum}. Coherent quantum linear-system algorithms encode solutions in quantum states \cite{harrow2009quantum,childs2017quantum}, while variational solvers \cite{bravo2023variational} and quadratic unconstrained binary optimization (QUBO) encodings for annealing architectures \cite{rogers2020floating} provide alternative approximate representations. In the present framework, such a method enters only after all required right-hand-side solutions have been reconstructed or decoded and assembled into an explicit gain whose represented residual can be recomputed. The resulting candidate is then subject to the same correction, certification, fallback, and finite-horizon response definitions as a classical candidate. The quantum routine is thus one component of an otherwise classical inference pipeline, consistent with the broader benchmarking principle that hybrid quantum--classical workflows should be assessed against explicit classical references \cite{ji2026quantumdeep}; no quantum-advantage claim is made. We demonstrate this interface at small scale with a variational quantum linear solver (VQLS), including terminal measurements on superconducting hardware, and with QUBO-encoded candidates sampled on a quantum annealer.

The principal contributions are threefold. First, we establish a finite-horizon response theory that separates current computational defect from inherited gain drift. The leading response is quadratic, while recursive covariance feedback first enters at quartic order through opposing innovation-covariance inflation and local-gain reoptimization, with an explicit sixth-order remainder that is uniform over bounded defect sequences at fixed horizon. Second, we characterize the exact local repairability boundary for a prescribed correction subspace and norm budget, separating subspace and budget obstructions and distinguishing correction existence from learned attainment and residual acceptance. Third, we define a learner independent execution interface based on residual certification and verified fallback, and quantify the resulting reduction in certified solver depth while assessing recursive response separately from local acceptance. Complete gains reconstructed from quantum solvers enter through the same interface, providing a concrete pathway for integrating quantum resources into certified recursive estimation.

The remainder of the paper is organized as follows.
Section~\ref{sec:inexa_gain_compu} defines the recursive-estimation and computational framework. Section~\ref{sec:finit_horiz_respon} develops the finite-horizon response theory, and Sec.~\ref{sec:ml_correc_and_local} establishes local repairability, machine-learning correction, and certification. Section~\ref{sec:results} presents the results and Sec.~\ref{sec:discu_and_conclu} concludes. Supporting derivations and numerical methods are given in the Appendixes.

\section{Inexact Gain Computation in Recursive State Estimation}
\label{sec:inexa_gain_compu}

\subsection{Reference and implemented recursions}
\label{subsec:refer_and_implem_recur}

We consider the finite-horizon linear-Gaussian model
\begin{equation}
\begin{aligned}
    x_k &= F x_{k-1} + B\bar u_{k-1} + B\xi_{k-1} + w_{k-1},\\
    z_k &= H x_k + v_k,
    \qquad k=1,\ldots,T,
    \label{eq:state_space_model}
\end{aligned}
\end{equation}
where \(x_k\in\R^n\) is the hidden state, \(z_k\in\R^m\) is the measurement, and \(F\in\R^{n\times n}\), \(H\in\R^{m\times n}\), and \(B\in\R^{n\times p}\) are the state-transition, observation, and input matrices, respectively. The known input is \(\bar u_{k-1}\in\R^p\), with \(n,m,p,T\in\mathbb N\). The input, process, and measurement perturbations are zero-mean Gaussian,
\[
\xi_{k-1}\sim\mathcal N(0,V_{k-1}),\;
w_{k-1}\sim\mathcal N(0,Q_{k-1}),\;
v_k\sim\mathcal N(0,R_k),
\]
with covariance matrices \(V_{k-1}\in\R^{p\times p}\), \(Q_{k-1}\in\R^{n\times n}\), and \(R_k\in\R^{m\times m}\) satisfying
\begin{equation}
V_{k-1}\succeq0,\quad
Q_{k-1}\succeq0,\quad
R_k\succeq\underline r I_m,\quad
\underline r>0.
\label{eq:noise_covariance_bounds}
\end{equation}
Here \(I_d\) denotes the \(d\times d\) identity matrix, and for symmetric matrices \(M\) and \(N\), \(M\succeq N\) means that \(M-N\) is positive semidefinite.
The initial state satisfies \(x_0\sim\mathcal N(\mu_0,P_0)\), with initial mean \(\mu_0\in\R^n\) and covariance \(P_0\in\R^{n\times n}\), \(P_0\succeq0\).
The initial state and noise variables are mutually independent, with each noise sequence independent across time. The input sequence \(\bar u_{0:T-1}\) and covariance schedule \((V_{0:T-1},Q_{0:T-1},R_{1:T})\) are deterministic and predeclared, although the covariances may vary with \(k\).

For any covariance matrix \(P\succeq0\) and gain
\(K\in\R^{n\times m}\), define the Joseph covariance map
\begin{equation}
\mathcal{J}_k(P,K)
:=
(I_n-KH)P(I_n-KH)^{\mathsf T}
+KR_kK^{\mathsf T},
\label{eq:exact_cov}
\end{equation}
where \((\cdot)^{\mathsf T}\) denotes matrix transpose.
The reference recursion is the correctly specified conditional-mean Kalman estimator. Its prediction and update are
\begin{equation}
\begin{aligned}
\widetilde x_k^\star
&= F\widehat x_{k-1}^\star+B\bar u_{k-1},
\\
\widetilde P_k^\star
&= F\widehat P_{k-1}^\star F^{\mathsf T}
+BV_{k-1}B^{\mathsf T}+Q_{k-1},
\\
S_k^\star
&= H\widetilde P_k^\star H^{\mathsf T}+R_k,
\\
K_k^\star
&= \widetilde P_k^\star H^{\mathsf T}(S_k^\star)^{-1},
\\
r_k^\star&=z_k-H\widetilde x_k^\star,\\
\widehat x_k^\star
&= \widetilde x_k^\star
+K_k^\star r_k^\star,
\\
\widehat P_k^\star
&= \mathcal{J}_k(\widetilde P_k^\star,K_k^\star).
\end{aligned}
\label{eq:reference_recursion}
\end{equation}
As noted in Sec.~\ref{sec:intro}, the reference covariance and gain sequences do not depend on the realized measurements in this setting.
A tilde denotes a prediction quantity, a hat denotes a posterior
quantity, and \(\star\) denotes the exact reference recursion. Thus \(r_k^\star\), \(S_k^\star\), and \(K_k^\star\) are the reference innovation, innovation covariance, and Kalman gain, respectively. Because \(S_k^\star\succeq R_k\succeq\underline r I_m\), \(K_k^\star\) is well defined. Under the stated linear-Gaussian assumptions and matched initialization, \(\widehat x_k^\star=\E[x_k\mid z_1,\ldots,z_k]\) is the conditional-mean estimate and \(\widehat P_k^\star=\Cov(x_k\mid z_1,\ldots,z_k)\) its posterior error covariance.
Table~\ref{tab:notation_kalman_filter} in Appendix~\ref{app:notation} summarizes the notation; Appendix~\ref{app:joseph_identities} collects the Joseph-update identities.

The implemented estimator propagates the same physical model but uses the gain actually executed at each step. Let \(\widetilde K_k\in\R^{n\times m}\) denote this executed gain; its construction is specified in Sec.~\ref{subsec:gain_hiera}. The implemented recursion is
\begin{equation}
\begin{aligned}
\widetilde x_k
&=F\widehat x_{k-1}+B\bar u_{k-1},
\\
\widetilde P_k
&=F\widehat P_{k-1}F^{\mathsf T}
+BV_{k-1}B^{\mathsf T}+Q_{k-1},
\\
S_k
&=H\widetilde P_kH^{\mathsf T}+R_k,
\\
r_k&=z_k-H\widetilde x_k,
\\
\widehat x_k
&=\widetilde x_k+\widetilde K_k r_k,
\\
\widehat P_k
&=\mathcal{J}_k(\widetilde P_k,\widetilde K_k).
\end{aligned}
\label{eq:executed_recursion}
\end{equation}
Here \((\widetilde x_k,\widetilde P_k)\) is the implemented prediction pair, \(S_k\) is the implemented innovation covariance, and \(r_k\) is the implemented innovation. The same executed gain is used in the state and Joseph covariance updates and therefore determines the covariance entering the next prediction.
The Joseph update preserves positive semidefiniteness for any executed gain. Hence \(P_0\succeq0\) implies
\[
\widetilde P_k,\widehat P_k\succeq0,
\qquad
S_k\succeq R_k\succeq\underline r I_m\succ0.
\]
Thus the implemented innovation system is nonsingular at every step.

\begin{assumption}[Matched initialization]
\label{ass:matched_initialization}
At the beginning of the analyzed horizon, the reference and implemented estimators have matched state estimates and covariances,
\[
\widehat x_0=\widehat x_0^\star=\mu_0,
\qquad
\widehat P_0=\widehat P_0^\star=P_0.
\]
\end{assumption}

Matched initialization isolates the finite-horizon response generated by the executed computational defects. Nonmatched initialization introduces an additional initial-response contribution and is outside the defect-only result considered here.

\subsection{Gain hierarchy and computational defects}
\label{subsec:gain_hiera}

The executed gain acts on the implemented covariance, which may already differ from the reference covariance. To separate current solve error from this inherited covariance effect, define the exact \emph{local} Kalman gain \(K_k^{\mathrm{loc}}\in\R^{n\times m}\) associated with \(\widetilde P_k\) by
\begin{equation}
    K_k^{\mathrm{loc}} := \widetilde{P}_k H^{\mathsf T} S_k^{-1}.
    \label{eq:local_gain}
\end{equation}
Here ``local" denotes the exact solution of the current implemented gain system with \(\widetilde P_k\) fixed. It is a mathematical comparator and need not be evaluated during deployment. The reference gain \(K_k^\star\) is instead associated with \(\widetilde P_k^\star\). Thus \(K_k^{\mathrm{loc}}=K_k^\star\) when the prediction covariances coincide; otherwise their difference represents gain drift inherited from earlier covariance mismatch.

At step \(k\), the approximate solver returns \(K_k^{\mathrm{alg}}\) and the learned corrector proposes a bounded correction \(\Delta_k^{\mathrm{ml}}\), giving
\[
K_k^{\mathrm{cand}}
:=
K_k^{\mathrm{alg}}+\Delta_k^{\mathrm{ml}}.
\]
Setting \(\Delta_k^{\mathrm{ml}}=0\) recovers the uncorrected case.
The candidate and, if needed, the fallback are evaluated under the same residual certificate. The gain actually propagated is denoted \(\widetilde K_k\).
The computational path is therefore
\[
K_k^{\mathrm{alg}}
\xrightarrow{\ +\Delta_k^{\mathrm{ml}}\ }
K_k^{\mathrm{cand}}
\xrightarrow{\ \text{certificate/fallback}\ }
\widetilde K_k.
\]

We distinguish
\begin{equation}
\begin{aligned}
A_k
&:=
K_k^{\mathrm{alg}}-K_k^{\mathrm{loc}},\\
E_k
&:=
\widetilde K_k-K_k^{\mathrm{loc}},
\\
\Gamma_k
&:=
K_k^{\mathrm{loc}}-K_k^\star.
\end{aligned}
\label{eq:gain_error_definitions}
\end{equation}
Thus \(A_k\) is the raw current-step defect, \(E_k\) the defect actually executed, and \(\Gamma_k\) the gain drift inherited from preceding covariance mismatch. Consequently, \(\widetilde K_k-K_k^\star=E_k+\Gamma_k\), and, when the corrected candidate is executed, \(E_k=A_k+\Delta_k^{\mathrm{ml}}.\) The correction acts on the current computational defect, whereas the recursive response depends on the executed defect and inherited drift.

\subsection{Solver residual interface and local admissibility}
\label{subsec:solver_residual_interface}

For any gain \(L\in\R^{n\times m}\), define the residual of the current implemented gain system by
\begin{equation}
\rho_k(L)
:=
LS_k-\widetilde P_kH^{\mathsf T}.
\label{eq:matrix-residual}
\end{equation}
Because \(K_k^{\rm loc}S_k = \widetilde P_kH^{\mathsf T},\) the residual satisfies \(\rho_k(L) = (L-K_k^{\rm loc})S_k.\)
Since \(S_k\succ0\), \(L-K_k^{\rm loc} = \rho_k(L)S_k^{-1}\), so the residual exactly determines the local gain defect for the current gain system.
For the raw and executed gains,
\[
\begin{aligned}
\rho_k^{\rm alg}
&:=\rho_k(K_k^{\rm alg})
=A_kS_k,
\\
\rho_k^{\rm exec}
&:=\rho_k(\widetilde K_k)
=E_kS_k.
\end{aligned}
\]

Fix a local gain-defect tolerance \(\delta_{\rm adm}>0\). A gain \(L\) is locally admissible at step \(k\) when
\[
\|L-K_k^{\rm loc}\|_F\le\delta_{\rm adm}.
\]
Let $\ell_k$ be a positive lower bound on the smallest singular value of $S_k$,
\begin{equation}
0<\ell_k\le\sigma_{\min}(S_k),
\label{eq:certified_innovation_lower_bound}
\end{equation}
with \(\ell_k=\underline r\) as a valid uniform choice. Then
\[
\|L-K_k^{\rm loc}\|_F
\le
\frac{\|\rho_k(L)\|_F}{\sigma_{\min}(S_k)}
\le
\frac{\|\rho_k(L)\|_F}{\ell_k},
\]
so
\begin{equation}
\|\rho_k(L)\|_F
\le
\ell_k\delta_{\mathrm{adm}}
\label{eq:accept_rule}
\end{equation}
is sufficient for local admissibility. The condition is not necessary: failure of Eq.~\eqref{eq:accept_rule} does not imply that \(L\) is inadmissible. For \(L=\widetilde K_k\), satisfaction of the criterion guarantees \(\|E_k\|_F\le\delta_{\rm adm}.\)

Numerical candidates are assembled from approximate solutions of the transposed gain systems
\begin{equation}
S_k x_k^{(j)}=[H\widetilde P_k]_{:,j},
\qquad j=1,\ldots,n,
\label{eq:column_gain_systems}
\end{equation}
whose exact solutions form the columns of \((K_k^{\rm loc})^{\mathsf T}\). Our primary classical realization applies zero-start conjugate gradient (CG) independently to these systems, using the iteration count \(t\) to vary the computational effort and approximation depth~\cite{hestenes1952methods,greenbaum1997iterative}. Figure~\ref{fig:local_solver_landscape} in the Appendix illustrates how the local gain defect varies with CG iteration count and innovation dimension. The returned vectors are assembled into \(K_k^{\rm alg}\), and its matrix residual is recomputed explicitly before correction and certification. Implementation details are given in Appendix~\ref{app:classi_solve_implem}.

The framework itself does not depend on CG: any candidate-generation method can enter the same execution chain provided that its outputs can be assembled into an explicit represented gain \(K_k^{\rm alg}\) for which the residual in Eq.~\eqref{eq:matrix-residual} can be evaluated.

\section{Finite-Horizon Response to Inexact Gain Computation}
\label{sec:finit_horiz_respon}

\subsection{Covariance mismatch and residual-drift identity}
\label{subsec:covari_misma_and_task_respo}

Define the covariance mismatch
\begin{equation}
\Xi_k:=\widehat P_k-\widehat P_k^\star,
\qquad
k=0,\ldots,T,
\label{eq:cov_error_def}
\end{equation}
so that \(\Xi_0=0\) under Assumption~\ref{ass:matched_initialization}, and the reference and local closed-loop matrices
\[
\Phi_k^\star:=(I-K_k^\star H)F,
\qquad
\Phi_k^{\rm loc}:=(I-K_k^{\rm loc}H)F.
\]
Since
\[
\widetilde P_k-\widetilde P_k^\star
=
F\Xi_{k-1}F^{\mathsf T},
\quad
S_k-S_k^\star
=
(HF)\Xi_{k-1}(HF)^{\mathsf T},
\]
the inherited gain drift satisfies the exact identity
\begin{equation}
\Gamma_k
=
\Phi_k^\star\Xi_{k-1}(HF)^{\mathsf T}S_k^{-1}.
\label{eq:gain_drift_exact}
\end{equation}
A derivation is given in Appendix~\ref{app:covariance_mismatch_and_gain_drift}.

\begin{proposition}[Exact covariance-mismatch recursions] \label{prop:exact_covariance_recursions}
Under the standing model assumptions,
\begin{align}
\Xi_k
={}&\Phi_k^{\mathrm{loc}}\Xi_{k-1}
(\Phi_k^{\mathrm{loc}})^{\mathsf T}
+E_kS_kE_k^{\mathsf T}
+\Gamma_kS_k^\star\Gamma_k^{\mathsf T}
\nonumber\\
={}&\Phi_k^\star\Xi_{k-1}(\Phi_k^\star)^{\mathsf T}
+E_kS_kE_k^{\mathsf T}
-\Gamma_kS_k\Gamma_k^{\mathsf T}.
\label{eq:covariance_mismatch_recursions}
\end{align}
If $\Xi_0\succeq0$, then $\Xi_k\succeq0$ for \(k=0,\ldots,T\).
\end{proposition}

\begin{proof}
See Appendix~\ref{app:covariance_mismatch_and_gain_drift}.
\end{proof}

The first identity in Eq.~\eqref{eq:covariance_mismatch_recursions} is the local-centered form; the second is the reference-centered form.
The local-centered form makes positive semidefiniteness explicit, whereas the reference-centered form isolates the subtractive gain-drift term used in the finite-horizon accounting below.

For fixed task metrics \(M_k\in\R^{n\times n}\), \(M_k\succeq0\), define \(\|y\|_{M_k}^2:=y^{\mathsf T}M_ky\) and the finite-horizon task response
\begin{equation}
\mathcal R_T^M := \sum_{k=1}^{T}\tr(M_k\Xi_k).
\label{eq:task_risk_definition}
\end{equation}
Thus \(\mathcal R_T^M\ge0\). Under Proposition~\ref{prop:statistical_identification}, it equals the conditional excess estimation-error risk.

For \(j\le t\), define the reference propagator
\begin{equation}
\Psi_{t,j}^\star
:=
\begin{cases}
I_n, & t=j,\\
\Phi_t^\star\Phi_{t-1}^\star\cdots\Phi_{j+1}^\star, & t>j,
\end{cases}
\label{eq:reference_propagator}
\end{equation}
and the finite-horizon task-response operator
\begin{equation}
W_{j,T}^M
:=
\sum_{t=j}^{T}
(\Psi_{t,j}^\star)^{\mathsf T}
M_t
\Psi_{t,j}^\star
\succeq0.
\label{eq:task_response_operator}
\end{equation}
A covariance injection \(X_j\succeq0\) at step \(j\) therefore contributes
\(\tr(W_{j,T}^M X_j)\) to the accumulated reference-propagated response over \(j{:}T\).

Unrolling the reference-centered covariance recursion gives the exact decomposition
\begin{equation}
\mathcal R_T^M
=
\mathcal Q_T^{\rm res}
-
\mathcal Q_T^{\rm drift},
\label{eq:exact_residual_risk_accounting}
\end{equation}
where
\begin{align}
\mathcal Q_T^{\rm res}
&:=
\sum_{j=1}^{T}\tr\!\left[
W_{j,T}^M E_jS_jE_j^{\mathsf T}\right] \notag\\
&{}\;=
\sum_{j=1}^{T}
\operatorname{tr}\!\left[
W_{j,T}^M
\rho_j^{\rm exec}S_j^{-1}
(\rho_j^{\rm exec})^{\mathsf T}
\right],
\label{eq:residual_risk_contribution}
\\
\mathcal Q_T^{\rm drift}
&:=
\sum_{j=1}^{T}
\operatorname{tr}\!\left[
W_{j,T}^M
\Gamma_jS_j\Gamma_j^{\mathsf T}
\right]
\ge0.
\label{eq:gain_drift_correction}
\end{align}
Consequently,
\begin{equation}
0\le\mathcal R_T^M\le\mathcal Q_T^{\rm res},
\qquad
\mathcal Q_T^{\rm res}-\mathcal R_T^M
=
\mathcal Q_T^{\rm drift}.
\label{eq:residual_risk_slack}
\end{equation}
Thus the residual contribution is an exact upper bound on the finite-horizon response, with slack given by the recursively induced gain-drift contribution. The derivation is given in Appendix~\ref{app:exact_residual_accounting}. This recentered identity supplies the accounting coordinate; the additional result below is the explicit expansion of the endogenous matrices \(S_k\) and \(\Gamma_k\) generated by covariance feedback.

\subsection{Signed-quartic finite-horizon response}
\label{subsec:quartic_finit_horiz_respon}

Fix a finite horizon $T$ and define
\begin{equation}
\delta:=\max_{1\le k\le T}\|E_k\|_F.
\label{eq:finite_horizon_defect_amplitude}
\end{equation}
The fixed-reference quadratic form associated with the executed local defects is
\begin{equation}
\mathcal B_T^M(E)
:=
\sum_{k=1}^{T}
\tr\!\left[
W_{k,T}^M E_kS_k^\star E_k^{\mathsf T}
\right].
\label{eq:fixed_reference_quadratic_form}
\end{equation}
To isolate the leading covariance response, define
\begin{equation}
\Omega_0:=0,
\;
\Omega_k
:=
\Phi_k^\star\Omega_{k-1}(\Phi_k^\star)^{\mathsf T}
+
E_kS_k^\star E_k^{\mathsf T},
\;
1\le k\le T.
\label{eq:leading_covariance_response}
\end{equation}
The matrices \(\Omega_k\) are positive semidefinite and homogeneous quadratic in the prescribed defect sequence \(E_{1:k}\).

Define the quartic innovation-inflation and local-gain-reoptimization functionals
\begin{align}
\mathcal C_T^{M,\mathrm{infl}}(E)
&:=
\sum_{k=1}^{T}
\tr\!\left[
W_{k,T}^M
E_kHF\,\Omega_{k-1}(HF)^{\mathsf T}E_k^{\mathsf T}
\right],
\label{eq:quartic_inflation_response}
\\
\mathcal C_T^{M,\mathrm{reopt}}(E)
&:=
\sum_{k=1}^{T}
\tr\!\left[
W_{k,T}^M
\Phi_k^\star\Omega_{k-1}(HF)^{\mathsf T}
(S_k^\star)^{-1}
\right.
\notag\\
&\qquad\left.
{}\times HF\,\Omega_{k-1}(\Phi_k^\star)^{\mathsf T}
\right].
\label{eq:quartic_reoptimization_response}
\end{align}
Both functionals are nonnegative and homogeneous quartic in \(E_{1:T}\). Here, a prescribed intrinsic defect sequence \(E_{1:T}\) denotes an argument of the deterministic finite-horizon response map and does not imply advance knowledge of future execution defects by the online estimator.

\begin{figure*}[t]
  \centering
  \includegraphics[width=\textwidth]{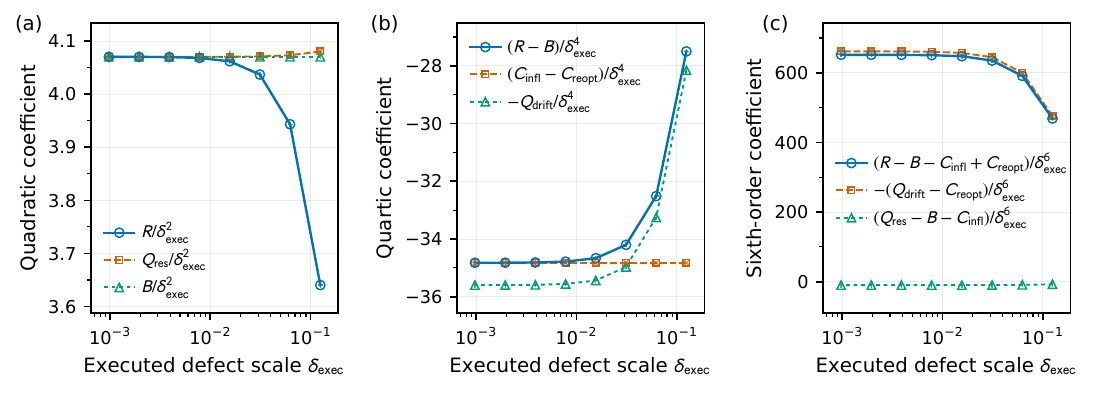}
  \caption{
  Signed-quartic response and sixth-order remainder for one synthetic linear-Gaussian instance constructed from the network geometry of the pandapower IEEE 14-bus case~\cite{pandapower} (IEEE 14; Sec.~\ref{subsec:benchmark}) with \(n=27\), \(m=64\), \(T=100\), and \(M_k=I\). The reference and implemented recursions are initialized identically, and a single prescribed defect-direction sequence is used to construct gains at eight amplitudes from \(2^{-10}\) to \(2^{-3}\), where \(\delta_{\rm exec}:=\max_k\|E_k\|_F\) is the represented executed-defect amplitude and corresponds to \(\delta\) in Eq.~\eqref{eq:finite_horizon_defect_amplitude}. In the figure, \(R,B,Q_{\rm res},Q_{\rm drift},C_{\rm infl},C_{\rm reopt}\) denote \(\mathcal R_T^M,\mathcal B_T^M,\mathcal Q_T^{\rm res},\mathcal Q_T^{\rm drift},\mathcal C_T^{M,\mathrm{infl}},\mathcal C_T^{M,\mathrm{reopt}}\), respectively. (a) Quadratic effect. (b) Quartic effect. (c) Sixth-order remainder and its residual and drift components.
  }
  \label{fig:signed_quartic}
\end{figure*}

\begin{theorem}[Signed-quartic finite-horizon response]
\label{thm:quartic_finit_horiz_respon}
Suppose the deterministic model and covariance conditions of Sec.~\ref{subsec:refer_and_implem_recur} and the matched initialization in Assumption~\ref{ass:matched_initialization} hold. Fix $T<\infty$, the reference schedule, task metrics \(M_{1:T}\) with \(M_k\succeq0\), and a finite cap \(\bar\delta<\infty\). Uniformly over all prescribed intrinsic defect sequences satisfying
\[
\delta=\max_{1\le k\le T}\|E_k\|_F\le\bar\delta,
\]
the covariance mismatch and induced gain drift satisfy
\begin{equation}
\max_{1\le k\le T}\|\Xi_k\|_F=O_T(\delta^2),
\qquad
\max_{1\le k\le T}\|\Gamma_k\|_F=O_T(\delta^2).
\label{eq:second_order_covariance_gain_drift}
\end{equation}
Moreover,
\begin{align}
\mathcal Q_T^{\rm res}
&=
\mathcal B_T^M(E)
+
\mathcal C_T^{M,\mathrm{infl}}(E)
+
O_T(\delta^6),
\label{eq:residual_quartic_expansion}
\\
\mathcal Q_T^{\rm drift}
&=
\mathcal C_T^{M,\mathrm{reopt}}(E)
+
O_T(\delta^6).
\label{eq:drift_quartic_expansion}
\end{align}
Consequently,
\begin{equation}
\mathcal R_T^M
=
\mathcal B_T^M(E)
+
\mathcal C_T^{M,\mathrm{infl}}(E)
-
\mathcal C_T^{M,\mathrm{reopt}}(E)
+
O_T(\delta^6).
\label{eq:signed_quartic_response}
\end{equation}
More precisely, there exists a finite constant \(C_{6,T}^M(\bar\delta)\), independent of the defect directions and of \(\delta\) within the prescribed cap, such that
\begin{equation}
\left|
\mathcal R_T^M
-
\mathcal B_T^M(E)
-
\mathcal C_T^{M,\mathrm{infl}}(E)
+
\mathcal C_T^{M,\mathrm{reopt}}(E)
\right|
\le
C_{6,T}^M(\bar\delta)\delta^6.
\label{eq:quartic_uniform_remainder}
\end{equation}
\end{theorem}

\begin{proof}
See Appendix~\ref{app:proof_of_quart_finit_horiz}.
\end{proof}

The symbol \(O_T(\delta^j)\) in Eqs.~\eqref{eq:second_order_covariance_gain_drift}--\eqref{eq:signed_quartic_response} means that the implied constant may depend on \(T\), the fixed model matrices, the reference schedule, the task metrics \(M_{1:T}\), the floor \(\underline r\), and the cap \(\bar\delta\), but not on the defect sequence within that cap. Equation~\eqref{eq:quartic_uniform_remainder} makes the same convention explicit for \(\mathcal R_T^M\). In particular, when the magnitude of the net signed quartic correction exceeds the sixth-order remainder bound, its sign determines whether the quadratic approximation under- or overpredicts the exact response; if it exceeds twice that bound, including the quartic correction is guaranteed to reduce the absolute prediction error.

Theorem~\ref{thm:quartic_finit_horiz_respon} also gives the quadratic normal form
\begin{equation}
\mathcal R_T^M
=
\mathcal B_T^M(E)
+
O_T(\delta^4),
\label{eq:fourth_order_risk_normal_form}
\end{equation}
which is a uniform absolute finite-horizon expansion on the bounded defect tube. Because \(\mathcal B_T^M\) may vanish, this does not imply uniform relative accuracy over the full defect class. Task-null directions, relative accuracy of the quadratic normal form, and the symmetry and analyticity properties of the response map are discussed in Appendix~\ref{app:consequences_of_the_signed}.

The local centering at $K_k^{\mathrm{loc}}$ exposes the recursive order structure. The fixed-reference response \(\mathcal B_T^M\) is quadratic in the executed defects, while previous defects displace the implemented covariance only at second order. Recursive feedback can therefore first contribute at quartic order. This is the structure of a perturbative response expansion. \(\mathcal B_T^M\) propagates each defect through the unperturbed reference loop, while the quartic terms describe the leading feedback from defect-induced covariance displacement. Innovation inflation couples earlier displacement to later defects, whereas local-gain reoptimization contributes even after a single isolated defect. The quartic correction consists of two competing nonnegative mechanisms: \(\mathcal C_T^{M,\mathrm{infl}}\) captures the increased residual-side exposure caused by innovation covariance inflation, whereas \(\mathcal C_T^{M,\mathrm{reopt}}\) captures the response removed as the exact local gain adapts to the displaced covariance. Their difference is not sign-definite in general.

Figure~\ref{fig:signed_quartic} displays the asymptotic hierarchy predicted by Theorem~\ref{thm:quartic_finit_horiz_respon}. In Fig.~\ref{fig:signed_quartic}(a), the normalized finite-horizon response approaches the fixed-reference quadratic contribution as the executed defect decreases. Figure~\ref{fig:signed_quartic}(b) isolates the leading recursive-feedback correction: along this prescribed defect direction, the quartic contribution is negative, with local-gain reoptimization exceeding innovation-covariance inflation. After both quartic terms are retained, the residual in Fig.~\ref{fig:signed_quartic}(c) is consistent with the sixth-order remainder in Eq.~\eqref{eq:quartic_uniform_remainder}.

From the retained full-precision values underlying Fig.~\ref{fig:signed_quartic}, the quadratic and signed-quartic relative prediction errors, both normalized by \(\lvert\mathcal R_T^M\rvert\), are \(11.803\%\) and \(3.144\%\), respectively, at \(\delta_{\rm exec}=1/8\); at \(\delta_{\rm exec}=1/16\), the corresponding errors are \(3.221\%\) and \(0.229\%\).
Along this controlled defect sequence, retaining the signed-quartic term therefore improves prediction of the exact recursive response beyond the quadratic approximation.

\subsection{Residual-response bounds and finite-horizon control}
\label{subsec:residu_respon_bound_and_finite_horizon}

To identify \(\mathcal R_T^M\) with physical estimation-error risk, we additionally require measurement-independent execution.

\begin{proposition}[Statistical identification under measurement-independent execution]
\label{prop:statistical_identification}
Let \(e_k:=\widehat x_k-\widehat x_k^\star\), and let \(\cG\) be the execution environment defined in Appendix~\ref{app:statistical_identification}. Suppose that \(\cG\) is independent of the physical initial state and complete disturbance trajectory, and that the executed gain sequence \(\widetilde K_{1:T}\) is determined entirely by \(\cG\) (equivalently, is \(\cG\)-measurable). Under Assumption~\ref{ass:matched_initialization},
\begin{align}
\widehat P_k
&=
\E\!\left[
(x_k-\widehat x_k)(x_k-\widehat x_k)^{\mathsf T}
\,\middle|\,\cG
\right].
\label{eq:implemented_covariance_interpretation}\\
\Xi_k
&=
\E[e_ke_k^{\mathsf T}\mid\cG]
=
\Cov(e_k\mid\cG),
\label{eq:covariance_mismatch_statistical}
\end{align}
where the second identity requires the reference recursion to be the correctly specified conditional-mean estimator. Consequently,
\begin{align}
\mathcal R_T^M
&=
\sum_{k=1}^{T}
\E[e_k^{\mathsf T}M_ke_k\mid\cG],
\label{eq:conditional_physical_risk_identification}\\
&=
\sum_{k=1}^{T}
\E\!\left[
\|x_k-\widehat x_k\|_{M_k}^2
-
\|x_k-\widehat x_k^\star\|_{M_k}^2
\,\middle|\,\cG
\right].
\label{eq:conditional_excess_risk_identification}
\end{align}
\end{proposition}

\begin{proof}
See Appendix~\ref{app:statistical_identification}.
\end{proof}

For each executed step, define the exact residual contribution
\begin{equation}
q_{k,T}
:=
\tr\!\left[
W_{k,T}^M
\rho_k^{\rm exec}
S_k^{-1}
(\rho_k^{\rm exec})^{\mathsf T}
\right].
\label{eq:exact_step_residual_contribution}
\end{equation}
Then \(q_{k,T}\ge0\) and \(\mathcal Q_T^{\rm res}=\sum_{k=1}^{T}q_{k,T}\). Unlike \(\|\rho_k^{\rm exec}\|_F\), \(q_{k,T}\) retains both the state-side task weight \(W_{k,T}^M\) and the innovation-side metric \(S_k^{-1}\); residuals with equal Frobenius norm can therefore have different finite-horizon contributions. A spectral characterization of the fixed-reference quadratic exposure and its residual-coordinate form at a matched current step is given in Appendix~\ref{subsec:residu_magnit_and_task_depend}.

Using the lower bound \(\ell_k\le\sigma_{\min}(S_k)\), define the task-aware and radial upper bounds
\begin{align}
c_{k,T}^{W,\ell}
&:=
\ell_k^{-1}
\tr\!\left[
W_{k,T}^M\rho_k^{\mathrm{exec}}
(\rho_k^{\mathrm{exec}})^{\mathsf T}
\right],
\label{eq:task_aware_step_upper_bound}\\
c_{k,T}^{\mathrm{rad}}
&:=
\frac{\lambda_{\max}(W_{k,T}^M)}{\ell_k}
\|\rho_k^{\mathrm{exec}}\|_F^2.
\label{eq:radial_step_upper_bound}
\end{align}
Since \(S_k^{-1}\preceq\ell_k^{-1}I_m\) and \(W_{k,T}^M\preceq\lambda_{\max}(W_{k,T}^M)I_n\),
\[
q_{k,T}
\le
c_{k,T}^{W,\ell}
\le
c_{k,T}^{\mathrm{rad}}.
\]
Thus \(c_{k,T}^{W,\ell}\) retains the directional task weighting, whereas \(c_{k,T}^{\mathrm{rad}}\) depends on the residual only through its Frobenius norm.

\begin{proposition}[Composable finite-horizon task-response control]
\label{prop:finite_horizon_respon_control}
Under the assumptions of the exact accounting identity [Eq.~\eqref{eq:exact_residual_risk_accounting}], including matched initialization,
\begin{equation}
0\le
\mathcal R_T^M
\le
\sum_{k=1}^{T}q_{k,T}
\le
\sum_{k=1}^{T}c_{k,T}^{W,\ell}
\le
\sum_{k=1}^{T}c_{k,T}^{\mathrm{rad}}.
\label{eq:risk_bound_hierarchy}
\end{equation}
Consequently, if nonnegative allocations \(b_{k,T}\) satisfy
\begin{equation}
\sum_{k=1}^{T}b_{k,T}\le\tau,
\label{eq:task_budget_allocation}
\end{equation}
and each gain actually executed at step \(k\), including any fallback, satisfies
\(c_{k,T}^{W,\ell}\le b_{k,T},\) then
\begin{equation}
\mathcal R_T^M\le\tau.
\label{eq:global_task_risk_guarantee}
\end{equation}
\end{proposition}

\begin{proof}
See Appendix~\ref{app:finite_horizon_respon_control}.
\end{proof}

For the exact covariance model, the local residual criterion of Eq.~\eqref{eq:accept_rule}, applied to the executed gain, corresponds to the radial finite-horizon allocation \(c_{k,T}^{\mathrm{rad}}\le b_{k,T}^{\mathrm{adm}}\), with
\begin{equation}
b_{k,T}^{\mathrm{adm}}
:=
\lambda_{\max}(W_{k,T}^M)\,\ell_k\,\delta_{\mathrm{adm}}^2;
\label{eq:admissibility_induced_task_budget}
\end{equation}
the implication is an equivalence when \(\lambda_{\max}(W_{k,T}^M)>0\), whereas the radial allocation is vacuous when \(W_{k,T}^M=0\). Thus, if every executed step satisfies Eq.~\eqref{eq:accept_rule}, then
\begin{equation}
\mathcal R_T^M
\le
\tau_T^{\mathrm{adm}},
\qquad
\tau_T^{\mathrm{adm}}
:=
\sum_{k=1}^{T}b_{k,T}^{\mathrm{adm}}.
\label{eq:admissibility_induced_horizon_bound}
\end{equation}
This connects the local admissibility criterion to a sufficient finite-horizon response bound. The represented-system certificate of Sec.~\ref{subsec:residual_certifi_and_fail_close} establishes this exact-model implication only when matrix-formation and covariance-propagation errors are separately controlled. The reverse construction, from a prescribed task-response allocation to local residual tolerances, is given in Appendix~\ref{app:finite_horizon_respon_control}. Uniform-in-time sufficient bounds under strict Euclidean contraction are given separately in Appendix~\ref{app:uniform_in_time_bounds}.

\section{Local Repairability and Machine-Learning Correction}
\label{sec:ml_correc_and_local}

The response theory of Sec.~\ref{sec:finit_horiz_respon} takes the executed defect sequence as given.  We now separate three pre-execution questions: whether the declared correction class contains an admissible action, whether the learned map attains one, and whether the proposed candidate passes the independent residual check.

\subsection{Correction class and local repairability}
\label{subsec:correc_class_and_local}

At step \(k\), let \(\mathcal V_{\rm corr}\subseteq\R^{n\times m}\) be the correction subspace and \(\delta_c\ge0\) the correction budget. Define
\[
\mathcal C_k
:=
\left\{
\Delta\in\mathcal V_{\rm corr}:
\|\Delta\|_F\le\delta_c
\right\}.
\]
The class \(\mathcal C_k\) specifies the available correction actions independently of how a particular proposal is constructed.

Let \(\Pi_{\rm corr}\) be the Frobenius-orthogonal projector onto \(\mathcal V_{\rm corr}\), and decompose
\[
A_k=A_k^\parallel+A_k^\perp,
\qquad
A_k^\parallel:=\Pi_{\rm corr}A_k,
\qquad
A_k^\perp:=A_k-A_k^\parallel.
\]
Only \(A_k^\parallel\) can be modified by a correction in \(\mathcal C_k\); \(A_k^\perp\) is an irreducible component for the declared correction subspace.

\begin{theorem}[Exact one-step repairability boundary]
\label{thm:repairability}
For the correction class \(\mathcal C_k\), define \(\varepsilon_k\ge0\) by
\[
\begin{aligned}
\varepsilon_k^2
&:=
\min_{\Delta\in\mathcal C_k}
\|A_k+\Delta\|_F^2
\\
&=
\|A_k^\perp\|_F^2
+
\bigl(\|A_k^\parallel\|_F-\delta_c\bigr)_+^2,
\end{aligned}
\]
where \(a_+:=\max\{a,0\}\). Consequently,
\[
\exists\,\Delta\in\mathcal C_k:\
\|A_k+\Delta\|_F\le\delta_{\rm adm}
\quad\Longleftrightarrow\quad
\varepsilon_k\le\delta_{\rm adm}.
\]
\end{theorem}

\begin{proof}
See Appendix~\ref{app:repair_bound_and_minim_correc}.
\end{proof}

The minimization is a projection problem over a norm-bounded correction subspace. Its role here is to provide an exact necessary-and-sufficient criterion for local correction feasibility, against which learned attainment and residual acceptance can be evaluated separately.

Geometrically, the theorem states that the reachable set \(K_k^{\rm alg}+\mathcal C_k\) intersects the admissible ball centered on \(K_k^{\rm loc}\) if and only if \(\varepsilon_k\le\delta_{\rm adm}\). Figure~\ref{fig:conceptual_response}(a) illustrates an intersection for which the learned candidate nevertheless lies outside the admissible ball, separating existence of an admissible correction from learned attainment. For the same exact gain system, a candidate outside the ball cannot satisfy the sufficient residual condition in Eq.~\eqref{eq:accept_rule}, whereas a candidate inside need not satisfy it. Residual acceptance is therefore a separate requirement; certification for represented numerical systems is addressed in Sec.~\ref{subsec:residual_certifi_and_fail_close}.

The repairability boundary partitions each local problem into four mutually exclusive regimes, summarized in Table~\ref{tab:regions}. The harmless regime depends only on the raw computational defect and the admissibility tolerance. For a raw-inadmissible defect, the orthogonal component \(A_k^\perp\) determines whether the declared correction subspace can reach the admissible ball: if it cannot, the instance is subspace-obstructed independently of the correction budget. Otherwise, repairability is determined by whether the available budget \(\delta_c\) is sufficient to remove the required part of \(A_k^\parallel\). Thus subspace obstruction cannot be removed by increasing \(\delta_c\) at fixed \(\mathcal V_{\rm corr}\), whereas budget obstruction can. These regimes characterize one-step feasibility within the declared correction class; learned attainment and residual acceptance remain separate questions. The minimizing correction and the corresponding minimum required correction budget are derived in Appendix~\ref{app:repair_bound_and_minim_correc}.

\begin{table*}[tb]
\centering
\caption{%
    \label{tab:regions}%
    Local repairability boundary at step \(k\) for the declared correction class. Here \(A_k\) is the raw local computational defect, \(A_k^\perp\) is its component orthogonal to the correction subspace, \(\varepsilon_k\) is the minimum postcorrection defect attainable within the declared correction budget, and \(\delta_{\rm adm}\) is the local admissibility tolerance.}
\begin{ruledtabular}
    \begin{tabular}{lll}
        Class & Condition & Mechanism \\
        \hline
        Harmless
        & \(\|A_k\|_F \le \delta_{\mathrm{adm}}\)
        & The raw solve is already locally admissible \\
        Repairable
        & $\|A_k\|_F > \delta_{\mathrm{adm}},\ \varepsilon_k \le \delta_{\mathrm{adm}}$
        & The correction class contains a locally admissible action \\
        Subspace-obstructed
        & \(\|A_k^\perp\|_F>\delta_{\mathrm{adm}}\)
        & The correction subspace cannot reach the admissible ball \\
        Budget-obstructed & \(\|A_k^\perp\|_F\le\delta_{\mathrm{adm}},\ \varepsilon_k>\delta_{\mathrm{adm}}\)& The available budget is insufficient within the correction subspace
    \end{tabular}
\end{ruledtabular}
\end{table*}

\begin{figure}[t]
  \centering
  \includegraphics[width=\columnwidth]{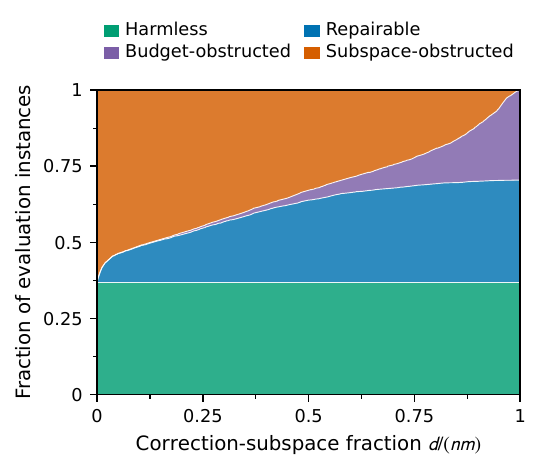}
  \caption{
  Repairability boundary for finite conjugate gradient (CG) gain defects. The stacked bands show the fractions of 4096 held-out evaluation instances classified as harmless (green), repairable (blue), budget-obstructed (purple), and subspace-obstructed (orange) by Theorem~\ref{thm:repairability} as a function of the normalized correction-subspace fraction \(d/(nm)\), where \(d=\dim(\mathcal V_{\rm corr})\). Here \(n=27\), \(m=16\), and \(nm=432\). For each independently generated local estimation instance, the approximate gain is obtained after three zero-start, unpreconditioned CG iterations. The correction subspaces are nested prefixes of a single principal component analysis (PCA) basis fitted on an independent acquisition cohort, and the independently calibrated correction budget \(\delta_c\) is held fixed across \(d\). The harmless fraction is independent of \(d\), and subspace obstruction vanishes at \(d=nm\), where only budget obstruction remains.
  }
  \label{fig:repairability_boundary}
\end{figure}

To evaluate local repairability separately from the recursive studies, we use a fixed ensemble of synthetic linear-Gaussian instances constructed from the IEEE 14-bus power-network benchmark~\cite{pandapower}, with \(n=27\), \(m=16\), and gain-space dimension \(nm=432\). The local instances use prediction covariances obtained after 1000 fixed Joseph-form Riccati updates from a common initialization, and the approximate gain is obtained after three zero-start, unpreconditioned conjugate-gradient (CG) iterations applied independently to the gain-system right-hand sides. A single raw, uncentered, unnormalized principal component analysis (PCA) basis is fitted on an independent acquisition cohort, and the evaluated correction subspaces are its nested prefixes with \(d=0,1,\ldots,432\). On a disjoint calibration cohort, the correction budget \(\delta_c\) is set to the nearest-rank median of the full-space minimum required correction budgets over raw-inadmissible instances and is then held fixed across \(d\). The four repairability classes are evaluated on a separate held-out cohort. The predeclared diagnostic radius is used only for this local one-step classification and is distinct from the runtime residual certificate and the finite-horizon response bounds.

Figure~\ref{fig:repairability_boundary} shows how the two obstruction mechanisms change with available correction capacity. The harmless fraction is independent of \(d\) because it is determined before correction. Increasing \(d\) reduces the component of the raw defect orthogonal to the correction subspace, allowing instances that are subspace-obstructed at smaller \(d\) to become either repairable or budget-obstructed according to the fixed budget \(\delta_c\). Each held-out instance is classified directly by Theorem~\ref{thm:repairability}, separating geometric repairability from learned coefficient prediction.

\subsection{Bounded machine-learning correction}
\label{subsec:bound_ml_correc}

The repairability boundary characterizes the complete correction class but does not select an action within it. We use a checkpoint-specific learned map to propose one bounded correction from the raw represented residual; the admissibility criterion and subsequent execution decision remain independent of the learner.

At each CG iteration \(t\), an offline construction defines a frozen Frobenius-orthonormal correction basis
\[
\mathcal V_{{\rm corr},t}
=
\operatorname{span}\{U_{t,1},\ldots,U_{t,d_{{\rm eff},t}}\},
\qquad
\langle U_{t,i},U_{t,j}\rangle_F=\delta_{ij},
\]
where \(d_{{\rm eff},t}\le d_{\rm cap}\) is the number of numerically supported directions retained under the prescribed correction-dimension cap \(d_{\rm cap}\) (Appendix~\ref{app:correc_basis_and_learn_model}). The same offline construction fixes an uncentered RMS feature scaling, a no-intercept ridge map
\(\Theta_t\in\R^{d_{{\rm eff},t}\times nm},\)
and a checkpoint-local correction radius \(\delta_{c,t}^{\star}\), all of which are frozen before deployment. If \(d_{{\rm eff},t}=0\), no learned correction is applied.

The deployed feature \(\phi_{k,t}\in\R^{nm}\) is an uncentered RMS-scaled representation of the explicitly recomputed raw residual \(\rho_{k,t}^{\rm alg}\). It contains neither the current measurement nor the innovation, and neither the exact local gain nor the local computational defect. The radius instantiates the correction budget but is calibrated independently of \(\delta_{\rm adm}\) and of the minimum budget required for repairability. Basis construction, feature scaling, targets, calibration, and fitting are specified in Appendix~\ref{app:correc_basis_and_learn_model}.

The predicted coefficients are projected onto the checkpoint-local coefficient ball,
\[
\widehat c_{k,t}
=
\operatorname{proj}_{\|c\|_2\le\delta_{c,t}^{\star}}
(\Theta_t\phi_{k,t})
\in\R^{d_{{\rm eff},t}},
\]
where \(\operatorname{proj}_{\|c\|_2\le r}\) denotes Euclidean projection onto the closed coefficient ball of radius \(r\). 
The proposed correction is then
\[
\Delta_{k,t}^{\rm ml}
=
-\sum_{i=1}^{d_{{\rm eff},t}}
[\widehat c_{k,t}]_iU_{t,i}.
\]
The corresponding candidate gain is \(K_{k,t}^{\rm cand} = K_{k,t}^{\rm alg}+\Delta_{k,t}^{\rm ml}.\)
Its residual is recomputed explicitly:
\[
\begin{aligned}
\rho_{k,t}^{\rm cand}
&:=
\rho_k(K_{k,t}^{\rm cand})
\\
&=
K_{k,t}^{\rm cand}S_k-\widetilde P_kH^{\mathsf T}
=
\rho_{k,t}^{\rm alg}+\Delta_{k,t}^{\rm ml}S_k.
\end{aligned}
\]

Frobenius orthonormality and coefficient projection give
\[
\|\Delta_{k,t}^{\rm ml}\|_F
=
\|\widehat c_{k,t}\|_2
\le
\delta_{c,t}^{\star},
\]
so every proposed correction belongs to the declared checkpoint-local correction class. Because the features are uncentered and the ridge map has no intercept,
\[
\rho_{k,t}^{\rm alg}=0
\quad\Longrightarrow\quad
\Delta_{k,t}^{\rm ml}=0,
\]
so the construction is zero consistent. The explicitly recomputed candidate residual \(\rho_{k,t}^{\rm cand}\) is assessed by the certification and fallback procedure in the next section.

\subsection{Residual certification and verified fallback}
\label{subsec:residual_certifi_and_fail_close}

For a represented candidate \(L\), the matrix residual is recomputed explicitly from \(L\), \(S_k\), and \(\widetilde P_kH^{\mathsf T}\). Equation~\eqref{eq:accept_rule} gives the exact-arithmetic sufficient criterion for local admissibility. Its finite-precision implementation accounts for residual-evaluation and norm-computation error. The forward-error guarantee additionally requires \(0<\ell_k\le\sigma_{\min}(S_k)\) for the current represented matrix. The candidate is accepted only when the resulting residual upper bound is no larger than the corresponding verified lower bound on \(\ell_k\delta_{\rm adm}\).

The policy-specific approximate candidate is evaluated first. If it is verified, that gain is used in the state and Joseph covariance updates. Otherwise, a fallback gain is computed by a Cholesky solve of the current represented gain system and evaluated under the same represented-residual certificate. The fallback is additionally checked through its direct backward residual.

Provided this represented-matrix lower-bound premise holds, numerical certification bounds the Frobenius-norm error of the executed gain relative to the exact solution of the represented gain system by \(\delta_{\rm adm}\). This guarantee concerns the represented gain system; identifying this error with \(E_k\) requires separate control of matrix-formation and covariance-propagation errors. Under the exact covariance model, the residual criterion bounds the current executed defect \(E_k\), but not the inherited gain drift \(\Gamma_k\).

The admissibility tolerance links local execution to recursive response: the same \(\delta_{\rm adm}\) defines the admissible ball in Theorem~\ref{thm:repairability} and the residual condition in Eq.~\eqref{eq:accept_rule}. Under the exact covariance model, if every executed gain satisfies that condition, then \(\delta\le\delta_{\rm adm}\); with matched initialization, Eq.~\eqref{eq:admissibility_induced_horizon_bound} follows. Transferring the represented-system certificate to this exact-model conclusion additionally requires the matrix-formation and covariance-propagation controls stated above. If the hypotheses of Theorem~\ref{thm:quartic_finit_horiz_respon} hold on the defect tube with \(\bar\delta=\delta_{\rm adm}\), the finite-horizon response admits the quadratic and signed-quartic expansion with sixth-order remainder bounded by \(C_{6,T}^M(\delta_{\rm adm})\,\delta_{\rm adm}^6\). Certified defects of equal norm can contribute differently to \(q_{k,T}\), so recursive consequence is evaluated separately in Sec.~\ref{sec:results}.

\section{Results}
\label{sec:results}

We evaluate the effects of bounded learned correction on certified solver effort, complete runtime, and recursive response across measurement dimensions, network sizes, and horizons. Repairability, learned attainment, and residual acceptance are treated as distinct local properties, while finite-horizon response is evaluated separately from the executed gains using the exact residual--drift identity~[Eq.~\eqref{eq:exact_residual_risk_accounting}]. The same execution framework is also applied to reconstructed quantum-solver gains. The classical computations were performed primarily on the JURECA supercomputer at Forschungszentrum Jülich~\cite{JURECA}. The quantum-hardware components comprise terminal VQLS measurements on the Origin Wukong superconducting processor \cite{originqcloud} and QUBO sampling on the D-Wave Advantage2 system JUPSI \cite{mcgeoch2022advantage2}.

\subsection{Experimental design and evaluation protocol}
\label{subsec:benchmark}

The benchmark models are synthetic linear-Gaussian systems built on standard power-network test cases distributed with pandapower~\cite{pandapower}; we label them by bus count as IEEE~\(N\). For a case with \(N\) buses, the state contains the real and imaginary bus-voltage coordinates with the imaginary coordinate of the slack (reference) bus removed, giving \(n=2N-1\). Each measurement is a normalized scalar coordinate, either a voltage coordinate or a real or imaginary branch-current coordinate obtained from the branch admittance matrices, so \(H\) is exactly linear and \(m\) may exceed \(n\). The network enters only through its topology and branch admittances: the state-transition matrix \(F\) is a stable synthetic construction informed by the topology, the process and measurement covariances are synthetic positive-definite models, there is no input term, and the simulated state is zero-mean. The measurement covariance is constant, whereas a predeclared, slowly varying scalar schedule modulates the transition process covariance as \(Q_k=c_kQ_{\rm base}\), so the gain systems vary along the chronology. The models are controlled benchmarks for computation within recursive estimation. Appendix~\ref{app:benchmark_construction} specifies the construction.

For each solver family, we compare a monitored policy with a learned-corrected policy under the same residual certificate and verified-fallback rule. Quantities requiring the exact local gain \(K_k^{\rm loc}\), including the raw defect \(A_k\), are restricted to offline construction, calibration, and post-execution evaluation and are not available to the deployed learner. The monitored policies (M-CG, M-VQLS, and M-QUBO) submit the raw gain candidate directly to the certificate. Their learned-corrected counterparts (LC-CG, LC-VQLS, and LC-QUBO) first apply the frozen bounded correction and then submit the corrected candidate. CG candidates are obtained after a fixed iteration count, whereas VQLS and QUBO candidates are reconstructed and assembled into complete gains before residual certification. The classical CG policies are deterministic and measurement independent, providing a controlled setting for studying finite-compute defects. Within each comparison, all policies begin deployment from the same commissioned reference state estimate and covariance. After conditioning on the commissioning measurements and reindexing at the final commissioning posterior, the shared initial mean and covariance satisfy Assumption~\ref{ass:matched_initialization}; their common prediction supplies the prior for the first deployment update. Approximate gains enter the recursion only after commissioning.

\begin{figure*}[t]
  \centering
  \includegraphics[width=\textwidth]{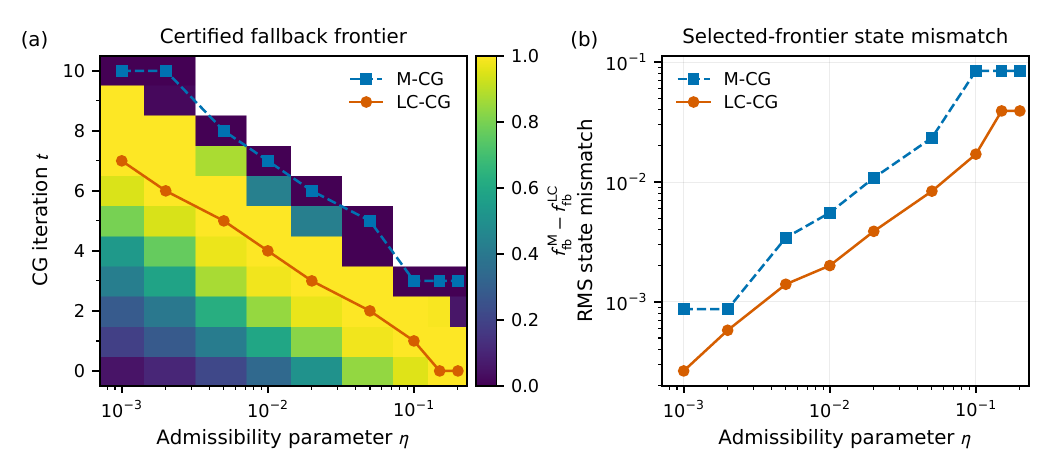}
  \caption{Monitored conjugate gradient (M-CG) and learned-corrected conjugate gradient (LC-CG) on one IEEE 14 realization with \(n=27\), \(m=64\), and correction dimension cap \(d_{\rm cap}=64\). Both policies use the same residual certificate, admissibility tolerance, and verified fallback at each \(\eta\). Evaluation covers 600 deployment steps following 400 commissioning steps. (a) Reduction in fallback fraction, \(f_{\rm fb}^{\rm M}-f_{\rm fb}^{\rm LC}\), from Eq.~\eqref{eq:fallback_fraction}, versus admissibility parameter \(\eta\in\{0.001,0.002,0.005,0.01,0.02,0.05,0.1,0.15,0.2\}\) and CG iteration \(t\). Curves mark each policy's smallest evaluated CG iteration with zero deployment fallbacks; blank cells were not evaluated. (b) Normalized root mean square (RMS) estimator-to-reference state mismatch [Eq.~\eqref{eq:normalized_rms_state_mismatch}] at each policy's own zero-fallback frontier.
  }
  \label{fig:scaling_frontier}
\end{figure*}

The admissibility tolerance is parameterized as \(\delta_{\rm adm}(\eta)=\eta\,\delta_{\rm FH}\), where \(\eta>0\) is dimensionless and \(\delta_{\rm FH}\) is a configuration-specific scale fixed from the commissioning reference before deployment. For a commissioning window of \(C\) updates,
\begin{equation}
\begin{aligned}
J_{\rm tr}
&:=
\sum_{k=1}^{C}\tr\!\left(\widehat P_k^\star\right),
\qquad
G_{\rm tr}
:=
\underline r
\sum_{k=1}^{C}
\lambda_{\max}\!\left(W_{k,C}^{I}\right),\\
\delta_{\rm FH}
&:=
\sqrt{J_{\rm tr}/G_{\rm tr}} .
\end{aligned}
\label{eq:commissioning_scale}
\end{equation}
Here \(W_{k,C}^{I}\) is the response operator of Eq.~\eqref{eq:task_response_operator} for the identity task metric over the commissioning horizon. With the uniform lower bound \(\ell_k=\underline r\) used in these studies, Eq.~\eqref{eq:admissibility_induced_horizon_bound} gives \(\tau_C^{\rm adm}=G_{\rm tr}\delta_{\rm adm}^2\), so that \(\eta^2=\tau_C^{\rm adm}/J_{\rm tr}\) is the ratio of the admissibility-induced radial response upper bound to the accumulated reference posterior uncertainty on the commissioning horizon. This commissioning normalization sets the experimental admissibility scale; it is distinct from a deployment-horizon response certificate. Within each configuration, \(\delta_{\rm FH}\) is held fixed across \(\eta\), CG iteration, and policy. Equation~\eqref{eq:accept_rule} gives \(\|\rho_k(L)\|_F\le\underline r\,\eta\,\delta_{\rm FH}\); the executable certificate applies the corresponding finite-precision residual and denominator enclosures described in Sec.~\ref{subsec:residual_certifi_and_fail_close}. Smaller \(\eta\) imposes a tighter local admissibility requirement.

For an executed policy \(\mathsf a\in\{\mathrm{M},\mathrm{LC}\}\) over deployment index set \(\mathcal D\), define the fallback fraction
\begin{equation}
f_{\rm fb}^{(\mathsf a)}
:=
\frac{N_{\rm fb}^{(\mathsf a)}}{|\mathcal D|},
\label{eq:fallback_fraction}
\end{equation}
where \(N_{\rm fb}^{(\mathsf a)}\) is the number of deployment steps on which the verified fallback is executed.
For the common reference trajectory, define the reference-state RMS scale
\[
s_{\mathrm{ref}}
:=
\left(
|\mathcal D|^{-1}
\sum_{k\in\mathcal D}
\|\widehat x_k^\star\|_2^2
\right)^{1/2}.
\]
The per-step relative state mismatch is
\begin{equation}
\epsilon_k^{(\mathsf a)}
:=
\frac{
\|\widehat x_k^{(\mathsf a)}-\widehat x_k^\star\|_2
}{
s_{\mathrm{ref}}
},
\qquad k\in\mathcal D.
\label{eq:relative_state_mismatch}
\end{equation}
Its deployment-averaged normalized RMS value is
\begin{equation}
\epsilon_{\mathrm{rms}}^{(\mathsf a)}
:=
\left(
|\mathcal D|^{-1}
\sum_{k\in\mathcal D}
\bigl(\epsilon_k^{(\mathsf a)}\bigr)^2
\right)^{1/2}.
\label{eq:normalized_rms_state_mismatch}
\end{equation}
All state-mismatch metrics use the reference posterior estimate evaluated on the same deployment realization and window. These realized state-space quantities are distinct from the finite-horizon response \(\mathcal R_T^M\).

\subsection{Recursive execution and finite-horizon consequence}
\label{subsec:certif_recursi}

\begin{figure*}[t]
  \centering
  \includegraphics[width=\textwidth]{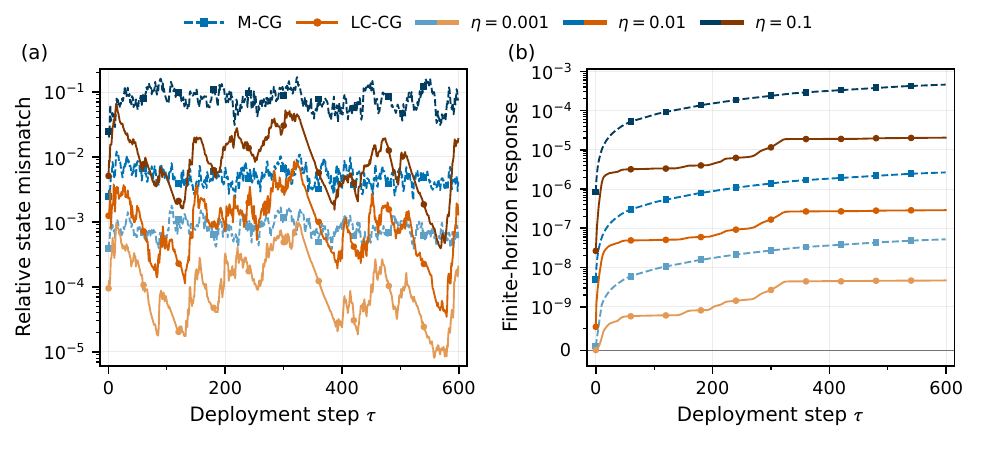}
  \caption{Deployment trajectories for one IEEE 14 realization with \(n=27\), \(m=64\), and \(d_{\mathrm{cap}}=64\), following 400 commissioning steps. The three admissibility parameters are \(\eta=0.001\), \(0.01\), and \(0.1\); the corresponding M-CG and LC-CG zero-fallback iterations are \((10,7)\), \((7,4)\), and \((3,1)\). (a) Per-step relative state mismatch [Eq.~\eqref{eq:relative_state_mismatch}]. (b) Cumulative finite-horizon response from the signed residual-minus-drift contributions. The response weights are fixed by the full 600-step deployment horizon, so intermediate values are cumulative contributions to that fixed-horizon response rather than responses recomputed for shorter horizons. Deployment step \(\tau=0,\ldots,599\) indexes the saved posterior updates, with \(\tau=0\) denoting the first deployment update.}
  \label{fig:cg_trajectory_response}
\end{figure*}

We first evaluate M-CG and LC-CG on the IEEE 14 benchmark \cite{pandapower}.
Figure~\ref{fig:scaling_frontier} connects the residual execution rule to its recursive consequence. Across all nine admissibility parameters, LC-CG reaches zero fallback two to four CG iterations earlier than M-CG [Fig.~\ref{fig:scaling_frontier}(a)]. At the corresponding policy-specific frontiers, its normalized estimator-to-reference mismatch is also smaller at every \(\eta\) [Fig.~\ref{fig:scaling_frontier}(b)]. These are distinct quantities. The frontier measures certified solver depth, whereas the mismatch measures the state-space consequence of the executed policy.

Figure~\ref{fig:cg_trajectory_response} shows the deployment trajectories at three selected admissibility parameters. Evaluated at each policy's own zero-fallback frontier, LC-CG has a smaller deployment-averaged normalized RMS state mismatch from Eq.~\eqref{eq:normalized_rms_state_mismatch} and a smaller total finite-horizon response for all three displayed parameters.

\begin{table*}[t]
\caption{Comparison of local repairability, learned attainment, and residual acceptance on the same deployment steps for IEEE 14 LC-CG with \(m=64\). Each row is a separate \(N_{\rm dep}=600\)-step trajectory. \(t^\star\) denotes the selected zero-fallback iteration; up to two preceding nonnegative iterations are included. \(N_{\rm har}\), \(N_{\rm rep}\), \(N_{\rm sub}\), and \(N_{\rm bud}\) count harmless, repairable, subspace-obstructed, and budget-obstructed deployment steps. Attainment is the fraction of repairable deployment steps whose learned candidate is locally admissible; acceptance is the fraction of those attained candidates that pass the residual certificate.}
\label{tab:repairability_attainment_acceptance}
\centering
\begin{ruledtabular}
\begin{tabular}{c c r r r r r r}
\(\eta\) & \(t\) & \(N_{\rm har}/N_{\rm dep}\) & \(N_{\rm rep}/N_{\rm dep}\) & \(N_{\rm sub}/N_{\rm dep}\) & \(N_{\rm bud}/N_{\rm dep}\) & Attainment (\%) & Acceptance (\%) \\
\hline
0.001 & 5           & 0  & 0.982 & 0 & 0.018 & 100.0 & 81.2 \\
0.001 & 6           & 0  & 1.000 & 0 & 0  & 100.0 & 94.5 \\
0.001 & \(7^\star\) & 0  & 1.000 & 0 & 0  & 100.0 & 100.0 \\
0.01  & 2           & 0  & 1.000 & 0 & 0  & 100.0 & 84.3 \\
0.01  & 3           & 0  & 1.000 & 0 & 0  & 100.0 & 98.0 \\
0.01  & \(4^\star\) & 0  & 1.000 & 0 & 0  & 100.0 & 100.0 \\
0.1   & 0           & 0  & 1.000 & 0 & 0  & 100.0 & 95.0 \\
0.1   & \(1^\star\) & 0.043 & 0.957 & 0 & 0  & 100.0 & 100.0 \\
\end{tabular}
\end{ruledtabular}
\end{table*}

We further compare geometric repairability, learned attainment, and residual acceptance at the same deployment steps. Each step is classified by its local repairability region (Table~\ref{tab:regions}). Within the repairable subset, we separately record whether the retained learned candidate satisfies \(\|K_{k,t}^{\rm cand}-K_k^{\rm loc}\|_F\le\delta_{\rm adm}\) and whether that same candidate passes the residual certificate. Both proportions are descriptive frequencies over the retained trajectories.
We apply the same-sample census to LC-CG at three tolerances, \(\eta\in\{0.001,0.01,0.1\}\). At each tolerance, we include the selected zero-fallback iteration \(t^\star\) from the frontier study and up to two immediately preceding nonnegative iteration counts. This extension rule was fixed before the additional census, and each iteration count is evaluated on its own recursive trajectory.

Table~\ref{tab:repairability_attainment_acceptance} reports the three quantities separately. The shallowest trajectory at the strictest tolerance contains budget-obstructed steps, whereas some raw defects at the loosest selected setting are already locally admissible. No subspace-obstructed steps occur in the eight trajectories. Every retained candidate in the repairable subset attains local admissibility, but residual acceptance ranges from \(81.2\%\) to \(98.0\%\) before the selected frontier and reaches \(100\%\) at each selected zero-fallback iteration. The observed pre-frontier gap therefore lies between learned attainment and residual acceptance, rather than between geometric repairability and learned attainment. This is consistent with the residual certificate providing a sufficient but not necessary condition for local admissibility (Sec.~\ref{subsec:correc_class_and_local}).

\subsection{Measurement-dimension, network, and horizon dependence}
\label{subsec:measu_netwo_horiz_depend}

\begin{figure*}[t]
  \centering
  \includegraphics[width=\textwidth]{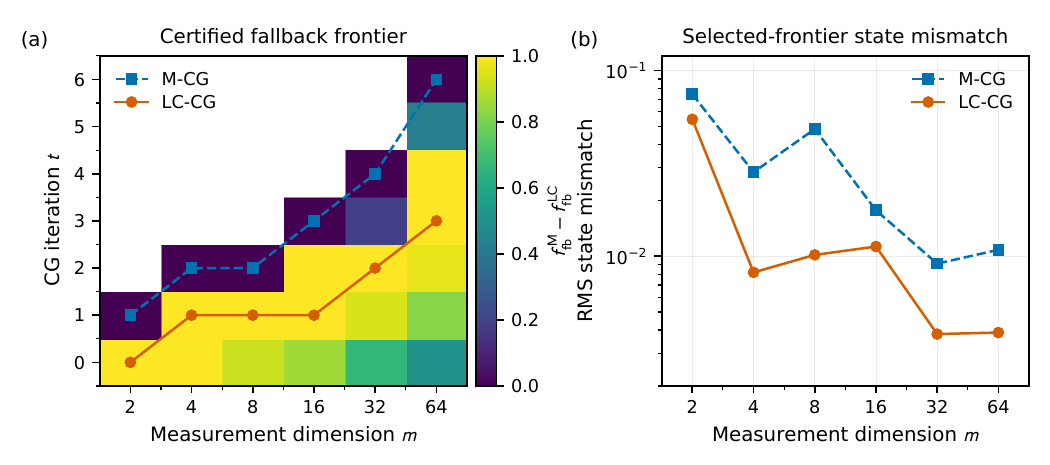}
  \caption{Measurement-dimension dependence on IEEE 14 with \(n=27\), \(d_{\rm cap}=64\), \(\eta=0.020\), 400 commissioning steps, and 600 deployment steps. (a) Fallback-fraction reduction \(f_{\rm fb}^{\rm M}-f_{\rm fb}^{\rm LC}\) [Eq.~\eqref{eq:fallback_fraction}] versus measurement dimension \(m\in\{2,4,8,16,32,64\}\) and CG iteration \(t\). Curves mark each policy's smallest evaluated zero-fallback iteration; blank cells were not evaluated. (b) Normalized RMS estimator-to-reference state mismatch [Eq.~\eqref{eq:normalized_rms_state_mismatch}] at each policy's own selected zero-fallback frontier.}
  \label{fig:measurement_scaling}
\end{figure*}

We vary the measurement dimension on IEEE 14 at fixed \(\eta=0.020\). Across \(m=2,4,8,16,32,64\), LC-CG reaches zero deployment fallback one to three CG iterations earlier than M-CG [Fig.~\ref{fig:measurement_scaling}(a)], and its state mismatch at the selected frontier is smaller at every \(m\) [Fig.~\ref{fig:measurement_scaling}(b)].

\begin{figure*}[t]
  \centering
  \includegraphics[width=\textwidth]{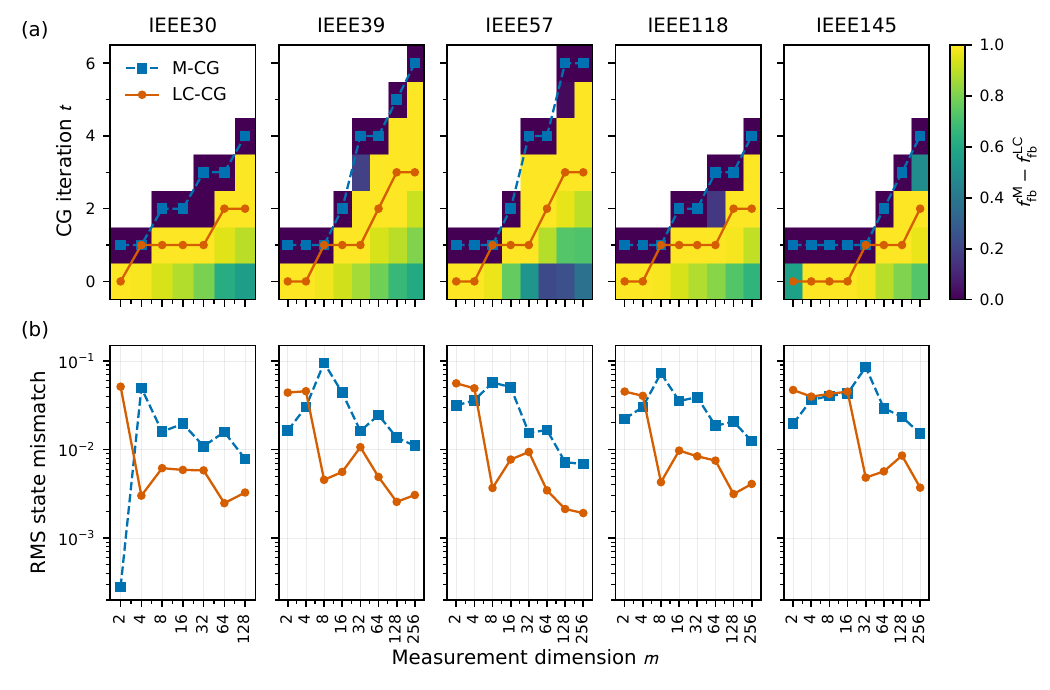}
  \caption{Network and measurement-dimension dependence at \(\eta=0.020\), with \(d_{\rm cap}=64\), 400 commissioning steps, and 600 deployment steps, for five larger IEEE test systems, shown from left to right: IEEE 30, 39, 57, 118, and 145. (a) Fallback-fraction reduction \(f_{\rm fb}^{\rm M}-f_{\rm fb}^{\rm LC}\) [Eq.~\eqref{eq:fallback_fraction}] over measurement dimension \(m\) and CG iteration \(t\), with overlaid zero-fallback frontiers for M-CG and LC-CG. (b) Normalized RMS estimator-to-reference state mismatch [Eq.~\eqref{eq:normalized_rms_state_mismatch}] at each policy's own selected zero-fallback frontier.
  }
  \label{fig:network_resolved}
\end{figure*}

Figure~\ref{fig:network_resolved} extends the frontier comparison to IEEE 30, 39, 57, 118, and 145. At the smallest measurement dimensions, LC-CG at its shallower zero-fallback frontier has a larger state mismatch than M-CG in several configurations, whereas at larger \(m\) its mismatch is smaller. A shallower zero-fallback solver depth therefore does not guarantee a smaller recursive mismatch.

Figure~\ref{fig:horizon_scaling} compares horizon dependence under fixed \(C=400\) and proportional commissioning. Across the evaluated chronology lengths, the M-CG frontier changes little under either protocol, while LC-CG retains a substantially shallower frontier without a monotonic benefit from increasing the commissioning window. The fixed-window results show that the frontier advantage persists without increasing the commissioning sample count. Differences between protocols also reflect changes in the deployment interval and commissioning-derived quantities, so they do not isolate the effect of commissioning-data quantity. The corresponding normalized RMS state mismatch remains a separate quantity: its dependence on commissioning protocol is network dependent, with particularly visible realization-to-realization variation for IEEE 57. Together with Fig.~\ref{fig:network_resolved}, these results distinguish the computational effect of learned correction from the state-space consequence of the selected executed policy.

\begin{figure*}[t]
    \centering
    \includegraphics[width=\textwidth]{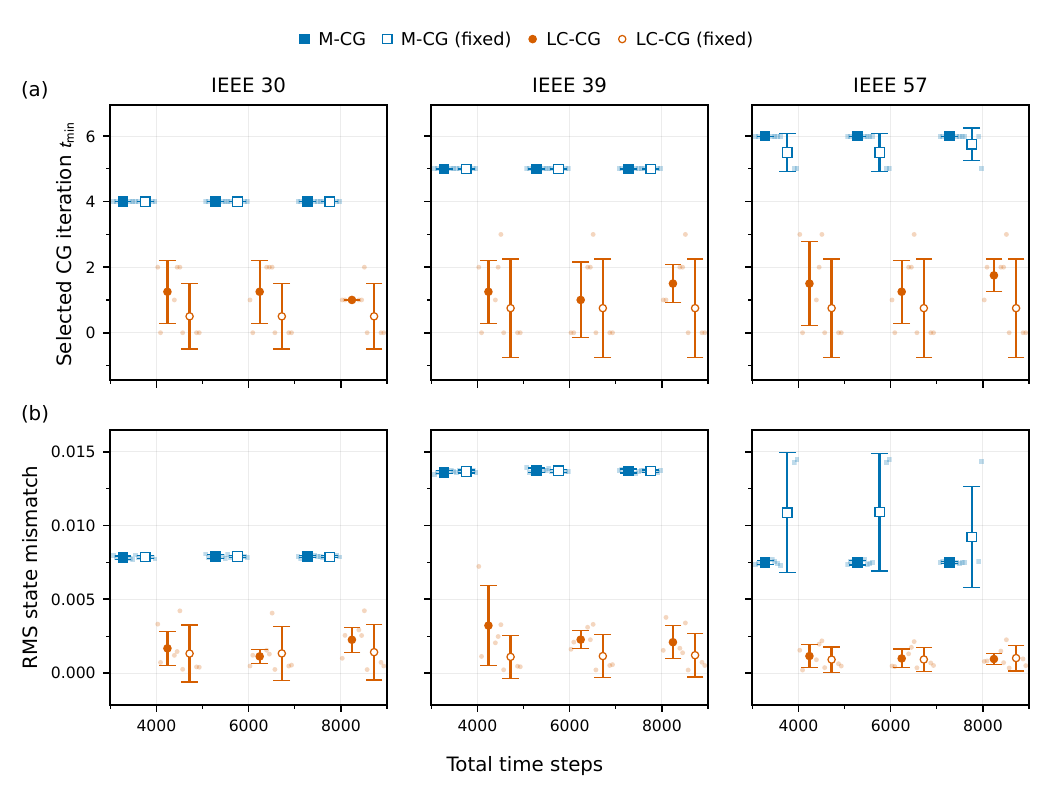}
    \caption{
    Horizon dependence under fixed and proportional commissioning for IEEE~30, IEEE~39, and IEEE~57, shown from left to right, with \(m=128\), \(\eta=0.020\), and \(d_{\rm cap}=64\), at total chronology lengths \(L\in\{4000,6000,8000\}\). Fixed commissioning uses \(C=400\), whereas proportional commissioning uses \(C=0.4L\), with deployment occupying the remaining \(L-C\) steps. (a) Each policy's selected zero-fallback CG iteration \(t_{\min}\). (b) Normalized RMS estimator-to-reference state mismatch [Eq.~\eqref{eq:normalized_rms_state_mismatch}] at each policy's own selected frontier. Filled symbols denote proportional commissioning and open symbols denote fixed commissioning. Symbols show the mean over four realizations, and error bars indicate one sample standard deviation across realizations.
    }
    \label{fig:horizon_scaling}
\end{figure*}

\subsection{Computational cost and tolerance dependence}
\label{subsec:runtime}

To determine whether displacement of the zero-fallback CG frontier yields a computational benefit after correction and verification costs are included, we compare seven execution policies on IEEE 14 with \(n=27\), \(m=64\), and 600 deployment steps following 400 commissioning steps. The comparison includes the selected M-CG and LC-CG policies, M-CG evaluated at the LC-CG iteration, warm-start CG, warm-start preconditioned conjugate gradient (PCG), and innovation- and information-form direct solves. At each admissibility tolerance, all policies process the same observation sequence and are initialized from the same commissioned state estimate and covariance, with iteration selections fixed before timing and a common residual certificate and verified fallback. Complete online runtime includes candidate generation, learner evaluation and correction where applicable, residual recomputation, certification, any invoked fallback, and the recursive update. Table~\ref{tab:runtime-complete} reports the sample mean and sample standard deviation over seven timing repetitions of the same deployment schedule, while reusable method-specific setup costs are reported separately. Implementation details and the timing protocol are given in Appendix~\ref{app:runtime_compar_proto}.

\begin{table*}[t]
\caption{
Complete online runtime and reusable method-specific setup cost for the IEEE 14 tolerance study on one realization with \(n=27\), \(m=64\), 400 commissioning steps, and 600 deployment steps. \(T_{\rm on}\) is the complete online time for one 600-step deployment. Entries \(a\pm b\) report the sample mean \(a\) and sample standard deviation \(b\) over seven repetitions of the same deployment schedule. \(T_{\rm off}\) is additional reusable method-specific setup measured once, excluding offline preparation common to all methods. Selected M-CG and LC-CG denote their respective zero-fallback frontier policies, with iteration counts \(t=(10,7,3,3)\) and \(t=(7,4,1,0)\), respectively, in increasing-tolerance order. ``M-CG at LC-CG iteration'' evaluates uncorrected M-CG at the iteration selected by LC-CG. Warm-start CG initializes each gain solve from the previous executed gain, while warm-start preconditioned conjugate gradient (PCG) additionally uses a fixed Cholesky preconditioner constructed from the final commissioning innovation covariance. The innovation-form direct method factorizes the current innovation covariance, whereas the information-form direct method uses reusable fixed-\(R\) quantities and a per-step information-form solve. Online timing includes candidate generation, correction where applicable, residual recomputation and certification, verified fallback when invoked, and the recursive update; post-step evidence capture, file input and output, and external orchestration are excluded.
}
\label{tab:runtime-complete}
\begin{ruledtabular}
\begin{tabular}{l|rr|rr|rr|rr}
Policy
& \multicolumn{2}{c|}{$\eta=10^{-3}$}
& \multicolumn{2}{c|}{$\eta=10^{-2}$}
& \multicolumn{2}{c|}{$\eta=10^{-1}$}
& \multicolumn{2}{c}{$\eta=0.2$} \\
& $T_{\rm on}$ (s) & $T_{\rm off}$ (s)
& $T_{\rm on}$ (s) & $T_{\rm off}$ (s)
& $T_{\rm on}$ (s) & $T_{\rm off}$ (s)
& $T_{\rm on}$ (s) & $T_{\rm off}$ (s) \\
\colrule
Selected M-CG
& $9.073\pm0.027$ & \textemdash
& $7.607\pm0.031$ & \textemdash
& $5.587\pm0.032$ & \textemdash
& $5.636\pm0.031$ & \textemdash \\

Selected LC-CG
& $7.780\pm0.007$ & 2.571
& $6.300\pm0.010$ & 1.640
& $4.758\pm0.008$ & 0.597
& $4.247\pm0.009$ & 0.206 \\

M-CG at LC-CG iteration
& $11.434\pm0.027$ & \textemdash
& $9.996\pm0.040$ & \textemdash
& $8.396\pm0.033$ & \textemdash
& $7.601\pm0.040$ & \textemdash \\

Warm-start CG
& $5.134\pm0.007$ & \textemdash
& $4.767\pm0.011$ & \textemdash
& $4.603\pm0.011$ & \textemdash
& $4.646\pm0.010$ & \textemdash \\

Warm-start PCG
& $5.268\pm0.024$ & 0.000506
& $4.811\pm0.021$ & 0.000504
& $4.533\pm0.026$ & 0.000494
& $4.564\pm0.018$ & 0.000480 \\

Direct, innovation form
& $3.971\pm0.022$ & \textemdash
& $4.011\pm0.024$ & \textemdash
& $3.979\pm0.029$ & \textemdash
& $4.029\pm0.024$ & \textemdash \\

Direct, information form
& $3.975\pm0.005$ & 0.000161
& $4.010\pm0.007$ & 0.000158
& $3.970\pm0.008$ & 0.000186
& $4.018\pm0.008$ & 0.000161 \\
\end{tabular}
\end{ruledtabular}
\end{table*}

Across the four tolerances, selected LC-CG has lower measured online runtime than selected M-CG, by \(14.3\%\), \(17.2\%\), \(14.8\%\), and \(24.6\%\) for \(\eta=10^{-3},10^{-2},10^{-1}\), and \(0.2\), respectively. Comparisons across tolerances also reflect the corresponding change in the admissibility requirement. The component profiles in Table~\ref{tab:runtime-decomposition} of  Appendix~\ref{app:supporting_numerical_results} attribute this reduction to lower candidate-generation cost exceeding the added learner, correction, and residual-recomputation costs. Uncorrected M-CG at the LC-CG iteration lies below its own zero-fallback frontier and is the slowest policy at every tolerance, consistent with the cost of fallback solves on rejected raw candidates; the reduction therefore derives from the correction rather than from the lower iteration count alone. At \(\eta=0.2\), LC-CG operates at \(t=0\), with no CG matrix--vector products; certification then dominates the online cost. Because the zero-start raw candidate is then the zero gain, the executed LC-CG gain is the bounded learned correction itself. The frozen ridge map acts on the RMS-scaled raw residual, which in this case is \(\rho_{k,0}^{\rm alg}=-\widetilde P_kH^{\mathsf T}\); its predicted coefficients are projected onto the correction ball before the resulting gain is assessed by the residual certificate. The runtime reduction relative to selected M-CG at \(\eta=0.2\) therefore reflects certificate-gated learned gain construction rather than repair of a nonzero CG iterate.

Warm-start CG and PCG remain faster at the first three tolerances, while LC-CG has the lower measured online runtime at \(\eta=0.2\); both direct-solve baselines are faster at all four tolerances. The LC-CG training and setup cost \(T_{\rm off}\) is incurred once and includes exact local-gain solves for correction-target generation, training CG solves and residual construction, correction-basis construction by PCA, correction-budget calibration, and ridge fitting. Common preparation and frontier selection are excluded. Assuming reuse without retraining and that the measured mean online cost per deployment remains representative, the additional offline cost relative to selected M-CG is recovered by the end of the second 600-step deployment at \(\eta=10^{-3}\) and \(10^{-2}\), and by the end of the first deployment at \(\eta=10^{-1}\) and \(0.2\). A crossover with direct solvers may occur for larger or repeatedly changing gain systems, particularly for sparse or matrix-free problems where factorization reuse is limited and shallow Krylov iteration plus correction and certification costs less than a fresh factorization.

\subsection{Quantum solver candidates}
\label{subsec:quantu_solve_candi}

This section examines how errors in quantum-generated gain candidates affect a classical recursive estimator. Complete gains reconstructed using VQLS or decoded from annealer-sampled QUBOs enter through the represented-gain interface of Sec.~\ref{subsec:solver_residual_interface}. Incomplete optimization, finite sampling, hardware imperfections, and encoding resolution can affect candidate quality. Bounded correction, residual checking, and verified fallback determine the executed gain, whose defect affects both the current update and the covariance governing subsequent gains. These operations and the residual--drift response evaluation are performed classically. All realizations use the IEEE 5 benchmark with \(n=9\) and \(m=2\). The 40-step commissioning interval comprises 10 steps for correction-basis construction, 10 for correction-radius calibration, and 20 for learner fitting, followed by 60 deployment steps.

\subsubsection{VQLS execution across simulated and hardware backends}

The Statevector realization uses \(\eta=0.01\), whereas the finite-shot Sampler and Origin realizations use \(\eta=0.05\). The Statevector and finite-shot Sampler realizations use a one-parameter one-qubit ansatz and the global VQLS objective; optimizer, sampling, and reconstruction settings are given in Appendix~\ref{app:vqls_reconstr_and_solver_diagn}.

We also realize the \(m=2\) VQLS construction using an Origin Wukong superconducting quantum computer. The variational parameters are optimized locally using the finite-shot Sampler and then frozen, while the terminal \(X\)- and \(Z\)-measurements used to reconstruct each complete gain are acquired on the Origin hardware. The Statevector, finite-shot Sampler, and Origin realizations are configured separately, with different tolerances and sampling conditions, and are therefore not matched backend comparisons.

\begin{figure*}[t]
  \centering
  \includegraphics[width=\textwidth]{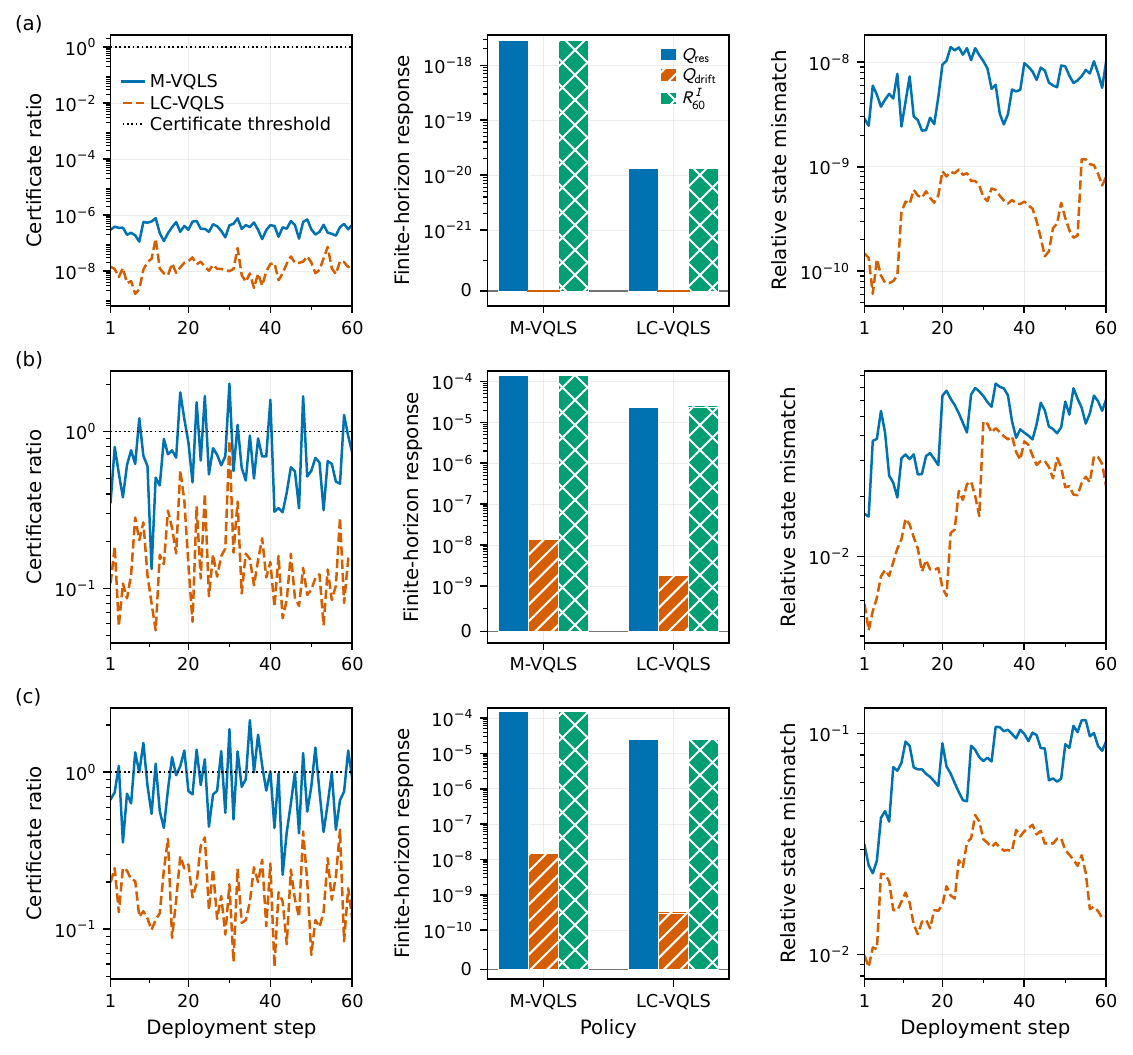}
  \caption{Complete-gain recursive execution of \(m=2\) variational quantum linear solver (VQLS) candidates on IEEE 5 with \(n=9\), using 40 commissioning steps followed by a 60-step deployment window. Rows show three separately configured realizations: (a) Statevector realization with \(\eta=0.01\), (b) finite-shot Sampler realization with \(\eta=0.05\), and (c) a hardware realization executed on the Origin quantum computer, with \(\eta=0.05\). In each row, the left panel shows the represented-residual certificate ratio for monitored VQLS (M-VQLS) and learned-corrected VQLS (LC-VQLS); the dotted line marks the acceptance threshold at unity. In the figure, \(Q_{\rm res}\), \(Q_{\rm drift}\), and \(R_{60}^{I}\) denote \(\mathcal Q_{60}^{\rm res}\), \(\mathcal Q_{60}^{\rm drift}\), and \(\mathcal R_{60}^{I}\), respectively. The middle panel reports these finite-horizon quantities, related by \(R_{60}^{I}=Q_{\rm res}-Q_{\rm drift}\), with the nonnegative drift contribution entering the response subtractively. The right panel shows the per-step relative state mismatch [Eq.~\eqref{eq:relative_state_mismatch}].}
  \label{fig:vqls_m2_execution_response}
\end{figure*}

Figure~\ref{fig:vqls_m2_execution_response} traces the complete execution chain from reconstructed VQLS gains through represented-residual certification and fallback to finite-horizon response and estimator-to-reference state mismatch. For the Statevector realization, both M-VQLS and LC-VQLS satisfy the residual certificate at every deployment step, providing a zero-fallback comparison in which learned correction reduces both the reported finite-horizon response and the realized state mismatch.

For the finite-shot Sampler realization, M-VQLS certifies 50 of 60 raw candidates and uses verified fallback on the remaining 10 steps. Along the LC-VQLS trajectory, 51 of the 60 raw candidates satisfy the certificate; bounded correction brings the remaining nine within the acceptance criterion, so all 60 corrected candidates execute without fallback. The resulting LC-VQLS policy has lower finite-horizon response and estimator-to-reference state mismatch than the executed M-VQLS policy.

For the Origin realization, 41 of the 60 M-VQLS raw candidates satisfy the certificate, with verified fallback used on the remaining 19 steps. Along the LC-VQLS trajectory, 38 raw candidates satisfy the certificate; after bounded correction, all 60 candidates satisfy it, corresponding to 22 certificate rescues and no fallback for LC-VQLS. The executed LC-VQLS policy also yields a lower finite-horizon response and estimator-to-reference state mismatch than M-VQLS.

Solver-level residual and timing diagnostics are reported in Appendix~\ref{app:vqls_residual_quality_and_solver_cost}. Over the evaluated innovation dimensions, the VQLS realizations incur substantially greater solver-level cost than CG and provide no evidence of a computational advantage; their role is to demonstrate solver-interface compatibility within the correction, certification, fallback, and recursive-response framework.

\subsubsection{QUBO execution on quantum annealing hardware}
\label{subsubsec:qubo_dwave}

We evaluate the same represented-gain interface using QUBO-generated candidates. Each innovation-space linear system is encoded with the binary construction of Appendix~\ref{app:binary_encod_of_the_qubo} at bit depth \(n_{\rm bit}=4\), sampled on a D-Wave Advantage2 quantum annealer, decoded into an explicit solution vector, and assembled across all \(n=9\) right-hand sides into a complete \(9\times2\) gain; the admissibility parameter is \(\eta=0.05\). Each QUBO then has \(mn_{\rm bit}=8\) logical binary variables. A decoded solution lies on the encoding grid and need not coincide with the continuous solution of the gain system even when the returned sample minimizes the QUBO objective; both this discretization error and any sampling suboptimality enter the raw defect \(A_k\) and are judged only through the recomputed represented residual. The commissioning and deployment gains are obtained from the D-Wave quantum processing unit samples, and the LC-QUBO corrector is fitted from the commissioning gains and frozen before deployment. Both policies execute throughout the deployment window without fallback.

\begin{figure}[t]
  \centering
  \includegraphics[width=\columnwidth]{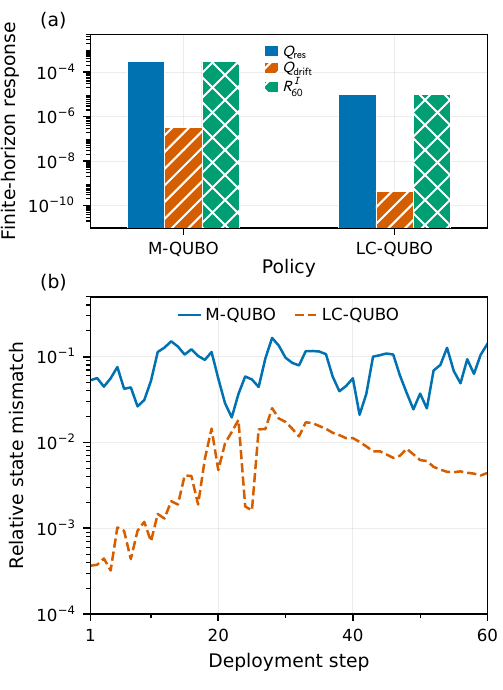}
  \caption{D-Wave execution of quadratic unconstrained binary optimization (QUBO)-generated gain candidates for recursive IEEE 5 estimation at \(m=2\), \(n_{\rm bit}=4\), and \(\eta=0.05\). Both policies execute without fallback. (a) Finite-horizon response.
  (b) Per-step relative state mismatch [Eq.~\eqref{eq:relative_state_mismatch}]. The LC-QUBO policy has lower finite-horizon response and realized state mismatch than M-QUBO.}
  \label{fig:ieee5_m2_dwave}
\end{figure}

Figure~\ref{fig:ieee5_m2_dwave} shows the resulting recursive consequence under matched zero-fallback execution. The identity-task response decreases from
\(\mathcal R_{60}^{I}=2.99\times10^{-4}\) for M-QUBO to
\(1.04\times10^{-5}\) for LC-QUBO, a factor of \(28.8\). As a deployment-level summary of the trajectories in Fig.~\ref{fig:ieee5_m2_dwave}(b), the normalized RMS estimator-to-reference state mismatch decreases from
\(8.53\times10^{-2}\) to \(9.52\times10^{-3}\), a factor of \(9.0\).
Because neither policy invokes fallback at this operating point, the comparison is between their executed QUBO-derived policies rather than identical raw candidates: after the initially shared gain system, M-QUBO and LC-QUBO evolve on their own recursive gain systems. Together with the VQLS results, this shows that the correction, certification, and finite-horizon response framework can be applied to complete gains reconstructed from two distinct quantum candidate-generation mechanisms.

\section{Discussion and Conclusions}
\label{sec:discu_and_conclu}

This paper connects two questions that arise when approximate computation is embedded in recursive estimation: what part of a computational defect can be repaired before a gain is executed, and what the executed defect does to the recursion afterward. Local repairability characterizes whether the declared correction class contains an action that brings the current defect within the local admissibility tolerance; the bounded learned corrector supplies one particular proposal within that class, and the independent residual criterion, with verified fallback, determines what is executed and bounds the executed defect by the same tolerance. Transferring this bound to the theoretical executed defect requires separate control of matrix-formation and covariance-propagation errors. The response theory then takes the executed defects as its input. At leading quadratic order, locally recentered computational defects contribute a fixed-reference exposure without resolving the feedback induced by their effect on later covariances and gains. The signed-quartic expansion identifies the first nonlinear recursive correction: innovation-covariance inflation increases the response, whereas local-gain reoptimization reduces it. Their competition determines the leading departure from the quadratic prediction, while the sixth-order remainder controls the unresolved higher-order contribution at fixed horizon.

The present theory isolates computational approximation under the stated fixed linear-Gaussian model and matched initialization. Changes in the physical model, network topology, measurement configuration, or covariance specification remain outside its scope. Such changes affect how a gain candidate is obtained, but not the represented-gain interface considered here: a compatible precomputed, reused, updated, or newly computed gain can be evaluated through the same residual interface and executable certificate. The finite-horizon response then depends on the defect of the gain actually executed rather than on the numerical route by which that candidate was generated.

The repairability analysis separates two distinct obstructions on local correction: unavailable correction directions and insufficient correction budget. Enlarging a nested correction subspace may make a previously subspace- or budget-obstructed defect repairable at the same correction radius. In the full gain space, only the budget obstruction remains. These geometric obstructions characterize what the declared correction class can achieve, not whether the learned proposal attains an admissible action or whether that action passes the residual certificate. Because correction acts on the represented candidate and certification on its recomputed residual, the execution framework is not tied to zero-start CG. Preconditioned, warm-started, or otherwise generated candidates can be evaluated through the same construction. The learned map evaluated in the classical CG study, however, is fitted specifically to zero-start CG defects. Whether a corrector trained on candidates from other numerical schemes produces additional certified frontier displacement beyond those schemes alone remains an open question.

The represented-gain interface is not tied to a particular quantum hardware modality. The realizations reported here use a gate-model superconducting processor and a quantum annealer. In the Origin realization, finite-shot Sampler optimization provides the VQLS parameters, while only the final \(X\)- and \(Z\)-measurements used for complete-gain reconstruction are acquired on superconducting hardware. For \(m=2\), the one-qubit circuit does not probe multiqubit connectivity or entangling-gate costs; these enter at larger innovation dimension~\cite{bravo2023variational}, motivating cross-platform studies on other architectures such as the trapped-ion JION system at Jülich~\cite{fzj2026jion}, which features an all-to-all connectivity between qubits. Hybrid quantum--classical learning or quantum sampling could also be used to propose corrections~\cite{ji2026quantumdeep}. Repairability remains determined by the raw defect and declared correction class, while acceptance is based on the recomputed residual. Potential benefits can therefore be assessed through certified recursive execution, including improved attainment or acceptance, lower complete cost, or lower finite-horizon response relative to optimized classical alternatives. More broadly, separating candidate generation from independently certified recursive execution provides a practical framework for investigating hybrid quantum--classical computation under common estimator-level criteria.

Beyond the power-grid-derived setting studied here, related computation--recursion constraints arise in two broader regimes. In large-scale geophysical data assimilation, the storage and linear-algebra costs of Kalman-type updates motivate Krylov and low-rank approximations, including conjugate-gradient and Lanczos constructions \cite{BardsleyEtAl2013,Freitag2020,LeProvost2022LowRankEnKF}. In hard-real-time feedback, latency and numerical-precision constraints directly shape recursive-estimator implementation. Predictive linear-quadratic-Gaussian control has been demonstrated in adaptive optics at kilohertz update rates \cite{poyneer2023lqg}, while real-time Kalman estimation of levitated mechanical motion has been implemented on a field-programmable gate array using fixed-point arithmetic \cite{setter2018realtime} and extended to feedback control in the quantum regime \cite{magrini2021realtime}. These settings illustrate two distinct sources of computation-limited recursive inference: scale-driven approximate linear algebra and latency- or precision-limited execution. Extending the present framework to low-rank or ensemble covariance representations, time-varying models, and finite-precision implementations is therefore a natural direction for future work.

More generally, when approximate computation is embedded in recursive estimation, the design question is not simply how to compute faster or how to learn better, but how a local computational defect changes the subsequent recursion, what part of that defect can be repaired before execution, and whether the resulting update satisfies the declared criterion for execution. The finite-horizon response theory makes the recursive consequence quantitative, while the repairability and residual-certificate analyses characterize the possibilities for correction and execution. Once executed, computational error is no longer merely a local numerical perturbation; it becomes part of the recursive inference dynamics.

\begin{acknowledgments}

Y.J. and D.W. acknowledge support from the project Entwicklungspartnerschaft Ionenfallen-Quantencomputer in NRW (EPIQ), funded by the Ministerium für Kultur und Wissenschaft des Landes Nordrhein-Westfalen (MKW NRW).
D.W. acknowledges support from the project QEC4QEA that has received funding from the European High Performance Computing Joint Undertaking (EuroHPC JU) under Grant Agreement No 101194322, and from the project Q-Neko that has received funding from the European High Performance Computing Joint Undertaking (EuroHPC JU) under Grant Agreement No 101241875, and from the project Jülich UNified Infrastructure for Quantum Computing (JUNIQ) that has received funding from the German Federal Ministry of Education and Research (BMBF) and the Ministry of Culture and Science of the State of North Rhine-Westphalia.
The authors gratefully acknowledge the Jülich Supercomputing Centre (\url{https://www.fz-juelich.de/ias/jsc}) for funding this project by providing computing time on the D-Wave Advantage™ System JUPSI through the Jülich UNified Infrastructure for Quantum computing (JUNIQ).
The authors gratefully acknowledge computing time on the supercomputer JURECA~\cite{JURECA} at Forschungszentrum Jülich under grant no.~eqip.
The authors also acknowledge Origin Quantum Computing Technology for access to the Origin Quantum Cloud platform and the Origin Wukong superconducting quantum computing resources.
Part of this work was supported by the project “Quantum-based Energy Grids (QuGrids),” funded under the programme “Profilbildung 2022,” an initiative of the Ministry of Culture and Science of the State of North Rhine-Westphalia, and by the project QSolid (Grant No. 13N16149), funded by the BMBF within the funding programme “Quantum Technologies – From Basic Research to Market.”

\end{acknowledgments}

\appendix

\section{Notation and Covariance Identities}
\label{app:statistical_covariance_identities}

\subsection{Notation table}
\label{app:notation}

Table~\ref{tab:notation_kalman_filter} summarizes the notation used throughout this paper.

\begin{table*}[tb]\small
\centering
\caption{\label{tab:notation_kalman_filter} Principal notation for recursive estimation, constrained correction, residual certification, and finite-horizon response.}
\begin{ruledtabular}
\begin{tabular}{ll@{\qquad}ll}
Symbol & Meaning & Symbol & Meaning \\
\hline
$x_k \in \mathbb{R}^n$ & Hidden state vector at time $k$ & $E_k$ & Executed local computational defect \\
$z_k \in \mathbb{R}^m$ & Measurement vector at time $k$ & $\Gamma_k$ & Local-reference gain drift \\
$\lambda_{\max/\min}(A)$ & Extremal eigenvalues of symmetric $A$ & $\Theta_t$ & No-intercept ridge map at checkpoint $t$ \\
$\bar{u}_{k-1} \in \mathbb{R}^p$ & Deterministic input on transition $k-1 \to k$ & $\rho_k(L)$ & Residual of gain $L$ \\
$\xi_{k-1} \in \mathbb{R}^p$ & Input perturbation, $\xi_{k-1} \sim \mathcal{N}(0,V_{k-1})$ & $\rho_k^{\mathrm{alg}}$ & Raw solver residual \\
$F \in \mathbb{R}^{n \times n}$ & State transition matrix & $\rho_k^{\mathrm{cand}}$ & Corrected-candidate residual \\
$H \in \mathbb{R}^{m \times n}$ & Measurement (observation) matrix & $\rho_k^{\mathrm{exec}}$ & Executed residual \\
$B \in \mathbb{R}^{n \times p}$ & Input matrix & $\phi_{k,t}$ & Residual feature vector at checkpoint $t$ \\
$Q_{k-1} \in \mathbb{R}^{n \times n}$ & Process noise covariance, $k-1 \to k$ & $R_t^{\rm feat}$ & Frozen diagonal matrix of feature scales \\
$R_k \in \mathbb{R}^{m \times m}$ & Measurement noise covariance at time $k$ & $\kappa(S_k)$ & Spectral condition number $\sigma_{\max}(S_k)/\sigma_{\min}(S_k)$ \\
$V_{k-1} \in \mathbb{R}^{p \times p}$ & Input-error covariance, $k-1 \to k$ & $\delta_c \ge 0$ & Correction-budget radius \\
$\widetilde{x}_k^\star,\widetilde{P}_k^\star$ & Reference predicted state and covariance & $\mathcal V_{\rm corr}$ & Corrector subspace in gain space \\
$\widehat{x}_k^\star,\widehat{P}_k^\star$ & Reference posterior state and covariance & $\Pi_{\rm corr}$ & Orthogonal projector onto $\mathcal V_{\rm corr}$ \\
$r_k^\star,S_k^\star$ & Reference innovation and covariance & $A_k^\parallel$ & Projection of $A_k$ onto $\mathcal V_{\rm corr}$ \\
$K_k^\star$ & Reference Kalman gain & $A_k^\perp$ & Component of $A_k$ outside $\mathcal V_{\rm corr}$ \\
$\widetilde{x}_k,\widetilde{P}_k$ & Implemented predicted state and covariance & $d$ & Dimension of the corrector subspace $\mathcal V_{\rm corr}$ \\
$\widehat{x}_k,\widehat{P}_k$ & Implemented posterior state and covariance & $\varepsilon_k$ & Minimum postcorrection defect norm \\
$r_k,S_k$ & Implemented innovation and covariance &  $\delta_{\rm adm}$ & Local gain-defect tolerance \\
$K_k^{\mathrm{loc}}$ & Exact gain for \(\widetilde P_k\) & $n$ & State dimension \\
$K_k^{\mathrm{alg}}$ & Raw solver output & $m$ & Measurement dimension \\
$K_k^{\mathrm{cand}}$ & Corrected candidate gain & $p$ & Input dimension \\
$\widetilde{K}_k$ & Executed gain & $\sigma_{\max / \min}(A)$ & Maximal / minimal singular values of $A$ \\
$\Delta_k^{\mathrm{ml}}$ & Candidate correction & $A \succ 0$ & $A$ is positive definite \\
$A_k$ & Raw local computational defect & $A \succeq 0$ & $A$ is positive semidefinite \\
$\ell_k$ & Positive lower bound on $\sigma_{\min}(S_k)$ & $t$ & CG checkpoint or iteration budget \\
$\tau$ & Finite-horizon task-response budget & $\mathcal G$ & Execution environment sigma-field \\
$\mathcal Y_k$ & Measurement information available up to time \(k\) & $\mathcal F_k^-$ & Information available before incorporating \(z_k\) \\
$\mathcal F_k^+$ & Information available after incorporating \(z_k\) & $\sigma(X_1,\ldots,X_j)$
& Sigma-field generated by $X_1,\ldots,X_j$\\
$w_{k-1}$ & Process noise & $v_k$ & Measurement noise \\
$\mu_0$ & Initial state mean & $P_0$ & Initial state covariance \\
$\underline r$ & Uniform measurement-covariance lower bound & $T$ & Finite horizon \\
$\mathcal J_k(P,K)$ & Joseph covariance map & $\mathcal C_k$ & Bounded correction class \\
$d_{\rm cap}$ & Prescribed correction-dimension cap & $\nu_t$ & Numerical rank of checkpoint-$t$ defect samples \\
$d_{{\rm eff},t}$ & Supported correction dimension & $U_{t,i}$ & Correction-basis matrix at checkpoint $t$ \\
$\delta_{c,t}^{\star}$ & Calibrated checkpoint-local correction radius & $\widehat c_{k,t}$ & Projected predicted correction coefficients \\
$\Xi_k$ & Posterior covariance mismatch & $e_k$ & Estimator-to-reference state difference \\
$\Phi_k^\star$ & Reference closed-loop matrix & $\Phi_k^{\rm loc}$ & Local closed-loop matrix \\
$M_k$ & Positive-semidefinite task metric & $\Psi_{t,j}^\star$ & Reference propagator from $j$ to $t$ \\
$W_{j,T}^M$ & Finite-horizon task-response operator & $\mathcal R_T^M$ & Task-weighted covariance response \\
$\mathcal Q_T^{\rm res}$ & Accumulated residual contribution & $\mathcal Q_T^{\rm drift}$ & Subtractive gain-drift contribution \\
$\delta$ & Maximum executed defect norm over $1{:}T$ & $\bar\delta$ & Fixed upper bound on $\delta$ \\
$\mathcal B_T^M(E)$ & Fixed-reference quadratic response & $\Omega_k$ & Leading quadratic covariance response \\
$\mathcal C_T^{M,\mathrm{infl}}(E)$ & Quartic innovation-inflation contribution & $\mathcal C_T^{M,\mathrm{reopt}}(E)$ & Quartic local-gain-reoptimization contribution \\
$C_{6,T}^M(\bar\delta)$ & Uniform sixth-order remainder constant & $\delta_{{\rm adm},k}^{(\tau)}$ & Local tolerance for allocation $b_{k,T}$ \\
$q_{k,T}$ & Exact stepwise residual contribution & $b_{k,T}$ & Prescribed stepwise task-response allocation \\
$c_{k,T}^{W,\ell}$ & Task-aware upper bound on $q_{k,T}$ & $c_{k,T}^{\rm rad}$ & Radial upper bound on $q_{k,T}$ \\
$b_{k,T}^{\rm adm}$ & Radial allocation induced by $\delta_{\rm adm}$ & $\tau_T^{\rm adm}$ & Induced finite-horizon response bound \\
\end{tabular}
\end{ruledtabular}
\end{table*}

\subsection{Joseph identities}
\label{app:joseph_identities}

For \(P\succeq0\), define the innovation covariance and the exact Kalman gain associated with \(P\) by
\begin{equation}
	S_k(P):=HPH^{\mathsf T}+R_k,
	\qquad
	K_k(P):=PH^{\mathsf T}S_k(P)^{-1}.
	\label{eq:appendix_gain_def}
\end{equation}
Because \(R_k\succeq\underline r I_m\) with \(\underline r>0\), we have \(S_k(P)\succ0\), so \(K_k(P)\) is well defined. For an arbitrary gain \(L\in\R^{n\times m}\), expanding the Joseph map gives
\[
\mathcal{J}_k(P,L)
=
P-LHP-PH^{\mathsf T}L^{\mathsf T}
+LS_k(P)L^{\mathsf T}.
\]
Using
\[
K_k(P)S_k(P)=PH^{\mathsf T},
\qquad
S_k(P)K_k(P)^{\mathsf T}=HP,
\]
and completing the square yields the exact identity
\begin{align}
	\mathcal{J}_k(P,L)
	={}&
	\mathcal{J}_k(P,K_k(P))
	\nonumber\\
	&+
	\bigl(L-K_k(P)\bigr)
	S_k(P)
	\bigl(L-K_k(P)\bigr)^{\mathsf T}.
	\label{eq:appendix_joseph_identity}
\end{align}
At the exact Kalman gain, the Joseph update reduces to the standard Kalman covariance update. Indeed,
\[
\begin{aligned}
	\mathcal{J}_k(P,K_k(P))
	&=
	P-K_k(P)HP-PH^{\mathsf T}K_k(P)^{\mathsf T}\\
	&\qquad+K_k(P)S_k(P)K_k(P)^{\mathsf T}\\
	&=
	P-K_k(P)HP\\
	&=
	P-PH^{\mathsf T}S_k(P)^{-1}HP.
\end{aligned}
\]
Hence
\begin{equation}
	\mathcal{J}_k(P,K_k(P))
	=
	(I-K_k(P)H)P
	=
	P-PH^{\mathsf T}S_k(P)^{-1}HP.
	\label{eq:appendix_riccati_form}
\end{equation}
The simplification in Eq.~\eqref{eq:appendix_riccati_form} relies on the exact-gain relation \(K_k(P)S_k(P)=PH^{\mathsf T}\) and does not hold for an arbitrary gain. In particular, setting \(L=K_k(P)+E\), with \(E\in\R^{n\times m}\), in Eq.~\eqref{eq:appendix_joseph_identity} gives
\begin{equation}
	\mathcal{J}_k(P,K_k(P)+E)-\mathcal{J}_k(P,K_k(P))
	=
	ES_k(P)E^{\mathsf T}
	\succeq0.
	\label{eq:appendix_joseph_perturbation}
\end{equation}

\subsection{Covariance mismatch and gain drift}
\label{app:covariance_mismatch_and_gain_drift}

For the current implemented prediction covariance, define the exact same-prior posterior covariance
\begin{equation}
	\widehat P_k^{\rm loc}
	:=
	\mathcal{J}_k(\widetilde P_k,K_k^{\rm loc}).
	\label{eq:local_cov}
\end{equation}
Equation~\eqref{eq:appendix_riccati_form} gives
\[
\widehat P_k^{\rm loc}
=
(I-K_k^{\rm loc}H)\widetilde P_k,
\]
while Eq.~\eqref{eq:appendix_joseph_perturbation}, with \(P=\widetilde P_k\) and \(E=E_k\), gives
\begin{equation}
	\widehat P_k-\widehat P_k^{\rm loc}
	=
	E_kS_kE_k^{\mathsf T}
	\succeq0.
	\label{eq:local_executed_covariance_injection}
\end{equation}
Consequently,
\[
\|\widehat P_k-\widehat P_k^{\rm loc}\|_F
\le
\|S_k\|_2\|E_k\|_F^2.
\]
Equation~\eqref{eq:local_executed_covariance_injection} is a same-prior identity; it does not imply that the global covariance mismatch \(\widehat P_k-\widehat P_k^\star\) equals \(E_kS_kE_k^{\mathsf T}\).

For the prediction-covariance mismatch, define
\begin{equation}
\Delta P_k^-
:=
\widetilde P_k-\widetilde P_k^\star
=
F\Xi_{k-1}F^{\mathsf T}.
\label{eq:prior_covariance_mismatch}
\end{equation}
Using
\begin{gather*}
K_k^{\rm loc}S_k=\widetilde P_kH^{\mathsf T},
\qquad
K_k^\star S_k^\star=\widetilde P_k^\star H^{\mathsf T},
\\
S_k-S_k^\star=H\Delta P_k^-H^{\mathsf T},
\end{gather*}
the gain drift satisfies
\begin{align*}
\Gamma_kS_k
&=
(I-K_k^\star H)\Delta P_k^-H^{\mathsf T},
\\
\Gamma_kS_k^\star
&=
(I-K_k^{\rm loc}H)\Delta P_k^-H^{\mathsf T}.
\end{align*}
Equivalently,
\[
\Gamma_k
=
\Phi_k^\star\Xi_{k-1}(HF)^{\mathsf T}S_k^{-1}
=
\Phi_k^{\rm loc}\Xi_{k-1}(HF)^{\mathsf T}(S_k^\star)^{-1},
\]
whose first representation is Eq.~\eqref{eq:gain_drift_exact} in the main text.

\begin{proof}[Proof of Proposition~\ref{prop:exact_covariance_recursions}]
Write \(\widehat{P}_k = \mathcal{J}_k(\widetilde{P}_k, \widetilde{K}_k)\) and \(\widehat{P}_k^\star = \mathcal{J}_k(\widetilde{P}_k^\star, K_k^\star)\). Insert the exact local gain to obtain
\begin{align}
    \Xi_k
    &= \bigl[
        \mathcal{J}_k(\widetilde{P}_k, \widetilde{K}_k)
        - \mathcal{J}_k(\widetilde{P}_k, K_k^{\mathrm{loc}})
    \bigr]
    \nonumber
    \\
    &\quad+
    \bigl[
        \mathcal{J}_k(\widetilde{P}_k, K_k^{\mathrm{loc}})
        - \mathcal{J}_k(\widetilde{P}_k^\star, K_k^{\mathrm{loc}})
    \bigr]
    \nonumber
    \\
    &\quad+
    \bigl[
        \mathcal{J}_k(\widetilde{P}_k^\star, K_k^{\mathrm{loc}})
        - \mathcal{J}_k(\widetilde{P}_k^\star, K_k^\star)
    \bigr].
    \label{eq:split_III}
\end{align}
The first term equals \(E_kS_kE_k^{\mathsf T}\) by Eq.~\eqref{eq:appendix_joseph_perturbation}. For the second term the gain is fixed, so the measurement-noise terms cancel and Eq.~\eqref{eq:prior_covariance_mismatch} gives
\[
\mathcal{J}_k(\widetilde P_k,K_k^{\rm loc})
-
\mathcal{J}_k(\widetilde P_k^\star,K_k^{\rm loc})
=
\Phi_k^{\rm loc}\Xi_{k-1}(\Phi_k^{\rm loc})^{\mathsf T}.
\]
For the third term, \(K_k^\star\) is the exact gain for \(\widetilde P_k^\star\), while \(K_k^{\rm loc}=K_k^\star+\Gamma_k\). Equation~\eqref{eq:appendix_joseph_perturbation} therefore gives
\[
\mathcal{J}_k(\widetilde P_k^\star,K_k^{\rm loc})
-
\mathcal{J}_k(\widetilde P_k^\star,K_k^\star)
=
\Gamma_kS_k^\star\Gamma_k^{\mathsf T}.
\]
Summing these terms proves the local-centered identity of Eq.~\eqref{eq:covariance_mismatch_recursions}.

For the reference-centered form, add and subtract the Joseph covariance obtained by applying the reference gain \(K_k^\star\) to the implemented prior covariance \(\widetilde P_k\). Since $K_k^{\mathrm{loc}}$ is optimal for $\widetilde P_k$ and $K_k^\star=K_k^{\mathrm{loc}}-\Gamma_k$, Eq.~\eqref{eq:appendix_joseph_identity} gives
\[
\mathcal{J}_k(\widetilde P_k,K_k^\star)
=
\mathcal{J}_k(\widetilde P_k,K_k^{\mathrm{loc}})
+\Gamma_kS_k\Gamma_k^{\mathsf T}.
\]
For fixed gain $K_k^\star$,
\[
\mathcal{J}_k(\widetilde P_k,K_k^\star)
-
\mathcal{J}_k(\widetilde P_k^\star,K_k^\star)
=
\Phi_k^\star\Xi_{k-1}(\Phi_k^\star)^{\mathsf T}.
\]
Therefore
\[
\mathcal{J}_k(\widetilde P_k,K_k^{\mathrm{loc}})
-
\mathcal{J}_k(\widetilde P_k^\star,K_k^\star)
=
\Phi_k^\star\Xi_{k-1}(\Phi_k^\star)^{\mathsf T}
-
\Gamma_kS_k\Gamma_k^{\mathsf T}.
\]
Adding the executed same-prior injection $E_kS_kE_k^{\mathsf T}$ proves Eq.~\eqref{eq:covariance_mismatch_recursions}.
In the local-centered identity, all three terms are positive semidefinite whenever \(\Xi_{k-1}\succeq0\); induction from \(\Xi_0\succeq0\) proves the final assertion.
\end{proof}

\section{Finite-Horizon Response Analysis}\label{app:proofs}
\label{app:finite_horizon_respon_analy}

\subsection{Exact residual-response identity}
\label{app:exact_residual_accounting}

\begin{proof}[Residual--drift decomposition]
Since matched initialization gives \(\Xi_0=0\), unrolling the reference-centered recursion~\eqref{eq:covariance_mismatch_recursions} gives
\[
\Xi_t
=
\sum_{j=1}^{t}
\Psi_{t,j}^\star
\left(
E_jS_jE_j^{\mathsf T}
-
\Gamma_jS_j\Gamma_j^{\mathsf T}
\right)
(\Psi_{t,j}^\star)^{\mathsf T}.
\]
By Proposition~\ref{prop:exact_covariance_recursions}, matched initialization implies \(\Xi_k\succeq0.\)
Moreover, unrolling the reference-centered recursion and using
\(\Gamma_jS_j\Gamma_j^{\mathsf T}\succeq0\) gives
\[
0\preceq\Xi_k
\preceq
\sum_{j=1}^{k}
\Psi_{k,j}^\star
E_jS_jE_j^{\mathsf T}
(\Psi_{k,j}^\star)^{\mathsf T}.
\]
By the definition~\eqref{eq:task_risk_definition},
\[
\mathcal R_T^M
=
\sum_{t=1}^{T}
\operatorname{tr}(M_t\Xi_t).
\]
Substitution of the preceding representation, interchange of the two finite sums, and cyclic invariance of the trace give the pathwise identity
\[
\mathcal R_T^M
=
\sum_{j=1}^{T}
\operatorname{tr}\!\left[
W_{j,T}^M
E_jS_jE_j^{\mathsf T}
\right] -
\sum_{j=1}^{T}
\operatorname{tr}\!\left[
W_{j,T}^M
\Gamma_jS_j\Gamma_j^{\mathsf T}
\right].
\]
Since \(\rho_j^{\rm exec}=E_jS_j\) and \(S_j\succ0\),
\[
E_jS_jE_j^{\mathsf T}
=
\rho_j^{\rm exec}
S_j^{-1}
(\rho_j^{\rm exec})^{\mathsf T}.
\]
Substitution yields Eq.~\eqref{eq:exact_residual_risk_accounting} in the main text. Each gain-drift term is nonnegative because \(W_{j,T}^M\succeq0\) and \(S_j\succ0\), which gives Eq.~\eqref{eq:residual_risk_slack}.
\end{proof}

\subsection{Proof of the signed-quartic response}
\label{app:proof_of_quart_finit_horiz}

\begin{proof}[Proof of Theorem~\ref{thm:quartic_finit_horiz_respon}]
The argument follows from the exact covariance identities and uses only deterministic fixed-horizon recursive bounds.

Because $\Xi_0=0$ and Proposition~\ref{prop:exact_covariance_recursions} gives $\Xi_k\succeq0$, one has
\begin{equation}
\|\Xi_k\|_F\le \tr(\Xi_k).
\label{eq:quartic_psd_trace_bound}
\end{equation}
Moreover, Eq.~\eqref{eq:prior_covariance_mismatch} together with the innovation-covariance definitions in Eqs.~\eqref{eq:reference_recursion} and~\eqref{eq:executed_recursion} implies
\begin{equation}
S_k-S_k^\star
=(HF)\Xi_{k-1}(HF)^{\mathsf T}\succeq0,
\label{eq:quartic_innovation_difference}
\end{equation}
and hence
\begin{equation}
\|S_k\|_2
\le
\|S_k^\star\|_2
+
\|HF\|_2^2\tr(\Xi_{k-1}).
\label{eq:quartic_innovation_norm_bound}
\end{equation}

Taking the trace of the reference-centered recursion~\eqref{eq:covariance_mismatch_recursions}, dropping its negative-semidefinite gain-drift term, and using $\|E_k\|_F\le\delta$ gives
\begin{align}
\tr(\Xi_k)
&\le
\|\Phi_k^\star\|_2^2\tr(\Xi_{k-1})
+
\|S_k\|_2\|E_k\|_F^2
\nonumber\\
&\le
\left(\|\Phi_k^\star\|_2^2+\|HF\|_2^2\delta^2\right)
\tr(\Xi_{k-1})
+
\|S_k^\star\|_2\delta^2.
\label{eq:quartic_trace_recursion}
\end{align}
For a fixed cap $\bar\delta$, define the finite recursive constants
\begin{align}
C_{\Xi,0}(\bar\delta)&:=0,
\nonumber\\
C_{\Xi,k}(\bar\delta)
&:=
\left(\|\Phi_k^\star\|_2^2+\|HF\|_2^2\bar\delta^2\right)
C_{\Xi,k-1}(\bar\delta)
+
\|S_k^\star\|_2.
\label{eq:quartic_covariance_constants}
\end{align}
Induction in Eq.~\eqref{eq:quartic_trace_recursion}, together with Eq.~\eqref{eq:quartic_psd_trace_bound}, yields
\begin{equation}
\|\Xi_k\|_F
\le
C_{\Xi,k}(\bar\delta)\delta^2,
\qquad 1\le k\le T.
\label{eq:quartic_covariance_bound}
\end{equation}
The exact gain-drift identity~\eqref{eq:gain_drift_exact} and $S_k\succeq R_k\succeq\underline r I$ then give
\begin{equation}
\begin{aligned}
\|\Gamma_k\|_F
&\le
C_{\Gamma,k}(\bar\delta)\delta^2,
\\
C_{\Gamma,k}(\bar\delta)
&:=
\frac{\|\Phi_k^\star\|_2\|HF\|_2}{\underline r}
C_{\Xi,k-1}(\bar\delta).
\end{aligned}
\label{eq:quartic_gain_drift_bound}
\end{equation}
This proves Eq.~\eqref{eq:second_order_covariance_gain_drift} on every fixed finite horizon.

We next compare the exact covariance mismatch with the leading response~\eqref{eq:leading_covariance_response}. Subtracting that recursion from the exact reference-centered recursion~\eqref{eq:covariance_mismatch_recursions} and using Eq.~\eqref{eq:quartic_innovation_difference} gives the exact identity
\begin{align}
	\Xi_k-\Omega_k
	={}&
	\Phi_k^\star(\Xi_{k-1}-\Omega_{k-1})(\Phi_k^\star)^{\mathsf T}
	\nonumber\\
	&+
	E_kHF\,\Xi_{k-1}(HF)^{\mathsf T}E_k^{\mathsf T}
	-
	\Gamma_kS_k\Gamma_k^{\mathsf T}.
\end{align}
For $0\le\delta\le\bar\delta$, Eq.~\eqref{eq:quartic_innovation_norm_bound} and Eq.~\eqref{eq:quartic_covariance_bound} imply
\[
\|S_k\|_2
\le
C_{S,k}(\bar\delta)
:=
\|S_k^\star\|_2
+
\|HF\|_2^2 C_{\Xi,k-1}(\bar\delta)\bar\delta^2.
\]
Define
\begin{align*}
	C_{\Xi-\Omega,0}(\bar\delta)&:=0,\\
	C_{\Xi-\Omega,k}(\bar\delta)
	&:=
	\|\Phi_k^\star\|_2^2 C_{\Xi-\Omega,k-1}(\bar\delta)
	+\|HF\|_2^2 C_{\Xi,k-1}(\bar\delta)
	\\
	&\quad
	+C_{S,k}(\bar\delta)C_{\Gamma,k}(\bar\delta)^2.
\end{align*}
Using $\|AXB\|_F\le\|A\|_2\|X\|_F\|B\|_2$ and $\|E_k\|_2\le\|E_k\|_F\le\delta$, induction yields
\begin{equation}
	\|\Xi_k-\Omega_k\|_F
	\le
	C_{\Xi-\Omega,k}(\bar\delta)\delta^4,
	\qquad 0\le k\le T.
	\label{eq:quartic_covariance_remainder_bound}
\end{equation}
Consequently,
\begin{equation}
	\|\Omega_k\|_F
	\le
	C_{\Omega,k}(\bar\delta)\delta^2,
	\quad
	C_{\Omega,k}(\bar\delta)
	:=
	C_{\Xi,k}(\bar\delta)
	+
	\bar\delta^2 C_{\Xi-\Omega,k}(\bar\delta).
	\label{eq:quartic_omega_bound}
\end{equation}

The resolvent identity and Eq.~\eqref{eq:quartic_innovation_difference} give
\begin{align}
	S_k^{-1}-(S_k^\star)^{-1}
	&=
	-S_k^{-1}(S_k-S_k^\star)(S_k^\star)^{-1},
	\nonumber\\
	\|S_k^{-1}-(S_k^\star)^{-1}\|_2
	&\le
	\frac{\|HF\|_2^2}{\underline r^2}
	C_{\Xi,k-1}(\bar\delta)\delta^2.
	\label{eq:quartic_inverse_remainder_bound}
\end{align}
Therefore the exact gain-drift identity~\eqref{eq:gain_drift_exact} satisfies
\begin{align}
	&\left\|
	\Gamma_k
	-
	\Phi_k^\star\Omega_{k-1}(HF)^{\mathsf T}(S_k^\star)^{-1}
	\right\|_F
	\nonumber\\
	&\quad\le
	\|\Phi_k^\star\|_2\|HF\|_2
	\left[
	\frac{C_{\Xi-\Omega,k-1}(\bar\delta)}{\underline r}
	\right.
	\nonumber\\
	&\qquad\left.
	+
	\frac{\|HF\|_2^2}{\underline r^2}
	C_{\Omega,k-1}(\bar\delta)C_{\Xi,k-1}(\bar\delta)
	\right]\delta^4.
	\label{eq:quartic_gain_drift_remainder_bound}
\end{align}

We next prove the two sixth-order score expansions. Since $\rho_k^{\mathrm{exec}}=E_kS_k$, Eq.~\eqref{eq:residual_risk_contribution} gives exactly
\begin{equation}
\mathcal Q_T^{\mathrm{res}}
=
\sum_{k=1}^{T}
\tr\!\left[W_{k,T}^M E_kS_kE_k^{\mathsf T}\right].
\label{eq:quartic_residual_E_form}
\end{equation}
Using
\[
S_k
=
S_k^\star
+
HF\,\Omega_{k-1}(HF)^{\mathsf T}
+
HF\,(\Xi_{k-1}-\Omega_{k-1})(HF)^{\mathsf T},
\]
we obtain the exact decomposition
\begin{align}
	\mathcal Q_T^{\rm res}
	={}&
	\mathcal B_T^M(E)
	+
	\mathcal C_T^{M,\mathrm{infl}}(E)
	\nonumber\\
	&+
	\sum_{k=1}^{T}
	\tr\!\left[
	W_{k,T}^M
	E_kHF\,(\Xi_{k-1}-\Omega_{k-1})(HF)^{\mathsf T}E_k^{\mathsf T}
	\right].
\end{align}
Hence, by Eq.~\eqref{eq:quartic_covariance_remainder_bound},
\begin{equation}
	\left|
	\mathcal Q_T^{\rm res}
	-
	\mathcal B_T^M(E)
	-
	\mathcal C_T^{M,\mathrm{infl}}(E)
	\right|
	\le
	C_{\mathrm{res},6,T}^M(\bar\delta)\delta^6,
	\label{eq:quartic_residual_sixth_order_bound}
\end{equation}
where
\begin{equation}
	C_{\mathrm{res},6,T}^M(\bar\delta)
	:=
	\|HF\|_2^2
	\sum_{k=1}^{T}
	\|W_{k,T}^M\|_F
	C_{\Xi-\Omega,k-1}(\bar\delta).
	\label{eq:quartic_residual_sixth_order_constant}
\end{equation}
This proves Eq.~\eqref{eq:residual_quartic_expansion}.

For the drift contribution, Eq.~\eqref{eq:gain_drift_exact} and the symmetry of $\Xi_{k-1}$ and $S_k$ give the exact identity
\begin{equation}
	\Gamma_kS_k\Gamma_k^{\mathsf T}
	=
	\Phi_k^\star
	\Xi_{k-1}(HF)^{\mathsf T}S_k^{-1}HF\,\Xi_{k-1}
	(\Phi_k^\star)^{\mathsf T}.
	\label{eq:quartic_exact_drift_product}
\end{equation}
The difference between the middle matrix in Eq.~\eqref{eq:quartic_exact_drift_product} and the middle matrix defining $\mathcal C_T^{M,\mathrm{reopt}}$ has the exact decomposition
\begin{align}
	&\Xi_{k-1}(HF)^{\mathsf T}S_k^{-1}HF\,\Xi_{k-1}
	-
	\Omega_{k-1}(HF)^{\mathsf T}(S_k^\star)^{-1}HF\,\Omega_{k-1}
	\nonumber\\
	&={}
	(\Xi_{k-1}-\Omega_{k-1})(HF)^{\mathsf T}S_k^{-1}HF\,\Xi_{k-1}
	\nonumber\\
	&{}\quad+
	\Omega_{k-1}(HF)^{\mathsf T}
	\bigl[S_k^{-1}-(S_k^\star)^{-1}\bigr]
	HF\,\Xi_{k-1}
	\nonumber\\
	&{}\quad+
	\Omega_{k-1}(HF)^{\mathsf T}(S_k^\star)^{-1}
	HF\,(\Xi_{k-1}-\Omega_{k-1}).
	\label{eq:quartic_drift_middle_decomposition}
\end{align}
Equations~\eqref{eq:quartic_covariance_bound}, \eqref{eq:quartic_covariance_remainder_bound}, \eqref{eq:quartic_omega_bound}, and~\eqref{eq:quartic_inverse_remainder_bound} show that the three terms on the right-hand side of Eq.~\eqref{eq:quartic_drift_middle_decomposition} are uniformly $O_T(\delta^6)$. More explicitly,
\begin{equation}
	\left|
	\mathcal Q_T^{\rm drift}
	-
	\mathcal C_T^{M,\mathrm{reopt}}(E)
	\right|
	\le
	C_{\mathrm{drift},6,T}^M(\bar\delta)\delta^6,
	\label{eq:quartic_drift_sixth_order_bound}
\end{equation}
with the finite direction-independent constant
\begin{align}
	C_{\mathrm{drift},6,T}^M(\bar\delta)
	:=
	\sum_{k=1}^{T}
	&\|W_{k,T}^M\|_F\|\Phi_k^\star\|_2^2
	\nonumber\\
	&\times
	\Bigg[
	\frac{\|HF\|_2^2}{\underline r}
	C_{\Xi-\Omega,k-1}(\bar\delta)C_{\Xi,k-1}(\bar\delta)
	\nonumber\\
	&\quad+
	\frac{\|HF\|_2^4}{\underline r^2}
	C_{\Omega,k-1}(\bar\delta)C_{\Xi,k-1}(\bar\delta)^2
	\nonumber\\
	&\quad+
	\frac{\|HF\|_2^2}{\underline r}
	C_{\Omega,k-1}(\bar\delta)C_{\Xi-\Omega,k-1}(\bar\delta)
	\Bigg].
	\label{eq:quartic_drift_sixth_order_constant}
\end{align}
This proves Eq.~\eqref{eq:drift_quartic_expansion}.

Finally, the exact accounting identity~\eqref{eq:exact_residual_risk_accounting} gives
\begin{align*}
	&\left|
	\mathcal R_T^M
	-
	\mathcal B_T^M(E)
	-
	\mathcal C_T^{M,\mathrm{infl}}(E)
	+
	\mathcal C_T^{M,\mathrm{reopt}}(E)
	\right|
	\\
	&\;\le
	\left|
	\mathcal Q_T^{\rm res}
	-
	\mathcal B_T^M(E)
	-
	\mathcal C_T^{M,\mathrm{infl}}(E)
	\right|
	\\
	&\quad+
	\left|
	\mathcal Q_T^{\rm drift}
	-
	\mathcal C_T^{M,\mathrm{reopt}}(E)
	\right|.
\end{align*}
Thus Eq.~\eqref{eq:quartic_uniform_remainder} holds with
\begin{equation}
	C_{6,T}^M(\bar\delta)
	:=
	C_{\mathrm{res},6,T}^M(\bar\delta)
	+
	C_{\mathrm{drift},6,T}^M(\bar\delta)
	<\infty,
\end{equation}
which is independent of the defect directions and of the actual $\delta\in[0,\bar\delta]$. Equations~\eqref{eq:residual_quartic_expansion}, \eqref{eq:drift_quartic_expansion}, and~\eqref{eq:signed_quartic_response} follow.

Additionally, $\Omega_k\succeq0$ and is homogeneous quadratic in $E_{1:k}$ by Eq.~\eqref{eq:leading_covariance_response}. Therefore each matrix injection in Eq.~\eqref{eq:quartic_inflation_response} is positive semidefinite and homogeneous quartic. Likewise,
\[
\Phi_k^\star\Omega_{k-1}(HF)^{\mathsf T}(S_k^\star)^{-1}
HF\,\Omega_{k-1}(\Phi_k^\star)^{\mathsf T}
=
Y_kY_k^{\mathsf T},
\]
with
\[
Y_k:=
\Phi_k^\star\Omega_{k-1}(HF)^{\mathsf T}(S_k^\star)^{-1/2},
\]
so the reoptimization injection is also positive semidefinite and homogeneous quartic. Since $W_{k,T}^M\succeq0$, this verifies
\[
\mathcal C_T^{M,\mathrm{infl}}(E)\ge0,
\qquad
\mathcal C_T^{M,\mathrm{reopt}}(E)\ge0.
\]
\end{proof}

\subsection{Consequences of the signed-quartic response}
\label{app:consequences_of_the_signed}

\begin{corollary}[Exact global task-nullness]
	\label{cor:exact_global_task_nullness}
	Under the assumptions of Theorem~\ref{thm:quartic_finit_horiz_respon}, if the prescribed intrinsic defect sequence \(E_{1:T}\) satisfies
	\begin{equation}
		\mathcal B_T^M(E)=0,
		\label{eq:global_task_null_condition}
	\end{equation}
	then
	\begin{equation}
		\mathcal R_T^M(E)
		=
		\mathcal Q_T^{\rm res}
		=
		\mathcal Q_T^{\rm drift}
		=
		0,
		\label{eq:exact_global_task_null_response}
	\end{equation}
	and
	\begin{equation}
		\mathcal C_T^{M,\mathrm{infl}}(E)
		=
		\mathcal C_T^{M,\mathrm{reopt}}(E)
		=
		0.
		\label{eq:global_task_null_quartic_terms}
	\end{equation}
\end{corollary}

\begin{proof}
	Because every summand in Eq.~\eqref{eq:fixed_reference_quadratic_form} is nonnegative, Eq.~\eqref{eq:global_task_null_condition} implies
	\[
	\left\|
	(W_{k,T}^M)^{1/2}
	E_k
	(S_k^\star)^{1/2}
	\right\|_F^2
	=0
	\qquad
	\text{for every }k.
	\]
	Since \(S_k^\star\succ0\), this gives
	\[
	(W_{k,T}^M)^{1/2}E_k=0
	\qquad
	\text{for every }k.
	\]
	Therefore
	\[
	\tr\!\left[
	W_{k,T}^M
	E_kS_kE_k^{\mathsf T}
	\right]
	=
	\left\|
	(W_{k,T}^M)^{1/2}E_kS_k^{1/2}
	\right\|_F^2
	=0,
	\]
	so \(\mathcal Q_T^{\rm res}=0\). Equation~\eqref{eq:exact_residual_risk_accounting}, together with
	\(\mathcal R_T^M\ge0\) and \(\mathcal Q_T^{\rm drift}\ge0\), then gives
	\[
	\mathcal R_T^M(E)
	=
	\mathcal Q_T^{\rm drift}
	=
	0.
	\]
	The same relation \((W_{k,T}^M)^{1/2}E_k=0\) makes every summand in
	Eq.~\eqref{eq:quartic_inflation_response} vanish, so
	\(\mathcal C_T^{M,\mathrm{infl}}(E)=0\).
	
	For any scalar \(0<\alpha\le1\), the scaled sequence \(\alpha E\) has defect amplitude \(\alpha\delta\le\bar\delta\) and satisfies
	\[
	\mathcal B_T^M(\alpha E)
	=
	\alpha^2\mathcal B_T^M(E)
	=
	0.
	\]
	Hence \(\mathcal R_T^M(\alpha E)=0\), and the preceding argument also gives
	\(\mathcal C_T^{M,\mathrm{infl}}(\alpha E)=0\).
	Applying Eq.~\eqref{eq:quartic_uniform_remainder} to \(\alpha E\) therefore yields
	\[
	\left|
	\mathcal C_T^{M,\mathrm{reopt}}(\alpha E)
	\right|
	\le
	C_{6,T}^M(\bar\delta)\alpha^6\delta^6.
	\]
	Since \(\mathcal C_T^{M,\mathrm{reopt}}\) is nonnegative and homogeneous quartic,
	\[
	\alpha^4\mathcal C_T^{M,\mathrm{reopt}}(E)
	\le
	C_{6,T}^M(\bar\delta)\alpha^6\delta^6,
	\]
	and hence
	\[
	\mathcal C_T^{M,\mathrm{reopt}}(E)
	\le
	C_{6,T}^M(\bar\delta)\alpha^2\delta^6
	\qquad
	\text{for every }0<\alpha\le1.
	\]
	Since the left-hand side is independent of \(\alpha\), whereas the right-hand side can be made arbitrarily small,
	\[
	\mathcal C_T^{M,\mathrm{reopt}}(E)=0.
	\]
\end{proof}

This corollary concerns nullity of the full finite-horizon prescribed defect sequence; it does not exclude higher-order interactions between a locally task-null defect and other non-null defects.

\paragraph{Quadratic normal form and relative accuracy.}

The quadratic normal form in Eq.~\eqref{eq:fourth_order_risk_normal_form} follows directly from Theorem~\ref{thm:quartic_finit_horiz_respon}, since \(\mathcal C_T^{M,\mathrm{infl}}\) and \(\mathcal C_T^{M,\mathrm{reopt}}\) are homogeneous quartic and the remaining term is \(O_T(\delta^6)\). Because \(\mathcal B_T^M\) is positive semidefinite and may become arbitrarily small or vanish, Eq.~\eqref{eq:fourth_order_risk_normal_form} does not imply uniform relative accuracy or a uniform \(\Theta(\delta^2)\) response over the full defect class. On any restricted class satisfying
\[
\mathcal B_T^M(E)\ge \mu\delta^2,
\qquad
\text{for some fixed }\mu>0,
\]
the fourth-order absolute remainder implies a uniform \(O_T(\delta^2)\) relative error with respect to \(\mathcal B_T^M\). Under exact global task-nullness, Corollary~\ref{cor:exact_global_task_nullness} gives \(\mathcal B_T^M(E)=\mathcal R_T^M(E)=0\), so relative error with respect to \(\mathcal B_T^M\) is undefined.

\paragraph{Symmetry and analyticity.}

The exact finite-horizon response is blockwise even in the prescribed intrinsic defects:
\[
\mathcal R_T^M(\sigma_1E_1,\ldots,\sigma_TE_T)
=
\mathcal R_T^M(E_1,\ldots,E_T),
\;
\sigma_k\in\{-1,1\}.
\]
If \(\Xi_{k-1}\) (cf.~Eq.~\eqref{eq:covariance_mismatch_recursions}) is unchanged, then \(S_k\) and \(\Gamma_k\) are unchanged, while
\[
(\sigma_kE_k)S_k(\sigma_kE_k)^{\mathsf T}
=
E_kS_kE_k^{\mathsf T}.
\]
Induction in Eq.~\eqref{eq:covariance_mismatch_recursions} therefore leaves every \(\Xi_k\), and hence \(\mathcal R_T^M\), unchanged.

For fixed finite \(T\), the map \(E_{1:T}\mapsto\mathcal R_T^M(E_{1:T})\) is real analytic near the origin. The recursion involves matrix additions, products, and inversion, while
\[
S_k
=
S_k^\star+HF\,\Xi_{k-1}(HF)^{\mathsf T}
\succeq
S_k^\star
\succeq
\underline r I_m.
\]
Blockwise evenness therefore implies that every nonzero Taylor monomial has even degree in each defect block. A monomial involving two distinct defect times has degree at least two in each and hence total degree at least four. Quartic terms need not, however, couple distinct times. These properties concern the intrinsic defect coordinates \(E_{1:T}\) and need not hold for solver iterations, stopping tolerances, correction parameters, or fallback decisions.

\subsection{Statistical identification under measurement-independent execution}
\label{app:statistical_identification}

\paragraph{Information structure and reference conditional moments}
The deterministic response theory of Sec.~\ref{sec:finit_horiz_respon} does not require an execution sigma-field. To prove Proposition~\ref{prop:statistical_identification} and identify the covariance response with physical conditional moments, we now make the corresponding information structure precise. In probability theory, a \emph{sigma-field} represents available information: it is the collection of events whose occurrence can be decided from the quantities observed so far. For random variables \(X_1,\ldots,X_j\), the notation \(\sigma(X_1,\ldots,X_j)\) denotes the sigma-field generated by them, i.e., the smallest sigma-field containing all information carried by their realized values.

Accordingly, let \(\mathcal Y_k:=\sigma(z_1,\ldots,z_k)\) denote the information contained in the physical measurements through time \(k\), and let \(\mathcal Y_0\) denote the trivial sigma-field, representing the absence of physical measurement information before the first measurement. The execution environment \(\cG\) is a separate sigma-field containing the auxiliary information and randomness that determine how the numerical computation is carried out but do not arise from the physical measurements. Examples include a predeclared execution schedule and independent algorithmic randomness; deterministic model and algorithm parameters are treated as fixed. Conditioning on \(\cG\) holds the numerical execution fixed while leaving the physical state, disturbances, and measurements random.

For two sigma-fields \(\mathcal A\) and \(\mathcal B\), define
\[
\mathcal A\vee\mathcal B:=\sigma(\mathcal A\cup\mathcal B)
\]
the smallest sigma-field containing both, which represents the information available from \(\mathcal A\) and \(\mathcal B\) jointly.
By the hypotheses of Proposition~\ref{prop:statistical_identification}, \(\cG\) is jointly independent of the sigma-field generated by the physical initial state \(x_0\) and the complete physical disturbance trajectory
\[
x_0,\quad
\xi_0,\ldots,\xi_{T-1},\quad
w_0,\ldots,w_{T-1},\quad
v_1,\ldots,v_T.
\]
Since the physical measurements through time \(j\) are functions of these variables under the fixed model and input schedule, \(\mathcal Y_j\) is contained in this physical sigma-field. Consequently, for every integrable physical random variable \(X\),
\[
\E[X\mid\mathcal Y_j\vee\cG]
=
\E[X\mid\mathcal Y_j].
\]
For square-integrable \(X\), applying the same argument to \(XX^{\mathsf T}\) preserves the corresponding conditional second moments.
Define
\begin{equation}
\cF_k^-:=\mathcal Y_{k-1}\vee\cG,
\qquad
\cF_k^+:=\mathcal Y_k\vee\cG.
\label{eq:information_structure}
\end{equation}
Thus $\cF_k^-$ represents the information available immediately before incorporating the current measurement \(z_k\), whereas $\cF_k^+$ additionally contains \(z_k\). The preceding conditional-expectation identity shows that adjoining \(\cG\) does not alter the conditional moments of the correctly specified reference estimator.

The reference prediction mean and covariance satisfy
\[
\widetilde x_k^\star
=
\E[x_k\mid\cF_k^-],
\qquad
\widetilde P_k^\star
=
\Cov(x_k\mid\cF_k^-).
\]
The reference posterior mean and covariance satisfy
\[
\widehat x_k^\star
=
\E[x_k\mid\cF_k^+],
\qquad
\widehat P_k^\star
=
\Cov(x_k\mid\cF_k^+).
\]
The reference innovation is
\[
r_k^\star
=
z_k-H\widetilde x_k^\star
=
H(x_k-\widetilde x_k^\star)+v_k,
\]
and therefore
\[
\E[r_k^\star\mid\cF_k^-]=0,
\qquad
\E[r_k^\star(r_k^\star)^{\mathsf T}\mid\cF_k^-]
=S_k^\star.
\]

\paragraph{Implemented covariance and reference-relative error}
\begin{proof}[Proof of Proposition~\ref{prop:statistical_identification}]
Define the implemented prediction and posterior errors by
\[
\widetilde\epsilon_k:=x_k-\widetilde x_k,
\qquad
\epsilon_k:=x_k-\widehat x_k.
\]
The physical model and implemented recursion give
\begin{align}
	\widetilde\epsilon_k
	&=
	F\epsilon_{k-1}
	+B\xi_{k-1}
	+w_{k-1},
	\label{eq:implemented_prediction_error}\\
	\epsilon_k
	&=
	(I-\widetilde K_kH)\widetilde\epsilon_k
	-\widetilde K_kv_k.
	\label{eq:implemented_posterior_error}
\end{align}
Under the hypotheses of Proposition~\ref{prop:statistical_identification}, the executed gain sequence is fixed conditional on \(\cG\), and joint independence implies that the physical initial state and disturbances retain their stated joint law conditional on \(\cG\). Assumption~\ref{ass:matched_initialization} gives
\[
\E[\epsilon_0\mid\cG]=0,
\qquad
\E[\epsilon_0\epsilon_0^{\mathsf T}\mid\cG]=P_0.
\]
Using Eqs.~\eqref{eq:implemented_prediction_error} and~\eqref{eq:implemented_posterior_error}, together with the zero means, mutual independence, and temporal independence of the physical disturbances, induction gives
\[
\E[\widetilde\epsilon_k\mid\cG]=0,
\qquad
\E[\epsilon_k\mid\cG]=0,
\]
and
\begin{align}
	\E\!\left[
	\widetilde\epsilon_k\widetilde\epsilon_k^{\mathsf T}
	\mid\cG
	\right]
	&=
	\widetilde P_k,
	\label{eq:implemented_prediction_covariance_statistical}\\
	\E\!\left[
	\epsilon_k\epsilon_k^{\mathsf T}
	\mid\cG
	\right]
	&=
	\mathcal{J}_k(\widetilde P_k,\widetilde K_k)
	=
	\widehat P_k.
	\label{eq:implemented_covariance_statistical}
\end{align}
Since \(\epsilon_k=x_k-\widehat x_k,\) Eq.~\eqref{eq:implemented_covariance_statistical} proves the conditional second-moment identity in Eq.~\eqref{eq:implemented_covariance_interpretation}.

We then define the reference estimation error and estimator mismatch by
\[
\epsilon_k^\star:=x_k-\widehat x_k^\star,
\qquad
e_k:=\widehat x_k-\widehat x_k^\star.
\]
Because \(\widehat x_k^\star=\E[x_k\mid\cF_k^+]\),
\[
\E[\epsilon_k^\star\mid\cF_k^+]=0.
\]
By matched initialization and induction through the prediction and update recursions, using the \(\cG\)-measurability of \(\widetilde K_k\), both \(\widehat x_k\) and \(\widehat x_k^\star\) are \(\cF_k^+\)-measurable, and hence so is \(e_k\).
Since \(\epsilon_k=\epsilon_k^\star-e_k,\) it follows that
\(\E[\epsilon_k\mid\cF_k^+]=-e_k,\) and hence
\[
\E[
\epsilon_k\epsilon_k^{\mathsf T}
\mid\cF_k^+]
= \widehat P_k^\star+e_ke_k^{\mathsf T},
\qquad
\Cov(\epsilon_k\mid\cF_k^+)
= \widehat P_k^\star.
\]
Thus \(\widehat P_k\) in Eq.~\eqref{eq:implemented_covariance_interpretation} is a second moment conditional on \(\cG\), not the posterior covariance conditioned on \(\cF_k^+\).
Since \(\cG\subseteq\cF_k^+\), the tower property gives
\begin{align}
\E[
\epsilon_k^\star e_k^{\mathsf T}
\mid\cG]
&=
\E\!\left[
\E[
\epsilon_k^\star e_k^{\mathsf T}
\mid\cF_k^+]
\mid\cG
\right]
\nonumber\\
&=
\E\!\left[
\E[
\epsilon_k^\star
\mid\cF_k^+]
e_k^{\mathsf T}
\mid\cG
\right]
=
0.
\label{eq:reference_mismatch_orthogonality}
\end{align}
Because
\(\widehat P_k^\star=\Cov(x_k\mid\cF_k^+)\)
and the reference covariance is fixed by the prescribed model and covariance schedule, the tower property gives
\[
\E\!\left[
\epsilon_k^\star(\epsilon_k^\star)^{\mathsf T}
\mid\cG
\right]
=
\widehat P_k^\star.
\]
Since \(\epsilon_k=\epsilon_k^\star-e_k\), Eqs.~\eqref{eq:implemented_covariance_statistical} and~\eqref{eq:reference_mismatch_orthogonality} together with the preceding identity imply Eq.~\eqref{eq:covariance_mismatch_statistical}, and hence \(\Xi_k\succeq0\).
Since both the implemented and reference conditional error means vanish, \(\E[e_k\mid\cG]=0\), so
\[
\Xi_k
=
\Cov(e_k\mid\cG).
\]

For the fixed task metrics \(M_k\succeq0\), using the quadratic metric defined in Sec.~\ref{subsec:covari_misma_and_task_respo}, Eq.~\eqref{eq:reference_mismatch_orthogonality} gives
\begin{align}
	&\E\!\left[
	\|x_k-\widehat x_k\|_{M_k}^2
	-
	\|x_k-\widehat x_k^\star\|_{M_k}^2
	\mid\cG
	\right]
	\nonumber\\
	&\qquad=
	\E[e_k^{\mathsf T}M_ke_k\mid\cG]
	=
	\tr(M_k\Xi_k).
	\label{eq:reference_relative_quadratic_error}
\end{align}
Summing Eq.~\eqref{eq:reference_relative_quadratic_error} over \(k=1,\ldots,T\) proves the two task-response identities in Proposition~\ref{prop:statistical_identification}.
\end{proof}

\subsection{Residual magnitude and task dependence}
\label{subsec:residu_magnit_and_task_depend}

For a nonzero linear comparison subspace $\mathcal U_k\subseteq\mathbb R^{n\times m}$, define
\begin{align}
\lambda_{k,\min}^{\mathcal U}
&:=
\min_{E\in\mathcal U_k\setminus\{0\}}
\frac{
\tr\!\left[W_{k,T}^M E S_k^\star E^{\mathsf T}\right]
}{
\|E\|_F^2
},
\label{eq:residual_exposure_lambda_min}\\
\lambda_{k,\max}^{\mathcal U}
&:=
\max_{E\in\mathcal U_k\setminus\{0\}}
\frac{
\tr\!\left[W_{k,T}^M E S_k^\star E^{\mathsf T}\right]
}{
\|E\|_F^2
}.
\label{eq:residual_exposure_lambda_max}
\end{align}
The extrema are well defined because the quotient is invariant under nonzero rescaling of $E$ and is continuous on any fixed nonzero Frobenius sphere in $\mathcal U_k$.
Here, for $E\neq0$, its direction means the normalized defect $E/\|E\|_F$; directional dependence therefore means variation of the task exposure with this orientation at fixed Frobenius magnitude.

\begin{proposition}[Residual magnitude and task dependence]
\label{prop:residu_magnit_and_task_depend}
For every $E\in\mathcal U_k$,
\begin{equation}
\lambda_{k,\min}^{\mathcal U}\|E\|_F^2
\le
\tr\!\left[W_{k,T}^M E S_k^\star E^{\mathsf T}\right]
\le
\lambda_{k,\max}^{\mathcal U}\|E\|_F^2,
\label{eq:directional_exposure_bounds}
\end{equation}
and both bounds are attained at every fixed nonzero Frobenius radius. Frobenius magnitude determines the fixed-reference quadratic exposure on $\mathcal U_k$ if and only if
\begin{equation}
\lambda_{k,\min}^{\mathcal U}=\lambda_{k,\max}^{\mathcal U}.
\label{eq:directional_scalar_sufficiency}
\end{equation}
If instead \(\lambda_{k,\min}^{\mathcal U}<\lambda_{k,\max}^{\mathcal U}\), equal-magnitude directions with unequal exposure exist. More generally, if $0<r_a<r_b$ and
\begin{equation}
r_a^2\lambda_{k,\max}^{\mathcal U}
>
r_b^2\lambda_{k,\min}^{\mathcal U},
\label{eq:directional_order_reversal_condition}
\end{equation}
then directions $E^{(a)},E^{(b)}\in\mathcal U_k$ exist with $\|E^{(a)}\|_F=r_a<r_b=\|E^{(b)}\|_F$ but with the larger fixed-reference quadratic exposure for $E^{(a)}$.
At a matched current step, $\Xi_{k-1}=0$, so $S_k=S_k^\star$ and
\begin{equation}
\tr\!\left[W_{k,T}^M E_kS_k^\star E_k^{\mathsf T}\right]
=
\tr\!\left[
W_{k,T}^M\rho_k^{\mathrm{exec}}(S_k^\star)^{-1}
(\rho_k^{\mathrm{exec}})^{\mathsf T}
\right].
\label{eq:directional_residual_risk_form}
\end{equation}
Thus the same directional statement holds in residual coordinates, with quadratic operator $(S_k^\star)^{-1}\otimes W_{k,T}^M$ on the corresponding residual directions.

Finally, let $E^{(a)}$ and $E^{(b)}$ be two prescribed intrinsic defect sequences satisfying Theorem~\ref{thm:quartic_finit_horiz_respon} under the same fixed theorem data, in particular the same finite horizon $T$, reference model and covariance schedule, task metrics $M_{1:T}$, innovation lower bound $\underline r$, and cap $\bar\delta$. Let
\[
\delta_a:=\max_j\|E_j^{(a)}\|_F,
\quad
\delta_b:=\max_j\|E_j^{(b)}\|_F,
\quad
\delta_a,\delta_b\le\bar\delta.
\]
Suppose
\[
\mathcal B_T^M(E^{(a)})
>
\mathcal B_T^M(E^{(b)}).
\]
If, in addition,
\begin{align}
&
\mathcal B_T^M(E^{(a)})
-
\mathcal B_T^M(E^{(b)})
+
\mathcal C_T^{M,\mathrm{infl}}(E^{(a)})
-
\mathcal C_T^{M,\mathrm{infl}}(E^{(b)})
\nonumber\\
&\quad-
\mathcal C_T^{M,\mathrm{reopt}}(E^{(a)})
+
\mathcal C_T^{M,\mathrm{reopt}}(E^{(b)})
\nonumber\\
&\quad>
C_{6,T}^M(\bar\delta)
\left(\delta_a^6+\delta_b^6\right),
\label{eq:directional_transfer_margin}
\end{align}
then
\begin{equation}
\mathcal R_T^M(E^{(a)})
>
\mathcal R_T^M(E^{(b)}).
\label{eq:directional_risk_order_transfer}
\end{equation}
\end{proposition}

\begin{proof}
For $E\in\mathbb R^{n\times m}$,
\begin{equation}
\tr\!\left[W_{k,T}^M E S_k^\star E^{\mathsf T}\right]
=
\left\|(W_{k,T}^M)^{1/2}E(S_k^\star)^{1/2}\right\|_F^2
\ge0.
\label{eq:directional_psd_form}
\end{equation}
For any nonzero \(E\in\mathcal U_k\), quadratic homogeneity and Eqs.~\eqref{eq:residual_exposure_lambda_min}--\eqref{eq:residual_exposure_lambda_max} give
\begin{align}
\tr\!\left[W_{k,T}^M E S_k^\star E^{\mathsf T}\right]
&=
\|E\|_F^2
\tr\!\left[
W_{k,T}^M
\frac{E}{\|E\|_F}
S_k^\star
\frac{E^{\mathsf T}}{\|E\|_F}
\right]
\nonumber\\
&\in
\left[
\lambda_{k,\min}^{\mathcal U}\|E\|_F^2,
\lambda_{k,\max}^{\mathcal U}\|E\|_F^2
\right].
\label{eq:residual_exposure_homogeneity}
\end{align}
The extrema on the unit sphere are attained, and rescaling the corresponding unit-Frobenius extremizers proves sharpness at every nonzero radius.

If $\lambda_{k,\min}^{\mathcal U}=\lambda_{k,\max}^{\mathcal U}$, Eq.~\eqref{eq:directional_exposure_bounds} collapses to a constant multiple of $\|E\|_F^2$, so magnitude alone determines the fixed-reference quadratic exposure. Conversely, if magnitude alone determines the fixed-reference quadratic exposure, the quadratic form is constant on the unit sphere of $\mathcal U_k$, and therefore its minimum and maximum there coincide. With column vectorization,
\begin{equation}
\tr\!\left[W_{k,T}^M E S_k^\star E^{\mathsf T}\right]
=
\operatorname{vec}(E)^{\mathsf T}
\left(S_k^\star\otimes W_{k,T}^M\right)
\operatorname{vec}(E),
\label{eq:directional_vec_form}
\end{equation}
so Eq.~\eqref{eq:directional_scalar_sufficiency} is equivalently the statement that the compression of $S_k^\star\otimes W_{k,T}^M$ to $\operatorname{vec}(\mathcal U_k)$ is a scalar multiple of the identity.

If $\lambda_{k,\max}^{\mathcal U}>\lambda_{k,\min}^{\mathcal U}$, unit-Frobenius minimizers and maximizers give equal-magnitude directions with unequal exposure. If in addition Eq.~\eqref{eq:directional_order_reversal_condition} holds, rescaling those two extremizers to radii $r_a$ and $r_b$, respectively, gives
\begin{align}
\tr\!\left[W_{k,T}^M E^{(a)}S_k^\star(E^{(a)})^{\mathsf T}\right]
&=
r_a^2\lambda_{k,\max}^{\mathcal U}
\nonumber\\
&>
r_b^2\lambda_{k,\min}^{\mathcal U}
\nonumber\\
&=
\tr\!\left[W_{k,T}^M E^{(b)}S_k^\star(E^{(b)})^{\mathsf T}\right].
\label{eq:directional_order_reversal_proof}
\end{align}
The case $\lambda_{k,\min}^{\mathcal U}=0$ is included: the minimizing direction is nonzero but has zero fixed-reference quadratic exposure for the declared metric and horizon. More generally, if the full prescribed intrinsic defect sequence satisfies $\mathcal B_T^M(E)=0$, then its exact finite-horizon response also vanishes, $\mathcal R_T^M(E)=0$; see Corollary~\ref{cor:exact_global_task_nullness}.

At a matched current step, $\Xi_{k-1}=0$ implies $S_k=S_k^\star$ and $K_k^{\rm loc}=K_k^\star$. Since $\rho_k^{\rm exec}=E_kS_k^\star$,
\begin{equation}
E_k=\rho_k^{\rm exec}(S_k^\star)^{-1},
\end{equation}
and substitution gives Eq.~\eqref{eq:directional_residual_risk_form}. Vectorization then yields
\begin{align}
&\tr\!\Bigl[
W_{k,T}^M\rho_k^{\rm exec}(S_k^\star)^{-1}
(\rho_k^{\rm exec})^{\mathsf T}
\Bigr]
\nonumber\\
&\qquad=
\operatorname{vec}(\rho_k^{\rm exec})^{\mathsf T}
\left[(S_k^\star)^{-1}\otimes W_{k,T}^M\right]
\operatorname{vec}(\rho_k^{\rm exec}).
\label{eq:directional_residual_vec_form}
\end{align}

Applying Eq.~\eqref{eq:quartic_uniform_remainder} to $E^{(a)}$ and $E^{(b)}$ gives
\begin{align}
&
\mathcal R_T^M(E^{(a)})
-
\mathcal R_T^M(E^{(b)}) \ge
\mathcal B_T^M(E^{(a)})
-
\mathcal B_T^M(E^{(b)})
\nonumber\\
&\qquad\qquad\qquad+
\mathcal C_T^{M,\mathrm{infl}}(E^{(a)})
-
\mathcal C_T^{M,\mathrm{infl}}(E^{(b)})
\nonumber\\
&\qquad\qquad\qquad-
\mathcal C_T^{M,\mathrm{reopt}}(E^{(a)})
+
\mathcal C_T^{M,\mathrm{reopt}}(E^{(b)})
\nonumber\\
&\qquad\qquad\qquad-
C_{6,T}^M(\bar\delta)
\left(\delta_a^6+\delta_b^6\right).
\label{eq:directional_transfer_proof}
\end{align}
Condition~\eqref{eq:directional_transfer_margin} therefore makes the right-hand side strictly positive and proves Eq.~\eqref{eq:directional_risk_order_transfer}. This is an exact-response ordering transfer: it is finite horizon and deterministic and uses only the assumptions of Theorem~\ref{thm:quartic_finit_horiz_respon}.
\end{proof}

\subsection{Finite-horizon response control}
\label{app:finite_horizon_respon_control}

\begin{proof}[Proof of Proposition~\ref{prop:finite_horizon_respon_control}]
From Eq.~\eqref{eq:certified_innovation_lower_bound},
\begin{equation}
S_k^{-1}\preceq\ell_k^{-1}I_m.
\label{eq:risk_control_inverse_bound}
\end{equation}
Congruence with the executed residual gives
\begin{equation}
\rho_k^{\mathrm{exec}}S_k^{-1}(\rho_k^{\mathrm{exec}})^{\mathsf T}
\preceq
\ell_k^{-1}\rho_k^{\mathrm{exec}}(\rho_k^{\mathrm{exec}})^{\mathsf T}.
\label{eq:risk_control_residual_congruence}
\end{equation}
Since $W_{k,T}^M\succeq0$, pairing Eq.~\eqref{eq:risk_control_residual_congruence} with $W_{k,T}^M$ and taking the trace yields
\begin{equation}
q_{k,T}\le c_{k,T}^{W,\ell}.
\label{eq:risk_control_first_upper_bound}
\end{equation}
Moreover,
\begin{equation}
W_{k,T}^M\preceq\lambda_{\max}(W_{k,T}^M)I_n,
\end{equation}
so
\begin{align}
c_{k,T}^{W,\ell}
&=
\ell_k^{-1}
\tr\!\left[(\rho_k^{\mathrm{exec}})^{\mathsf T}
W_{k,T}^M\rho_k^{\mathrm{exec}}\right]
\nonumber\\
&\le
\frac{\lambda_{\max}(W_{k,T}^M)}{\ell_k}
\|\rho_k^{\mathrm{exec}}\|_F^2
=
c_{k,T}^{\mathrm{rad}}.
\label{eq:risk_control_second_upper_bound}
\end{align}
Under the matched initialization, Eq.~\eqref{eq:exact_residual_risk_accounting} and $\mathcal Q_T^{\mathrm{drift}}\ge0$ give
\begin{equation}
\mathcal R_T^M\le\sum_{k=1}^{T}q_{k,T}.
\label{eq:risk_control_exact_upper_accounting}
\end{equation}
Summing Eqs.~\eqref{eq:risk_control_first_upper_bound} and~\eqref{eq:risk_control_second_upper_bound} proves Eq.~\eqref{eq:risk_bound_hierarchy}. If $c_{k,T}^{W,\ell}\le b_{k,T}$ for every $k$, Eq.~\eqref{eq:task_budget_allocation} immediately yields Eq.~\eqref{eq:global_task_risk_guarantee}.
\end{proof}

The finite-horizon bound also connects directly to the local admissibility criterion. Equation~\eqref{eq:radial_step_upper_bound} gives
\begin{equation}
c_{k,T}^{\mathrm{rad}}\le
\lambda_{\max}(W_{k,T}^M)\,\ell_k\,\delta_{\mathrm{adm}}^2
=b_{k,T}^{\mathrm{adm}}
\end{equation}
if and only if $\|\rho_k^{\mathrm{exec}}\|_F\le\ell_k\delta_{\mathrm{adm}}$ whenever $\lambda_{\max}(W_{k,T}^M)>0$. Equation~\eqref{eq:accept_rule} is precisely $\|\rho_k^{\mathrm{exec}}\|_F\le\ell_k\delta_{\mathrm{adm}}$, and is therefore equivalent to the radial allocation above whenever
\(\lambda_{\max}(W_{k,T}^M)>0\). If \(W_{k,T}^M=0\), both \(c_{k,T}^{\mathrm{rad}}\) and \(b_{k,T}^{\mathrm{adm}}\) vanish and the allocation holds trivially. Summing these allocations and applying Eq.~\eqref{eq:risk_bound_hierarchy} proves Eq.~\eqref{eq:admissibility_induced_horizon_bound}.

Conversely, for a prescribed allocation $b_{k,T}$ with $\lambda_{\max}(W_{k,T}^M)>0$, define the associated radial local tolerance by
\begin{equation}
\delta_{{\rm adm},k}^{(\tau)}
:=
\sqrt{
\frac{b_{k,T}}
{\lambda_{\max}(W_{k,T}^M)\,\ell_k}
}.
\label{eq:risk_calibrated_local_tolerance}
\end{equation}
Then
\[
\ell_k\delta_{{\rm adm},k}^{(\tau)}
=
\sqrt{
\frac{\ell_k b_{k,T}}
{\lambda_{\max}(W_{k,T}^M)}
},
\]
and substitution into Eq.~\eqref{eq:radial_step_upper_bound} gives
\[
\|\rho_k^{\rm exec}\|_F
\le
\ell_k\delta_{{\rm adm},k}^{(\tau)}
\quad\Longleftrightarrow\quad
c_{k,T}^{\rm rad}
\le
b_{k,T}.
\]

\subsection{Uniform-in-time bounds under strict Euclidean contraction}
\label{app:uniform_in_time_bounds}

Here we record a stronger sufficient regime under which a uniform cap on the executed local gain defects yields pointwise covariance and estimator-reference mean-square bounds for all time. These bounds do not control an infinite-horizon accumulated risk.

\begin{assumption}[Uniform reference regime and local defect cap]
\label{ass:uniform_tube_regime}
The reference and executed recursions extend to every \(k\geq1\) under the standing model conditions and retain matched covariance initialization, \(\Xi_0=0\). There exist constants \(\bar s^\star<\infty\) and \(\rho_\Phi\in[0,1)\) such that
\begin{align}
\sup_{k\geq1}\|S_k^\star\|_2
&\leq \bar s^\star,
\label{eq:uniform_reference_innovation_bound}\\
\sup_{k\geq1}\|\Phi_k^\star\|_2
&\leq \rho_\Phi,
\label{eq:ref_contraction}
\end{align}
and the executed path satisfies
\begin{equation}
\|E_k\|_F\leq\delta_{\rm adm}
\qquad\text{for every }k\geq1.
\label{eq:uniform_local_admissibility}
\end{equation}
\end{assumption}

The first condition in Assumption~\ref{ass:uniform_tube_regime} is an additional all-time boundedness assumption on the reference innovation covariance. In the time-invariant Kalman setting, such a bound follows, for example, under standard conditions ensuring boundedness or convergence of the reference Riccati recursion~\cite{anderson1979optimal,kailath2000linear}. The second condition requires strict contraction in the Euclidean operator norm and is stronger than Schur stability. Equation~\eqref{eq:uniform_local_admissibility} is likewise an all-time assumption on the executed path; successful verification at one step does not guarantee the existence of an admissible candidate or fallback at every later step.

The finite-horizon proof above already gives the scalar covariance comparison in Eq.~\eqref{eq:quartic_trace_recursion}. Under the uniform bounds of Assumption~\ref{ass:uniform_tube_regime}, the same comparison becomes
\begin{equation}
\operatorname{tr}(\Xi_k)
\leq
\left(
\rho_\Phi^2+\|HF\|_2^2\delta_{\rm adm}^2
\right)
\operatorname{tr}(\Xi_{k-1})
+
\bar s^\star\delta_{\rm adm}^2.
\label{eq:uniform_trace_comparison}
\end{equation}

The coefficient multiplying the previous mismatch contains two contributions: \(\rho_\Phi^2\) bounds propagation through the reference closed loop, while \(\|HF\|_2^2\delta_{\rm adm}^2\) accounts for feedback of the existing covariance mismatch through the implemented innovation covariance. A uniform-in-time bound therefore follows when
\begin{equation}
\rho_\Phi^2+\|HF\|_2^2\delta_{\rm adm}^2<1.
\label{eq:uniform_tube_condition}
\end{equation}
For \(\|HF\|_2>0\), this sufficient condition is equivalently
\begin{equation}
\delta_{\rm adm}
<
\frac{\sqrt{1-\rho_\Phi^2}}{\|HF\|_2}.
\label{eq:uniform_tube_defect_condition}
\end{equation}
Condition~\eqref{eq:uniform_tube_condition} is sufficient for the scalar comparison and is not asserted to be a necessary stability condition for the executed recursion.

\begin{corollary}[Uniform-in-time covariance and estimator-reference mean-square tubes]
\label{cor:uniform_time_tubes}
Under Assumption~\ref{ass:uniform_tube_regime} and condition~\eqref{eq:uniform_tube_condition},
\begin{equation}
\sup_{k\geq0}\|\Xi_k\|_F
\leq
\sup_{k\geq0}\operatorname{tr}(\Xi_k)
\leq
\frac{\bar s^\star\delta_{\rm adm}^2}
{1-\rho_\Phi^2-\|HF\|_2^2\delta_{\rm adm}^2}.
\label{eq:uniform_recursive_tube}
\end{equation}
If, in addition, the hypotheses of Proposition~\ref{prop:statistical_identification} and Assumption~\ref{ass:matched_initialization} hold consistently on every finite prefix, then
\begin{equation}
\sup_{k\geq0}
\E\!\left[\|e_k\|_2^2\mid\cG\right]
=
\sup_{k\geq0}\operatorname{tr}(\Xi_k)
\leq
\frac{\bar s^\star\delta_{\rm adm}^2}
{1-\rho_\Phi^2-\|HF\|_2^2\delta_{\rm adm}^2}.
\label{eq:uniform_state_tube}
\end{equation}
Here conditioning on the execution environment \(\cG\) means holding the computational execution fixed while the physical system remains random.
\end{corollary}

\begin{proof}
Since matched initialization gives \(\Xi_0=0\), Proposition~\ref{prop:exact_covariance_recursions} gives \(\Xi_k\succeq0\) for every \(k\geq0\). Equation~\eqref{eq:uniform_trace_comparison} therefore applies at every step. Since \(\operatorname{tr}(\Xi_0)=0\), iteration gives
\begin{equation}
\operatorname{tr}(\Xi_k)
\leq
\bar s^\star\delta_{\rm adm}^2
\sum_{j=0}^{k-1}
\left(
\rho_\Phi^2+\|HF\|_2^2\delta_{\rm adm}^2
\right)^j.
\end{equation}
Condition~\eqref{eq:uniform_tube_condition} makes this geometric series uniformly bounded, yielding
\begin{equation}
\operatorname{tr}(\Xi_k)
\leq
\frac{\bar s^\star\delta_{\rm adm}^2}
{1-\rho_\Phi^2-\|HF\|_2^2\delta_{\rm adm}^2}.
\end{equation}
Because \(\Xi_k\succeq0\), \(\|\Xi_k\|_F\leq\operatorname{tr}(\Xi_k)\), proving Eq.~\eqref{eq:uniform_recursive_tube}. Under the hypotheses of Proposition~\ref{prop:statistical_identification} and Assumption~\ref{ass:matched_initialization}, applied consistently on every finite prefix, Eq.~\eqref{eq:covariance_mismatch_statistical} applies at every finite \(k\), so \(\Xi_k=\E[e_ke_k^{\mathsf T}\mid\cG]\). Taking the trace and then the supremum proves Eq.~\eqref{eq:uniform_state_tube}.
\end{proof}

For a time-invariant reference loop, strict Euclidean contraction excludes some Schur-stable systems with substantial nonnormal transient amplification, so these bounds are conservative sufficient extensions of the finite-horizon theory, not runtime certificates or bounds on infinite-horizon accumulated task risk.

\section{Repairability Geometry and Learned-Correction Construction}
\label{app:ml_and_numer_method}

\subsection{Repairability boundary and minimum correction budget}
\label{app:repair_bound_and_minim_correc}

\begin{proof}[Proof of Theorem~\ref{thm:repairability}]
For every \(\Delta\in\mathcal V_{\rm corr}\), Frobenius orthogonality gives
\[
\|A_k+\Delta\|_F^2
=
\|A_k^\perp\|_F^2
+
\|A_k^\parallel+\Delta\|_F^2.
\]
Minimizing the second term over \(\|\Delta\|_F\le\delta_c\) amounts to projecting \(-A_k^\parallel\) onto the Frobenius ball of radius \(\delta_c\). One minimizing correction is
\[
\Delta_k^{\rm opt}
=
\begin{cases}
0,
&
\|A_k^\parallel\|_F=0,
\\[1mm]
-\min\!\left\{
1,\,
\dfrac{\delta_c}{\|A_k^\parallel\|_F}
\right\}A_k^\parallel,
&
\|A_k^\parallel\|_F>0.
\end{cases}
\]
Therefore
\[
\min_{\Delta\in\mathcal C_k}
\|A_k+\Delta\|_F^2
=
\|A_k^\perp\|_F^2
+
\bigl(\|A_k^\parallel\|_F-\delta_c\bigr)_+^2,
\]
which is the stated expression for \(\varepsilon_k^2\). Because the minimum is attained at \(\Delta_k^{\rm opt}\), a correction with \(\|A_k+\Delta\|_F\le\delta_{\rm adm}\) exists if and only if \(\varepsilon_k\le\delta_{\rm adm}\).
\end{proof}

The minimum correction radius required to enter the local admissible ball follows directly from Theorem~\ref{thm:repairability}. For a raw computational defect \(A_k\), the smallest correction radius for which local admissibility can be reached is \(\delta_{c,k}^{\min}=+\infty\) if \(\|A_k^\perp\|_F>\delta_{\rm adm}\), and otherwise
\[
\delta_{c,k}^{\min}
=
\left[
\|A_k^\parallel\|_F
-
\sqrt{
\delta_{\rm adm}^2-\|A_k^\perp\|_F^2
}
\right]_+ .
\]

Indeed, by the orthogonal decomposition above, if \(\|A_k^\perp\|_F>\delta_{\rm adm}\), no correction in the declared subspace can produce a locally admissible gain. Otherwise, admissibility requires
\[
\|A_k^\parallel+\Delta\|_F
\le
\sqrt{
\delta_{\rm adm}^2-\|A_k^\perp\|_F^2
},
\]
and $\delta_{c,k}^{\min}$ is the distance from the origin to the ball centered at \(-A_k^\parallel\) with radius \(\sqrt{\delta_{\rm adm}^2-\|A_k^\perp\|_F^2}.\)
Consequently,
\[
\varepsilon_k\le\delta_{\rm adm}
\quad\Longleftrightarrow\quad
\delta_{c,k}^{\min}\le\delta_c.
\]

The four local regimes therefore satisfy
\[
\text{subspace-obstructed}
\Longleftrightarrow
\delta_{c,k}^{\min}=+\infty,
\]
and
\[
\text{budget-obstructed}
\Longleftrightarrow
\delta_c<\delta_{c,k}^{\min}<+\infty.
\]

\subsection{Correction basis and learned model}
\label{app:correc_basis_and_learn_model}

The correction basis is constructed independently at each solver checkpoint \(t\) from the corresponding offline local computational-defect samples. Let \(A_{q,t}=K_{q,t}^{\rm alg}-K_q^{\rm loc}\) denote the raw local computational defect on basis-construction sample \(q\), and form the raw sample matrix
\[
Z_t
=
\begin{bmatrix}
\operatorname{vec}_{\rm row}(A_{1,t}) &
\cdots &
\operatorname{vec}_{\rm row}(A_{N_t^{\rm acq},t})
\end{bmatrix}
\in\mathbb R^{nm\times N_t^{\rm acq}}.
\]
No centering or coordinate normalization is applied before the decomposition
\[
Z_t
=
\mathsf U_t\Sigma_t\mathsf V_t^{\mathsf T}.
\]
Let \(\nu_t\) denote the numerical rank under a fixed singular-value threshold and define \(d_{{\rm eff},t}=\min\{d_{\rm cap},\nu_t\}\). Only numerically supported singular directions are retained. For \(1\le i\le d_{{\rm eff},t}\), the correction-basis matrices are
\[
U_{t,i}
=
\operatorname{unvec}_{\rm row}
\!\left([\mathsf U_t]_{:,i}\right),
\]
and satisfy
\[
\langle U_{t,i},U_{t,j}\rangle_F
=
\delta_{ij}.
\]
The basis is not completed within the numerical null space.

Residual features use a separate checkpoint-specific scaling fitted on the designated fit window. For \(z_{q,t} = \operatorname{vec}_{\rm row}(\rho_{q,t}^{\rm alg}),\) define the componentwise RMS amplitudes
\[
r_{t,j} :=
\left(
\frac{1}{N_t^{\rm fit}}
\sum_{q=1}^{N_t^{\rm fit}}
[z_{q,t}]_j^2
\right)^{1/2}.
\]
The frozen diagonal feature scale is
\[
[R_t^{\rm feat}]_{jj}
=
\max\{r_{t,j},f_t\},
\]
where \(f_t>0\) is a fixed machine-scale floor used to prevent vanishing feature scales, and the feature vector is \(\phi_{q,t}=(R_t^{\rm feat})^{-1}z_{q,t}\).
The absence of centering preserves the zero-residual to zero-feature relation used in Sec.~\ref{subsec:bound_ml_correc}.

The exact local gain enters only during offline construction of the training targets. For a training sample,
\[
A_{k,t}
=
K_{k,t}^{\rm alg}-K_k^{\rm loc},
\]
and the supported coordinates are
\[
a_{k,t}
=
\begin{bmatrix}
\langle U_{t,1},A_{k,t}\rangle_F &
\cdots &
\langle U_{t,d_{{\rm eff},t}},A_{k,t}\rangle_F
\end{bmatrix}^{\mathsf T} \in
\R^{d_{{\rm eff},t}}.
\]
The checkpoint-local correction radius \(\delta_{c,t}^{\star}\) is calibrated offline from a prescribed upper quantile of \(\|a_{k,t}\|_2\) on a disjoint calibration window. The ridge-training target is the radial projection
\[
c_{k,t}^{\rm tar}
=
\operatorname{proj}_{\|c\|_2\le\delta_{c,t}^{\star}}
(a_{k,t}) \in\R^{d_{{\rm eff},t}}.
\]
For each supported checkpoint, the no-intercept ridge map is fitted on the corresponding fit samples according to
\begin{align*}
\Theta_t
\in
\arg\min_{\Theta}
\Biggl\{
&\frac{1}{N_t^{\rm fit}}
\sum_{j=1}^{N_t^{\rm fit}}
\|\Theta\phi_{j,t}-c_{j,t}^{\rm tar}\|_2^2
\\
&+
\lambda\|\Theta\|_F^2
\Biggr\},
\qquad
\Theta\in\R^{d_{{\rm eff},t}\times nm},
\end{align*}
where \(\lambda>0\) and \(N_t^{\rm fit}\) is the number of fit samples.

The basis, feature scales, correction radius, and ridge map are frozen before deployment. For the recursive CG experiments, \(d_{\rm cap}=64\), the correction radius uses the \(0.95\) nearest-rank quantile of the checkpoint-local projected defect amplitudes, and the ridge parameter is \(\lambda=10^{-2}\). For the standard 400-step commissioning protocol, the basis-construction, calibration, and fit windows contain 100, 100, and 200 steps, respectively; the proportional-commissioning horizon study preserves the same \(1:1:2\) partition as the commissioning length is varied. The deployed feature map, coefficient projection, and bounded gain correction are defined in Sec.~\ref{subsec:bound_ml_correc}. The offline target is constructed to cancel the supported component of the local computational defect subject to the correction-radius constraint; it is not chosen to minimize either the represented residual norm or the finite-horizon task-response functional.

\section{Numerical Methods and Supporting Results}

\subsection{Benchmark construction}
\label{app:benchmark_construction}

The benchmark network data are the unmodified IEEE cases distributed with pandapower~3.5.4~\cite{pandapower}. Each case has a unique external-grid bus, which defines the reference bus. For an \(N\)-bus case, the state uses reduced rectangular voltage coordinates: all \(N\) real voltage components and all imaginary components except that of the reference bus, giving \(n=2N-1\). The measurement pool contains these voltage coordinates together with the real and imaginary branch-current coordinates at both ends of each in-service branch, obtained from the branch-admittance matrices. Every scalar measurement row is normalized to unit Euclidean norm.

Within each recursive realization, the model instance is fixed. In every recursive study, a fixed permutation of the complete measurement pool defines nested models: \(H\) at dimension \(m\) contains the first \(m\) rows, and \(R\) is the corresponding leading principal submatrix of a covariance constructed for the complete ordered pool. Thus \(m\) denotes the number of scalar measurement coordinates rather than the number of physical measurement devices. Measurement-noise variances use nominal scales \(2.5\times10^{-3}\) for voltage coordinates and \(6.0\times10^{-3}\) for branch-current coordinates, multiplied by independent log-normal factors with log standard deviation \(0.35\). A common correlation coefficient drawn in \([0,0.15)\) couples the real and imaginary coordinates associated with the same bus voltage or branch-end current; other off-diagonal correlations vanish. The resulting \(R\) is fixed over the recursive chronology and is a synthetic benchmark covariance rather than a field-calibrated sensor-noise model.

The state-transition matrix is synthetic but topology informed. The network graph Laplacian is constructed from the in-service lines and transformers, normalized by its largest eigenvalue, duplicated over the real and imaginary voltage coordinates, and reduced by removing the reference-bus imaginary coordinate. The transition matrix combines an identity contribution with coefficient drawn uniformly from \([0.94,0.985]\), topology-dependent Laplacian damping with coefficient drawn uniformly from \([0.01,0.08]\), and a small random skew-symmetric perturbation generated from Gaussian entries with standard deviation \(0.003\).

The base process covariance \(Q_{\rm base}\) combines heterogeneous diagonal variances with a positive-semidefinite correlated component of rank at most four. Its overall process-noise amplitude is drawn log-uniformly from \([10^{-3},8\times10^{-3}]\) and is squared in forming the covariance; the diagonal factors have log standard deviation \(0.25\), and the correlated-component strength is drawn uniformly from \([0,0.5]\). For transition \(k\to k+1\) during a recursive realization,
\[
Q_k=c_kQ_{\rm base},
\]
where \(\log c_k\) follows a pre-generated first-order autoregressive process with coefficient \(0.98\) and innovation standard deviation \(0.03\), bounded to \([-0.25,0.25]\) at each step. Hence \(c_k\in[e^{-0.25},e^{0.25}].\)
The complete process-covariance schedule is generated before execution and is shared by the compared policies.

For the recursive CG, runtime, and signed-quartic studies, the covariance supplied to the first recorded measurement update is
\[
P_{\rm init}^{-}
=
\frac{\operatorname{tr}(Q_{\rm base})}{n}I,
\]
with no Riccati burn-in. No additional prediction is applied before this first update: \(P_{\rm init}^{-}\) is the common supplied prediction prior at this boundary, with the prediction recursions applied from the next step. For state-trajectory evaluations, the simulated initial state is zero-mean Gaussian with covariance \(P_{\rm init}^{-}\), the reference estimate starts at zero, and compared policies within a configuration share the same physical state and noise realization. After the shared reference posterior at the commissioning stage, a common prediction supplies the first deployment prior.

\subsection{Classical solver implementations}
\label{app:classi_solve_implem}

The classical experiments use zero-start, unpreconditioned CG independently on the \(n\) transposed gain-system right-hand sides of Eq.~\eqref{eq:column_gain_systems} in the main text.
Let \(x_k^{(j)}\) denote the exact solution and \(x_{k,t}^{(j)}\) the order-\(t\) CG approximation. The returned vectors are assembled into \((K_{k,t}^{\rm alg})^{\mathsf T}\). An iteration \(t\) denotes the represented gain associated with the prescribed order-\(t\) CG approximation path.

In exact arithmetic, CG minimizes the solution error in the \(S_k\)-energy norm over the corresponding Krylov approximation space,
\[
x_{k,t}^{(j)}
\in
\arg\min_{x\in\mathcal K_{k,t}^{(j)}}
\|x_k^{(j)}-x\|_{S_k},
\qquad
\|y\|_{S_k}:=(y^{\mathsf T}S_ky)^{1/2},
\]
where, for the zero initial guess and \(r_{k,0}^{(j)}=[H\widetilde P_k]_{:,j}\),
\[
\mathcal K_{k,t}^{(j)}
:=
\operatorname{span}
\left\{
r_{k,0}^{(j)},
S_kr_{k,0}^{(j)},
\ldots,
S_k^{t-1}r_{k,0}^{(j)}
\right\},
\qquad t\ge1,
\]
with \(\mathcal K_{k,0}^{(j)}:=\{0\}\).
Because these Krylov spaces are nested,
\[
\|x_k^{(j)}-x_{k,t+1}^{(j)}\|_{S_k}
\le
\|x_k^{(j)}-x_{k,t}^{(j)}\|_{S_k}.
\]
Moreover, in exact arithmetic CG recovers the solution in at most \(m\) iterations. Its standard worst-case convergence estimate is
\[
\|x_k^{(j)}-x_{k,t}^{(j)}\|_{S_k}
\le
2
\left(
\frac{\sqrt{\kappa(S_k)}-1}
{\sqrt{\kappa(S_k)}+1}
\right)^t
\|x_k^{(j)}-x_{k,0}^{(j)}\|_{S_k},
\]
where \(\kappa(S_k) = \frac{\sigma_{\max}(S_k)} {\sigma_{\min}(S_k)}.\)
This is a worst-case energy-norm bound and does not imply monotonic decrease of the represented matrix residual at every iteration~\cite{greenbaum1997iterative,saad2003iterative}.

Under the conventional CG residual convention,
\[
r_{k,t}^{(j)}
:=
[H\widetilde P_k]_{:,j}
-
S_kx_{k,t}^{(j)},
\]
the corresponding column of the transposed matrix residual satisfies
\[
[(\rho_{k,t}^{\rm alg})^{\mathsf T}]_{:,j}
=
-r_{k,t}^{(j)}
\]
in exact arithmetic. The residual used by the execution interface is nevertheless recomputed from the assembled represented gain,
\[
\rho_{k,t}^{\rm alg}
=
K_{k,t}^{\rm alg}S_k-\widetilde P_kH^{\mathsf T},
\]
rather than taken solely from the recurrence-updated CG residual. This keeps subsequent evaluation tied to the actual returned gain.

Computational work is accounted for separately from the iterations. For an explicitly stored dense \(S_k\), an order-\(t\) CG path over all \(n\) right-hand sides has leading matrix--vector cost \(O(tnm^2)\) in the absence of early convergence, while forming the represented matrix residual requires \(O(nm^2)\) operations once the gain system has been constructed. A dense Cholesky computation of the exact local gain requires an \(O(m^3)\) factorization followed by \(n\) triangular solves with total cost \(O(nm^2)\). In a matrix-free implementation, the corresponding costs are instead determined by applications of the innovation operator. These operation counts describe individual computational components and do not by themselves establish an advantage for truncated CG.

\subsection{Runtime comparison protocol}
\label{app:runtime_compar_proto}

The policies in Table~\ref{tab:runtime-complete} begin from the same state estimate and covariance at the end of commissioning and process the same observation sequence. Each policy subsequently evolves its own state and covariance through its executed gains. The comparison therefore imposes a common admissibility requirement without requiring identical gain systems or equal achieved accuracy. Every candidate is checked by the represented-residual certificate, with verified fallback invoked upon rejection.

Selected M-CG uses a fixed number of CG iterations, while selected LC-CG applies the commissioned bounded correction to its CG candidate. The iteration counts are selected from the previously evaluated certified frontiers and fixed before timing. ``M-CG at LC-CG iteration" uses the selected LC-CG iteration count without learned correction. No iteration count or model parameter is adjusted using the timing results.

Warm-start CG solves the $n$ systems associated with $S_k K^{\mathsf T}=H\widetilde P_k$, using the previous executed gain as the initial guess. At the first deployment step, this guess is the final commissioning gain. Each system uses zero relative tolerance, an absolute residual tolerance of $\ell_k\delta_{\rm adm}/\sqrt{n}$, and at most $m=64$ iterations. Solver termination is followed by independent recomputation and certification of the complete gain residual. Warm-start PCG uses the same initialization and stopping criteria, together with a fixed Cholesky preconditioner constructed from the final commissioning innovation covariance. Its factorization is charged to offline setup, while its application through triangular solves is included in online runtime.

The innovation-form direct method computes a Cholesky factorization of the current $S_k$ and solves all gain right-hand sides together. The information-form method precomputes
\[
G=H^{\mathsf T}R^{-1}, \qquad \Lambda_0=GH
\]
using Cholesky solves with the fixed measurement covariance $R$. This reusable computation is charged to offline setup. At each deployment step, the candidate gain $L_k$ is obtained by solving
\[
\bigl(I_n+\widetilde P_k \Lambda_0\bigr)L_k=\widetilde P_k G
\]
through a general lower--upper factorization, since the coefficient matrix is not necessarily symmetric. Both direct candidates are verified against the current represented innovation system by the same residual certificate used for approximate candidates and for verified fallback.

Matrix computations use double-precision floating-point arithmetic, with execution restricted to one logical processor and single-threaded numerical libraries. Each policy receives an untimed eight-step warm-up before deployment. Seven timing repetitions are performed from the same commissioning handoff, with policy order cyclically rotated across repetitions. Online runtime includes initialization and all operations within the deployment steps, with the trained model already loaded in memory. Post-step diagnostic recording, file input and output, and external orchestration are excluded. Table~\ref{tab:runtime-complete} reports these repeated timings and additional reusable setup costs. The component measurements in Table~\ref{tab:runtime-decomposition} come from separate profiling runs with nonoverlapping timing categories, so their totals need not equal the repeated-runtime means. Common offline preparation and historical frontier-search costs are excluded.

\subsection{Quantum solver implementations}
\label{app:quantu_solve_implem}

\subsubsection{VQLS reconstruction}
\label{app:vqls_reconstr_and_solver_diagn}

A VQLS candidate enters the framework through the represented-gain interface of Sec.~\ref{subsec:solver_residual_interface}. For the \(m=2\) recursive realizations reported in Sec.~\ref{subsec:quantu_solve_candi}, no dimensional padding is required, and the 40-step commissioning interval is partitioned into 10 steps for correction-basis construction, 10 for capacity calibration of the correction radius, and 20 for learner fitting. The Statevector and finite-shot Sampler realizations scale \(S_k\) by \(\|S_k\|_2\), normalize each nonzero right-hand side, and use the one-parameter ansatz \(\lvert\psi(\theta)\rangle=R_Y(\theta)\lvert0\rangle\) with the global VQLS objective and COBYLA \cite{powell1994cobyla} with at most 100 objective evaluations for each nonzero right-hand side. The Statevector objective is evaluated exactly with optimizer tolerance \(10^{-8}\); the finite-shot Sampler uses 64 shots for each \(X\)- and \(Z\)-measurement setting at every stochastic objective evaluation and optimizer tolerance \(10^{-4}\). After COBYLA terminates, the terminal parameter \(\theta^\star\) is frozen and the terminal variational state is reconstructed deterministically; no additional terminal Sampler measurement is performed. For the Origin realization, fresh finite-shot Sampler runs with COBYLA optimization provide the terminal variational parameters, after which fresh terminal \(X\)- and \(Z\)-measurements are executed on the Origin Wukong processor (\texttt{WK\_C180}) using 64 shots per setting; no variational optimization is performed on hardware.

For all three realizations, the 40-step commissioning interval follows the reference covariance chronology. Approximate gains are used to construct the realization-specific correction data but are not propagated into the commissioning covariance. The resulting corrector is frozen before deployment. For each nonzero right-hand side \(y_k^{(j)}:=[H\widetilde P_k]_{:,j}\), let \(\psi_k^{(j)}\in\mathbb R^m\) denote the reconstructed unit solution direction. The least-squares amplitude and reconstructed approximate solution are
\[
\alpha_k^{(j)}
=
\frac{\bigl(S_k\psi_k^{(j)}\bigr)^{\mathsf T}y_k^{(j)}}
{\bigl\|S_k\psi_k^{(j)}\bigr\|_2^2},
\qquad
x_k^{(j)}=\alpha_k^{(j)}\psi_k^{(j)} .
\]
For the Statevector and finite-shot Sampler realizations, the solution direction is obtained from the frozen terminal variational state. For the Origin realization, it is reconstructed from the terminal hardware \(X\)- and \(Z\)-measurement data. A zero right-hand side gives the zero solution analytically. The \(n\) reconstructed solutions are then assembled into \(K_k^{\rm alg}\), after which the represented matrix residual is recomputed from the complete gain. Correction, certification, fallback, and recursive execution are applied only at this complete-gain level.

\subsubsection{Binary encoding of the QUBO formulation}
\label{app:binary_encod_of_the_qubo}

Here we give the QUBO encoding of the current implemented gain system used in Sec.~\ref{subsubsec:qubo_dwave}. The QUBO construction uses the same transposed local-gain systems defined in Eq.~\eqref{eq:column_gain_systems} in the main text, equivalently \(K_k^{\rm loc}S_k=\widetilde P_kH^{\mathsf T}\). The system is separable over the \(n\) columns of \((K_k^{\rm loc})^{\mathsf T}\).
Define
\[
x_k^{(j)}
:=
[(K_k^{\rm loc})^{\mathsf T}]_{:,j}
\in\R^m,
\qquad
y_k^{(j)}
:=
[H\widetilde P_k]_{:,j}
\in\R^m.
\]
Then
\[
S_k x_k^{(j)}
=
y_k^{(j)},
\qquad
j=1,\ldots,n,
\]
and, since \(S_k\succ0\),
\begin{equation}
x_k^{(j)}
=
\arg\min_{x\in\R^m}
\|S_kx-y_k^{(j)}\|_2^2.
\label{eq:colwise_ls}
\end{equation}

Let \(n_{\rm bit}\ge1\) denote the bit depth used to encode each scalar component. For the \(j\)-th right-hand side at step \(k\), choose a scalar encoding interval
\[
[x_{\min,k}^{(j)},x_{\max,k}^{(j)}],
\]
and define
\[
h_k^{(j)}
:=
\frac{
x_{\max,k}^{(j)}-x_{\min,k}^{(j)}
}{
2^{n_{\rm bit}}-1
}.
\]
This gives \(2^{n_{\rm bit}}\) equally spaced scalar levels on the prescribed interval. Applying the encoding independently to the \(m\) components gives
\begin{equation}
[x_{k,n_{\rm bit}}^{(j)}(\boldsymbol{\vartheta})]_i
=
x_{\min,k}^{(j)}
+
h_k^{(j)}
\sum_{\nu=0}^{n_{\rm bit}-1}
2^\nu\vartheta_{i\nu},
\qquad
i=1,\ldots,m,
\label{eq:binary_enc}
\end{equation}
with \(\vartheta_{i\nu}\in\{0,1\}\).

The encoding interval needs to be wide enough to contain the exact solution if range truncation is to be excluded. Since
\[
x_k^{(j)}
=
S_k^{-1}y_k^{(j)},
\]
we have
\begin{equation}
\|x_k^{(j)}\|_\infty
\le
\|x_k^{(j)}\|_2
\le
\frac{\|y_k^{(j)}\|_2}{\sigma_{\min}(S_k)}.
\label{eq:gain_range_bound}
\end{equation}
In exact arithmetic, \(S_k\succeq\underline r I_m\) therefore gives the conservative bound
\begin{equation}
\gamma_{k,\rm cons}^{(j)}
:=
\frac{\|y_k^{(j)}\|_2}{\underline r},
\label{eq:conservative_encoding_range}
\end{equation}
so that the symmetric choice
\[
x_{\min,k}^{(j)}
=
-\gamma_{k,\rm cons}^{(j)},
\qquad
x_{\max,k}^{(j)}
=
\gamma_{k,\rm cons}^{(j)}
\]
contains every component of \(x_k^{(j)}\). For a represented floating-point construction, the corresponding guaranteed interval requires a separately verified positive lower bound for the stored \(S_k\), as in Sec.~\ref{subsec:residual_certifi_and_fail_close}. A conservative interval can be wider than necessary, which increases the grid spacing at fixed \(n_{\rm bit}\). Increasing \(n_{\rm bit}\) reduces this representation spacing but also increases the number of binary variables; it improves encoding resolution, not necessarily the quality of a finite-time QUBO optimizer.

For the remainder of the construction, fix \(k\) and \(j\). Stack the \(mn_{\rm bit}\) binary variables into
\[
\boldsymbol{\vartheta} \in \{0,1\}^{mn_{\rm bit}}.
\]
Define the coding matrix
\[
\mathsf C_{n_{\rm bit}}
:=
I_m
\otimes
(1,2,4,\ldots,2^{n_{\rm bit}-1})
\in
\R^{m\times mn_{\rm bit}}.
\]
Then
\[
x_{k,n_{\rm bit}}^{(j)}(\boldsymbol{\vartheta})
=
x_{\min,k}^{(j)}\mathbf 1_m
+
h_k^{(j)}
\mathsf C_{n_{\rm bit}}
\boldsymbol{\vartheta}.
\]
Define
\[
\zeta_k^{(j)}
:=
y_k^{(j)}
-
S_kx_{\min,k}^{(j)}\mathbf 1_m
\in\R^m.
\]
Substitution into Eq.~\eqref{eq:colwise_ls} then gives
\[
S_kx_{k,n_{\rm bit}}^{(j)}(\boldsymbol{\vartheta})
-
y_k^{(j)}
=
h_k^{(j)}
S_k\mathsf C_{n_{\rm bit}}\boldsymbol{\vartheta}
-
\zeta_k^{(j)}.
\]
Hence, using \(\vartheta_r^2=\vartheta_r\) for binary variables,
\begin{align}
\min_{\boldsymbol{\vartheta}\in\{0,1\}^{mn_{\rm bit}}}
&
\left\|
h_k^{(j)}
S_k\mathsf C_{n_{\rm bit}}\boldsymbol{\vartheta}
-
\zeta_k^{(j)}
\right\|_2^2
\nonumber\\
=
\min_{\boldsymbol{\vartheta}\in\{0,1\}^{mn_{\rm bit}}}
&
\Biggl[
\boldsymbol{\vartheta}^{\mathsf T}
\Bigl(
(h_k^{(j)})^2
\mathsf C_{n_{\rm bit}}^{\mathsf T}
S_k^{\mathsf T}S_k
\mathsf C_{n_{\rm bit}}
\nonumber\\
&
\;
-
\operatorname{diag}
\!\bigl(
2h_k^{(j)}
\mathsf C_{n_{\rm bit}}^{\mathsf T}
S_k^{\mathsf T}\zeta_k^{(j)}
\bigr)
\Bigr)
\boldsymbol{\vartheta}
+
\|\zeta_k^{(j)}\|_2^2
\Biggr].
\label{eq:qubo_obj}
\end{align}
Equation~\eqref{eq:qubo_obj} is a quadratic unconstrained binary optimization problem in \(mn_{\rm bit}\) binary variables for each right-hand side. Writing the symmetric matrix multiplying \(\boldsymbol{\vartheta}\) in Eq.~\eqref{eq:qubo_obj} as \(Q_{\rm sym}\), a representation that stores each pair \(r<s\) only once uses linear coefficient \([Q_{\rm sym}]_{rr}\) for variable \(r\) and pairwise coefficient \(2[Q_{\rm sym}]_{rs}\) for \(r<s\), with \(\|\zeta_k^{(j)}\|_2^2\) as the constant offset. Constructing a complete gain candidate requires obtaining and decoding the corresponding solution for all \(n\) right-hand sides. After decoding, the \(n\) solution vectors are assembled into \(K_k^{\rm alg}\), and the represented matrix residual \(\rho_k^{\rm alg}\) is recomputed independently.

\begin{figure*}[t]
  \centering
  \begin{overpic}[width=\textwidth]{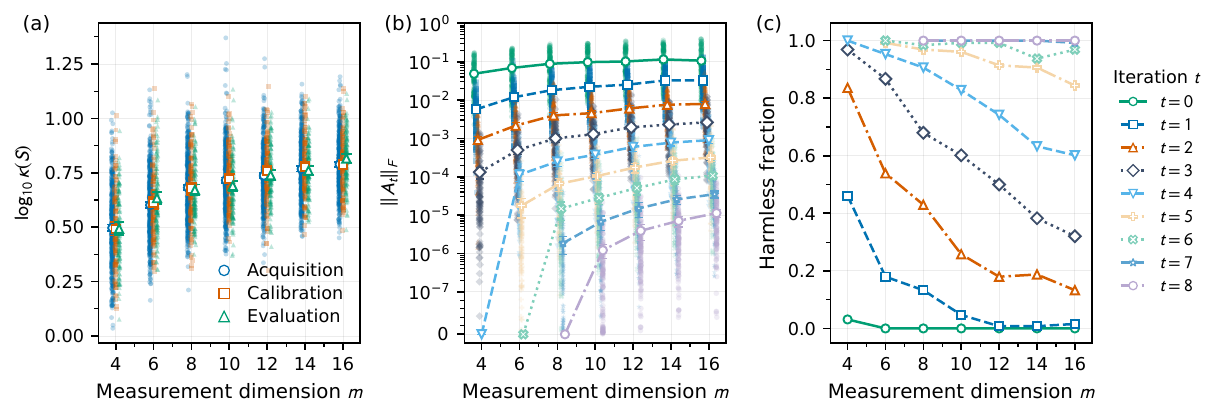}
  \end{overpic}
  \caption{
  Local numerical landscape of finite-CG gain computation for a set of IEEE 14-derived local instances based on the pandapower benchmark \cite{pandapower}. (a) Innovation-system conditioning \(\log_{10}\kappa(S)\) across acquisition (circles), calibration (squares), and evaluation (triangles) instances for \(m\in\{4,6,8,10,12,14,16\}\). Markers denote the mean \(\log_{10}\kappa(S)\) over the corresponding acquisition, calibration, or evaluation instances at each \(m\), and error bars show \(95\%\) confidence intervals based on the Student distribution. (b) Solver-induced local computational defect \(\|A_t\|_F=\|K_t^{\mathrm{alg}}-K^{\mathrm{loc}}\|_F\) along zero-start incremental-CG trajectories at the available checkpoints. The vertical axis is linear up to $10^{-7}$ and logarithmic above, so that numerically near-zero defects, as observed for \(m=t=4\), can be shown. (c) Empirical fraction of evaluation instances whose raw computational defect lies within the corresponding predeclared local classification radius and is therefore classified as harmless (cf.~Table~\ref{tab:regions}).
  }
  \label{fig:local_solver_landscape}
\end{figure*}

For the D-Wave realization, each right-hand-side problem uses the symmetric interval \([-\gamma_k^{(j)},\gamma_k^{(j)}]\), where \(\gamma_k^{(j)}\) is an outward-rounded upper bound on \(\|y_k^{(j)}\|_2/\ell_k\) obtained from the verified lower bound for the represented \(S_k\). Each eight-variable logical QUBO is sampled on the D-Wave Advantage2 System JUPSI (\(\texttt{Advantage2\_system2}\)) located at Forschungszentrum Jülich, using a fixed embedding into 12 physical qubits, 1000 reads, a \(20\,\mu{\rm s}\) annealing time, automatic scaling, uniform-torque-compensation chain strength with prefactor \(1.414\), and majority-vote unembedding. From each returned sample set, the decoded solution is selected solely by the minimum energy recomputed under the original unscaled logical QUBO, with exact ties resolved lexicographically. The nine selected right-hand-side solutions are assembled into the complete gain before its represented residual is recomputed.

\subsection{Supporting numerical results}
\label{app:supporting_numerical_results}

\paragraph{Local solver landscape.}

We characterize finite-CG gain computation on a set of IEEE 14-derived local instances \cite{pandapower}. Figure~\ref{fig:local_solver_landscape} shows innovation-system conditioning, local gain defects, and the fraction of evaluation instances within the corresponding predeclared local classification radius across measurement dimensions and CG checkpoints. This radius is used only to classify local gain defects.

\paragraph{Measurement and network trajectories.}

The frontier comparisons in Figs.~\ref{fig:measurement_scaling} and \ref{fig:network_resolved} are complemented in Figs.~\ref{fig:measurement_trajectory_response} and \ref{fig:network_trajectory_response} by the corresponding deployment trajectories, with each policy evaluated at its own selected zero-fallback iteration.

\begin{figure*}[t]
  \centering
  \includegraphics[width=\textwidth]{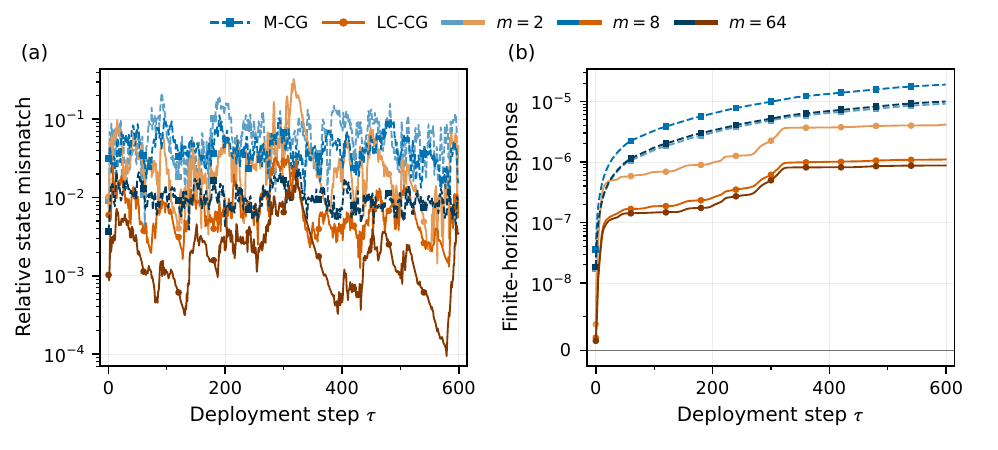}
  \caption{Deployment trajectories corresponding to Fig.~\ref{fig:measurement_scaling}, shown for \(m=2,8,64\). (a) Per-step relative state mismatch [Eq.~\eqref{eq:relative_state_mismatch}]. (b) Cumulative finite-horizon response.}
  \label{fig:measurement_trajectory_response}
\end{figure*}

\begin{figure*}[t]
  \centering
  \includegraphics[width=\textwidth]{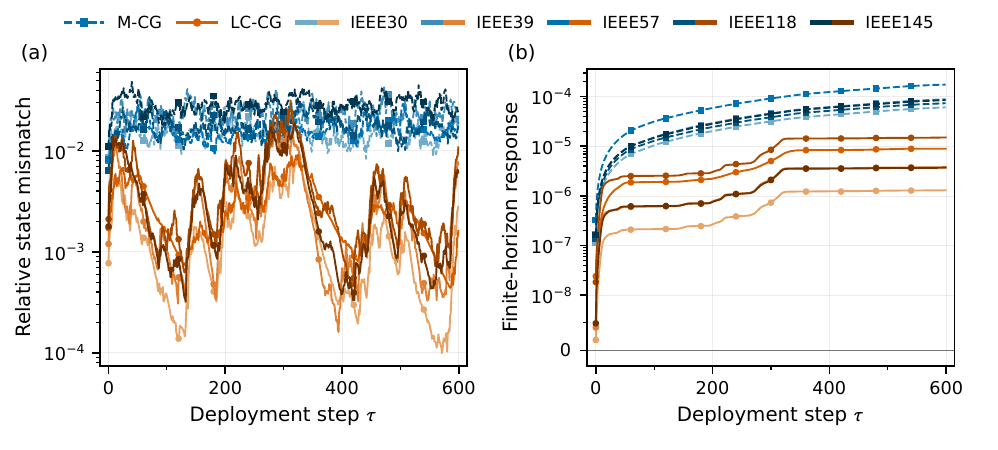}
  \caption{Deployment trajectories corresponding to Fig.~\ref{fig:network_resolved}, shown at \(m=64\) for IEEE 30, 39, 57, 118, and 145. (a) Per-step relative state mismatch [Eq.~\eqref{eq:relative_state_mismatch}]. (b) Cumulative finite-horizon response.}
  \label{fig:network_trajectory_response}
\end{figure*}

\paragraph{Runtime cost decomposition}

\begin{table*}[t]
\caption{
Profiled online-cost decomposition of the selected M-CG and LC-CG policies in the IEEE 14 tolerance study. Each column is a separate instrumented 600-step profiling pass; M$t$ and LC$t$ denote M-CG and LC-CG executed at iteration $t$. All entries are in seconds. Online setup and overhead combines initialization, innovation-problem formation, and online elapsed time not assigned to the remaining component timers. Candidate generation is the measured fixed-checkpoint candidate-construction scope; at LC0 no CG matrix--vector products are performed, and the small nonzero entry is due to zero-checkpoint candidate construction and validation. A dash denotes a component that is not executed. No verified fallback is invoked by the selected M-CG or LC-CG policies shown. Component times sum close to the complete profiling time before display rounding; because these are separate instrumented passes, their totals need not equal the seven-repeat means in Table~\ref{tab:runtime-complete}.
}
\label{tab:runtime-decomposition}
\begin{ruledtabular}
\begin{tabular}{l|rr|rr|rr|rr}
Component
& \multicolumn{2}{c|}{$\eta=10^{-3}$}
& \multicolumn{2}{c|}{$\eta=10^{-2}$}
& \multicolumn{2}{c|}{$\eta=10^{-1}$}
& \multicolumn{2}{c}{$\eta=0.2$} \\
& M10 & LC7
& M7 & LC4
& M3 & LC1
& M3 & LC0 \\
\colrule
Online setup and overhead
& 0.0898 & 0.0961
& 0.0903 & 0.0963
& 0.0899 & 0.0962
& 0.0900 & 0.0950 \\

Candidate generation
& 5.0639 & 3.5761
& 3.5723 & 2.0838
& 1.5830 & 0.5756
& 1.5676 & 0.0209 \\

Learner evaluation
& \textemdash & 0.0311
& \textemdash & 0.0317
& \textemdash & 0.0279
& \textemdash & 0.0251 \\

Correction construction and guarding
& \textemdash & 0.1790
& \textemdash & 0.1695
& \textemdash & 0.1687
& \textemdash & 0.1613 \\

Residual recomputation
& 0.0104 & 0.0190
& 0.0105 & 0.0190
& 0.0104 & 0.0188
& 0.0104 & 0.0184 \\

Certification
& 3.7792 & 3.7586
& 3.8122 & 3.8225
& 3.7829 & 3.7700
& 3.8722 & 3.8380 \\

Recursive update
& 0.1189 & 0.1201
& 0.1188 & 0.1199
& 0.1183 & 0.1189
& 0.1189 & 0.1188 \\
\colrule
Complete profiling time
& 9.0622 & 7.7802
& 7.6041 & 6.3427
& 5.5845 & 4.7762
& 5.6591 & 4.2775 \\
\end{tabular}
\end{ruledtabular}
\end{table*}

The component profiles show that, at each evaluated tolerance, the reduction in candidate-generation cost from the lower LC-CG checkpoint exceeds the additional learner-evaluation, correction-construction, and residual-recomputation costs. At \(\eta=0.2\), the selected LC policy operates at \(t=0\), so no CG matrix--vector products are performed. Learner evaluation, correction construction, and residual recomputation together account for \(4.79\%\) of the complete profiled online time, whereas independent certification accounts for \(89.73\%\). Certification is therefore the dominant remaining online cost in this zero-iteration regime.

\paragraph{VQLS residual quality and solver cost.}
\label{app:vqls_residual_quality_and_solver_cost}

To complement the recursive-execution results, we separately examine the residual quality and computational cost of the underlying innovation-space solver calls. Figure~\ref{fig:vqls_solver_diagnostics} compares classical CG with Statevector and finite-shot VQLS realizations using global and local VQLS objectives as the innovation dimension is varied. The timing reported here characterizes solver-level candidate generation and is distinct from the complete online runtime considered for the classical recursive policies in Sec.~\ref{subsec:runtime}.

\begin{figure*}[t]
  \centering
  \includegraphics[width=\textwidth]{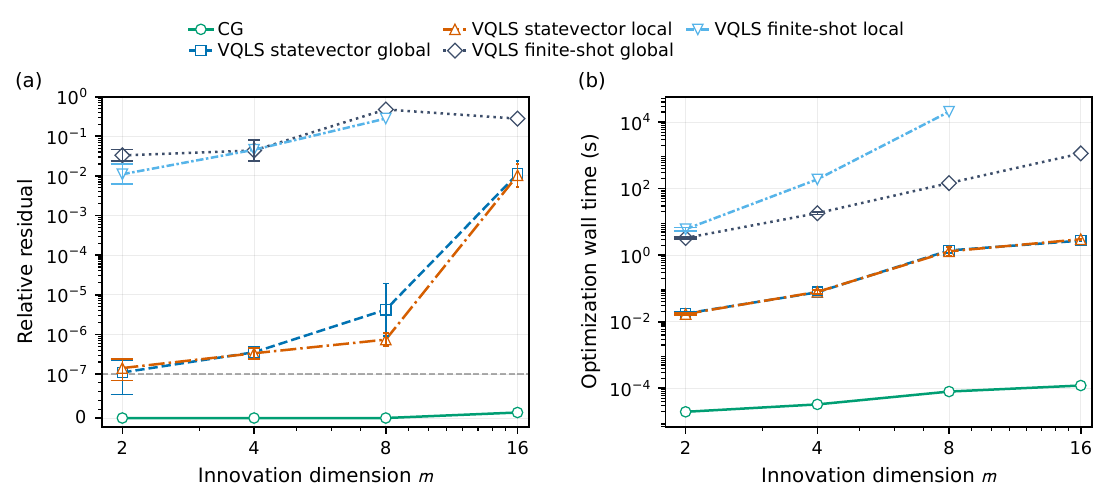}
  \caption{Single-right-hand-side solver comparison on IEEE 14-derived innovation systems as the innovation dimension \(m\) is varied, distinct from the complete-gain recursive experiments. Classical CG is compared with Statevector and finite-shot VQLS using global and local objectives. (a) Relative residual of the returned solution. (b) Wall time per solver invocation. VQLS markers show geometric means over optimizer restarts; error bars indicate \(95\%\) Student-\(t\) intervals in \(\log_{10}\) space. The horizontal \(10^{-7}\) line is a display guide. These implementation-specific timings are not end-to-end or resource-equivalent comparisons.}
  \label{fig:vqls_solver_diagnostics}
\end{figure*}

Figure~\ref{fig:vqls_solver_diagnostics} shows the solver-level residual quality and optimization wall time of the VQLS realizations relative to CG over the evaluated innovation dimensions; the interpretation of these implementation-specific timings is given in Sec.~\ref{subsec:quantu_solve_candi}.

% \nocite{*}

\bibliography{refs}% Produces the bibliography via BibTeX.

@article{ganeshamurthy2024bridging,
  author={Ganeshamurthy, Priyanka Arkalgud and Ghosh, Kumar and O'Meara, Corey and Cortiana, Giorgio and Schiefelbein-Lach, Jan and Monti, Antonello},
  journal={IEEE Access}, 
  title={Next Generation Power System Planning and Operation With Quantum Computation}, 
  year={2024},
  volume={12},
  number={},
  pages={182673-182692},
  doi={10.1109/ACCESS.2024.3509743}}

@article{golestan2023quantum,
  title={Quantum computation in power systems: An overview of recent advances},
  author={Golestan, Saeed and Habibi, MR and Mousavi, SY Mousazadeh and Guerrero, Josep M and Vasquez, Juan C},
  journal={Energy Reports},
  volume={9},
  pages={584--596},
  year={2023},
  publisher={Elsevier}
}

@ARTICLE{feng2024noisy,
  author={Feng, Fei and Zhang, Peng and Zhou, Yifan and Shamash, Yacov A.},
  journal={iEnergy}, 
  title={Noisy-intermediate-scale quantum power system state estimation}, 
  year={2024},
  volume={3},
  number={3},
  pages={135-141},
  doi={10.23919/IEN.2024.0019}}

@article{ji2026quantumdeep,
  title={Quantum Deep Learning: A Comprehensive Review},
  author={Ji, Yanjun and Chen, Zhao-Yun and Roth, Marco and Kreplin, David A and Schiffer, Christian and King, Martin and Anton, Oliver and Alam, M Sahnawaz and Krutzik, Markus and Willsch, Dennis and others},
  journal={arXiv preprint arXiv:2603.06644},
  year={2026}
}

@article{harrow2009quantum,
  title = {Quantum Algorithm for Linear Systems of Equations},
  author = {Harrow, Aram W. and Hassidim, Avinatan and Lloyd, Seth},
  journal = {Phys. Rev. Lett.},
  volume = {103},
  issue = {15},
  pages = {150502},
  numpages = {4},
  year = {2009},
  month = {Oct},
  publisher = {American Physical Society},
  doi = {10.1103/PhysRevLett.103.150502},
  url = {https://link.aps.org/doi/10.1103/PhysRevLett.103.150502}
}

@article{childs2017quantum,
  title={Quantum algorithm for systems of linear equations with exponentially improved dependence on precision},
  author={Childs, Andrew M and Kothari, Robin and Somma, Rolando D},
  journal={SIAM Journal on Computing},
  volume={46},
  number={6},
  pages={1920--1950},
  year={2017},
  publisher={SIAM}
}

@article{bravo2023variational,
  title={Variational quantum linear solver},
  author={Bravo-Prieto, Carlos and LaRose, Ryan and Cerezo, Marco and Subasi, Yigit and Cincio, Lukasz and Coles, Patrick J},
  journal={Quantum},
  volume={7},
  pages={1188},
  year={2023},
  publisher={Verein zur F{\"o}rderung des Open Access Publizierens in den Quantenwissenschaften}
}

@article{rogers2020floating,
  title={Floating-point calculations on a quantum annealer: Division and matrix inversion},
  author={Rogers, Michael L and Singleton Jr, Robert L},
  journal={Frontiers in Physics},
  volume={8},
  pages={265},
  year={2020},
  publisher={Frontiers Media SA}
}

@article{revach2022kalmannet,
  title={KalmanNet: Neural network aided Kalman filtering for partially known dynamics},
  author={Revach, Guy and Shlezinger, Nir and Ni, Xiaoyong and Escoriza, Adria Lopez and Van Sloun, Ruud JG and Eldar, Yonina C},
  journal={IEEE Transactions on Signal Processing},
  volume={70},
  pages={1532--1547},
  year={2022},
  publisher={IEEE}
}

@article{shi2024quantum,
  title={A quantum algorithm for the Kalman filter using block encoding},
  author={Shi, Hao and Zhang, Guofeng and Zhang, Ming},
  journal={arXiv preprint arXiv:2404.04554} ,
  year={2024}
}

@article{zhao2019dse,
  author  = {Zhao, Junbo and G{\'o}mez-Exp{\'o}sito, Antonio and Netto, Marcos and Mili, Lamine and Abur, Ali and Terzija, Vladimir and Kamwa, Innocent and Pal, Bikash Chandra and Singh, Abhinav Kumar and Qi, Junjian and Huang, Zhenyu and Meliopoulos, A. P. Sakis},
  title   = {Power System Dynamic State Estimation: Motivations, Definitions, Methodologies, and Future Work},
  journal = {IEEE Transactions on Power Systems},
  volume  = {34},
  number  = {4},
  pages   = {3188--3198},
  year    = {2019},
  doi     = {10.1109/TPWRS.2019.2894769}
}

@article{Freitag2020,
	author = {Freitag, Melina A.},
	title = {{Numerical linear algebra in data assimilation}},
	journal = {GAMM{-}Mitteilungen.},
	volume = {43},
	number = {3},
	pages = {e202000014},
	year = {2020},
	month = sep,
	issn = {0936-7195},
	publisher = {John Wiley {\&} Sons, Ltd},
	doi = {10.1002/gamm.202000014}
}

@article{LeProvost2022LowRankEnKF,
	author = {Le Provost, Mathieu and Baptista, Ricardo and Marzouk, Youssef and Eldredge, Jeff D.},
	title = {{A low-rank ensemble Kalman filter for elliptic observations}},
	journal = {Proc. A},
	volume = {478},
	number = {2266},
	year = {2022},
	month = oct,
	issn = {1364-5021},
	publisher = {The Royal Society},
	doi = {10.1098/rspa.2022.0182}
}

@InProceedings{pfortner2025compaware,
  title = 	 {Computation-Aware Kalman Filtering and Smoothing},
  author =       {Pf{\"o}rtner, Marvin and Wenger, Jonathan and Cockayne, Jon and Hennig, Philipp},
  booktitle = 	 {Proceedings of The 28th International Conference on Artificial Intelligence and Statistics},
  pages = 	 {2071--2079},
  year = 	 {2025},
  editor = 	 {Li, Yingzhen and Mandt, Stephan and Agrawal, Shipra and Khan, Emtiyaz},
  volume = 	 {258},
  series = 	 {Proceedings of Machine Learning Research},
  month = 	 {03--05 May},
  publisher =    {PMLR},
  url = 	 {https://proceedings.mlr.press/v258/pfortner25a.html}
}

@ARTICLE{Shlezinger2025AIAidedKF,
  author={Shlezinger, Nir and Revach, Guy and Ghosh, Anubhab and Chatterjee, Saikat and Tang, Shuo and Imbiriba, Tales and Dunik, Jindrich and Straka, Ondrej and Closas, Pau and Eldar, Yonina C.},
  journal={IEEE Signal Processing Magazine}, 
  title={Artificial Intelligence-Aided Kalman Filters: AI-Augmented Designs for Kalman-Type Algorithms}, 
  year={2025},
  volume={42},
  number={3},
  pages={52-76},
  doi={10.1109/MSP.2025.3569395}}

@misc{mortada2025recursivekalmannet,
  author       = {Mortada, Hassan and Falcon, Cyril and Kahil, Yanis and Clavaud, Math{\'e}o and Michel, Jean-Philippe},
  title        = {Recursive KalmanNet: Deep Learning-Augmented Kalman Filtering for State Estimation with Consistent Uncertainty Quantification},
  year         = {2025},
  eprint       = {2506.11639},
  archivePrefix= {arXiv},
  primaryClass = {cs.LG},
  url          = {https://arxiv.org/abs/2506.11639}
}

@INPROCEEDINGS{Stoyanova2025Quantum,
  author={Stoyanova, Ivelina and Şensebat, Orkun and Ji, Yanjun and Ganeshamurthy, Priyanka Arkalgud and Kajganic, Sonja and Willsch, Dennis and Monti, Antonello},
  booktitle={2025 IEEE PES Innovative Smart Grid Technologies Conference Europe (ISGT Europe)}, 
  title={Quantum Computing Methods for Dynamic State Estimation in Power Systems}, 
  year={2025},
  volume={},
  number={},
  pages={1-5},
  doi={10.1109/ISGTEurope64741.2025.11305506}}

@book{anderson1979optimal,
   title = "Optimal Filtering",
    author = {Anderson, Brian D. O. and Moore, John B.},
    year = "1979",
    isbn = "0-13-638122-7",
    publisher = "Prentice-Hall"
}

@book{kailath2000linear,
  title={Linear estimation},
  author={Kailath, Thomas and Sayed, Ali H and Hassibi, Babak},
  year={2000},
  publisher={Prentice Hall}
}

@article{hestenes1952methods,
  title={Methods of conjugate gradients for solving linear systems},
  author={Hestenes, Magnus R and Stiefel, Eduard},
  journal={Journal of research of the National Bureau of Standards},
  volume={49},
  number={6},
  pages={409--436},
  year={1952}
}

@book{greenbaum1997iterative,
  title={Iterative methods for solving linear systems},
  author={Greenbaum, Anne},
  year={1997},
  publisher={SIAM}
}

@book{saad2003iterative,
  title={Iterative methods for sparse linear systems},
  author={Saad, Yousef},
  year={2003},
  publisher={SIAM}
}

@article{pandapower,
  author={Thurner, Leon and Scheidler, Alexander and Schäfer, Florian and Menke, Jan-Hendrik and Dollichon, Julian and Meier, Friederike and Meinecke, Steffen and Braun, Martin},
  journal={IEEE Transactions on Power Systems}, 
  title={Pandapower—An Open-Source Python Tool for Convenient Modeling, Analysis, and Optimization of Electric Power Systems}, 
  year={2018},
  volume={33},
  number={6},
  pages={6510-6521},
  doi={10.1109/TPWRS.2018.2829021}}

@Article{JURECA,
  author           = {{J\"{u}lich Supercomputing Centre}},
  title            = {{JURECA: Data Centric and Booster Modules implementing the Modular Supercomputing Architecture at J\"{u}lich Supercomputing Centre}},
  doi              = {10.17815/jlsrf-7-182},
  number           = {A182},
  url              = {http://dx.doi.org/10.17815/jlsrf-7-182},
  volume           = {7},
  journal          = {J. of Large-Scale Res. Facil.},
  year             = {2021},
}

@article{li2025policy,
title = {Policy optimization of finite-horizon Kalman filter with unknown noise covariance},
journal = {Automatica},
volume = {177},
pages = {112320},
year = {2025},
issn = {0005-1098},
doi = {10.1016/j.automatica.2025.112320},
url = {https://www.sciencedirect.com/science/article/pii/S0005109825002134},
author = {Haoran Li and Yuan-Hua Ni}
}

@article{sun1998sensitivity,
title = {Sensitivity analysis of the discrete-time algebraic Riccati equation},
journal = {Linear Algebra and its Applications},
volume = {275-276},
pages = {595-615},
year = {1998},
note = {Proceedings of the Sixth Conference of the International Linear Algebra Society},
issn = {0024-3795},
doi = {10.1016/S0024-3795(97)10017-9},
url = {https://www.sciencedirect.com/science/article/pii/S0024379597100179},
author = {Ji-guang Sun}
}

@article{sun1998perturbation,
author = {Sun, Ji-guang},
title = {Perturbation Theory for Algebraic Riccati Equations},
journal = {SIAM Journal on Matrix Analysis and Applications},
volume = {19},
number = {1},
pages = {39-65},
year = {1998},
doi = {10.1137/S0895479895291303},
URL = {https://doi.org/10.1137/S0895479895291303}
}

@article{aalto2018spatial,
    author = {Aalto, Atte},
    title = {Spatial discretization error in Kalman filtering for discrete-time infinite dimensional systems},
    journal = {IMA Journal of Mathematical Control and Information},
    volume = {35},
    number = {Supplement_1},
    pages = {i51-i72},
    year = {2018},
    month = {04},
    issn = {0265-0754},
    doi = {10.1093/imamci/dnx015},
    url = {https://doi.org/10.1093/imamci/dnx015}
}

@article{BardsleyEtAl2013,
author = {Bardsley, Johnathan M. and Parker, Albert and Solonen, Antti and Howard, Marylesa},
title = {Krylov space approximate Kalman filtering},
journal = {Numerical Linear Algebra with Applications},
volume = {20},
number = {2},
pages = {171-184},
doi = {10.1002/nla.805},
url = {https://onlinelibrary.wiley.com/doi/abs/10.1002/nla.805},
year = {2013}
}

@ARTICLE{verhaegen1986numerical,
  author={Verhaegen, M. and Van Dooren, P.},
  journal={IEEE Transactions on Automatic Control}, 
  title={Numerical aspects of different Kalman filter implementations}, 
  year={1986},
  volume={31},
  number={10},
  pages={907-917},
  doi={10.1109/TAC.1986.1104128}}

@article{simoncini2003inexact,
author = {Simoncini, Valeria and Szyld, Daniel B.},
title = {Theory of Inexact Krylov Subspace Methods and Applications to Scientific Computing},
journal = {SIAM Journal on Scientific Computing},
volume = {25},
number = {2},
pages = {454-477},
year = {2003},
doi = {10.1137/S1064827502406415},
URL = {https://doi.org/10.1137/S1064827502406415}
}

@article{meidner2009goal,
url = {https://doi.org/10.1515/JNUM.2009.009},
title = {Goal-oriented error control of the iterative solution of finite element equations},
author = {D. Meidner and R. Rannacher and J. Vihharev},
pages = {143--172},
volume = {17},
number = {2},
journal = {Journal of Numerical Mathematics},
doi = {doi:10.1515/JNUM.2009.009},
year = {2009}
}

@article{endtmayer2020twoside,
author = {Endtmayer, B. and Langer, U. and Wick, T.},
title = {Two-Side a Posteriori Error Estimates for the Dual-Weighted Residual Method},
journal = {SIAM Journal on Scientific Computing},
volume = {42},
number = {1},
pages = {A371-A394},
year = {2020},
doi = {10.1137/18M1227275},
URL = {https://doi.org/10.1137/18M1227275}
}

@article{cockayne2019bayescg,
  title={A Bayesian conjugate gradient method (with discussion)},
  author={Cockayne, Jon and Oates, Chris J and Ipsen, Ilse CF and Girolami, Mark},
  year={2019},
journal = {Bayesian Analysis},
volume  = {14},
  number  = {3},
  pages   = {937--1012},
  doi     = {10.1214/19-BA1145},
  url     = {https://doi.org/10.1214/19-BA1145}
}

@article{mahmoud2004resilient,
	author = {Mahmoud, Magdi S.},
	title = {{Resilient linear filtering of uncertain systems}},
	journal = {Automatica},
	volume = {40},
	number = {10},
	pages = {1797--1802},
	year = {2004},
	month = oct,
	issn = {0005-1098},
	publisher = {Pergamon},
	doi = {10.1016/j.automatica.2004.05.007}
}

@article{deoliveira2005implementation,
	author = {de Oliveira, M. C. and Geromel, J. C.},
	title = {{H2 and H{$\infty$} Filtering Design Subject to Implementation Uncertainty}},
	journal = {IFAC Proceedings Volumes},
	volume = {36},
	number = {11},
	pages = {121--126},
	year = {2003},
	month = jun,
	issn = {1474-6670},
	publisher = {Elsevier},
	doi = {10.1016/S1474-6670(17)35650-1}
}

@article{fan2013ekfdse,
	author = {Fan, Lingling and Wehbe, Yasser},
	title = {{Extended Kalman filtering based real-time dynamic state and parameter estimation using PMU data}},
	journal = {Electr. Power Syst. Res.},
	volume = {103},
	pages = {168--177},
	year = {2013},
	month = oct,
	issn = {0378-7796},
	publisher = {Elsevier},
	doi = {10.1016/j.epsr.2013.05.016}
}

@ARTICLE{wang2018pmujacobian,
  author={Wang, Xiaozhe and Bialek, Janusz W. and Turitsyn, Konstantin},
  journal={IEEE Transactions on Power Systems}, 
  title={PMU-Based Estimation of Dynamic State Jacobian Matrix and Dynamic System State Matrix in Ambient Conditions}, 
  year={2018},
  volume={33},
  number={1},
  pages={681-690},
  doi={10.1109/TPWRS.2017.2712762}}

@INPROCEEDINGS{akhlaghi2017adaptive,
  author={Akhlaghi, Shahrokh and Zhou, Ning and Huang, Zhenyu},
  booktitle={2017 IEEE Power \& Energy Society General Meeting}, 
  title={Adaptive adjustment of noise covariance in Kalman filter for dynamic state estimation}, 
  year={2017},
  volume={},
  number={},
  pages={1-5},
  doi={10.1109/PESGM.2017.8273755}}

@article{poyneer2023lqg,
	author = {Poyneer, Lisa A. and Ammons, S. Mark and Kim, Mike K. and Bauman, Brian and Terrel-Perez, Jesse and Lemmer, Aaron J. and Nguyen, Jayke and Nguyen, Jayke},
	title = {{Laboratory demonstration of the prediction of wind-blown turbulence by adaptive optics at 8{\hspace{0.167em}}{\hspace{0.167em}}kHz with use of LQG control}},
	journal = {Appl. Opt.},
	volume = {62},
	number = {8},
	pages = {1871--1885},
	year = {2023},
	month = mar,
	issn = {2155-3165},
	publisher = {Optica Publishing Group},
	doi = {10.1364/AO.474730}
}

@article{setter2018realtime,
  title = {Real-time Kalman filter: Cooling of an optically levitated nanoparticle},
  author = {Setter, Ashley and Toro\ifmmode \check{s}\else \v{s}\fi{}, Marko and Ralph, Jason F. and Ulbricht, Hendrik},
  journal = {Phys. Rev. A},
  volume = {97},
  issue = {3},
  pages = {033822},
  numpages = {8},
  year = {2018},
  month = {Mar},
  publisher = {American Physical Society},
  doi = {10.1103/PhysRevA.97.033822},
  url = {https://link.aps.org/doi/10.1103/PhysRevA.97.033822}
}

@article{magrini2021realtime,
	author = {Magrini, Lorenzo and Rosenzweig, Philipp and Bach, Constanze and Deutschmann-Olek, Andreas and Hofer, Sebastian G. and Hong, Sungkun and Kiesel, Nikolai and Kugi, Andreas and Aspelmeyer, Markus},
	title = {{Real-time optimal quantum control of mechanical motion at room temperature}},
	journal = {Nature},
	volume = {595},
	pages = {373--377},
	year = {2021},
	month = jul,
	issn = {1476-4687},
	publisher = {Nature Publishing Group},
	doi = {10.1038/s41586-021-03602-3}
}

@misc{fzj2026jion,
  author       = {{Forschungszentrum J{\"u}lich}},
  title        = {Official Launch of {JION} Quantum Computer at J{\"u}lich},
  year         = {2026},
  month        = sep,
  day          = {3},
  howpublished = {Press release},
  url          = {https://www.fz-juelich.de/en/news/archive/press-release/2026/launch-jion-quantum-computer-juelich},
  note         = {Accessed 2026-09-20}
}

@misc{hsieh2019learning,
      title={Learning Neural PDE Solvers with Convergence Guarantees}, 
      author={Jun-Ting Hsieh and Shengjia Zhao and Stephan Eismann and Lucia Mirabella and Stefano Ermon},
      year={2019},
      eprint={1906.01200},
      archivePrefix={arXiv},
      primaryClass={math.NA},
      url={https://arxiv.org/abs/1906.01200}, 
}

@inproceedings{kaneda2023deep,
author = {Kaneda, Ayano and Akar, Osman and Chen, Jingyu and Kala, Victoria Alicia Trevino and Hyde, David and Teran, Joseph},
title = {A deep conjugate direction method for iteratively solving linear systems},
year = {2023},
publisher = {JMLR.org},
booktitle = {Proceedings of the 40th International Conference on Machine Learning},
articleno = {643},
numpages = {17},
location = {Honolulu, Hawaii, USA},
series = {ICML'23}
}

@incollection{powell1994cobyla,
	author = {Powell, M. J. D.},
	title = {{A Direct Search Optimization Method That Models the Objective and Constraint Functions by Linear Interpolation}},
	booktitle = {{Advances in Optimization and Numerical Analysis}},
	journal = {SpringerLink},
	pages = {51--67},
	year = {1994},
	isbn = {978-94-015-8330-5},
	publisher = {Springer},
	address = {Dordrecht, The Netherlands},
	doi = {10.1007/978-94-015-8330-5_4}
}

@techreport{mcgeoch2022advantage2,
  author      = {McGeoch, Catherine and Farr{\'e}, Pau and Boothby, Kelly},
  title       = {The {D-Wave} {Advantage2} Prototype},
  institution = {D-Wave Systems Inc.},
  address     = {Burnaby, BC, Canada},
  number      = {14-1063A-A},
  year        = {2022},
  url         = {https://www.dwavequantum.com/media/eixhdtpa/14-1063a-a_the_d-wave_advantage2_prototype-4.pdf}
}

@misc{originqcloud,
  author       = {{Origin Quantum Computing Technology (Hefei) Co., Ltd.}},
  title        = {Origin Quantum Cloud Platform},
  howpublished = {\url{https://qcloud.originqc.com.cn}},
  note         = {Origin Wukong superconducting processor; accessed September 23, 2026},
  year         = {2026}
}

\end{document}